\documentclass{amsart}

\usepackage{caption}
\usepackage{mathtools}
\usepackage{amssymb}
\usepackage{amsthm}
\usepackage{tensor}
\usepackage[safe]{tipa} 

\newtheorem{proposition}{Proposition}
\newtheorem{lemma}{Lemma}
\newtheorem{theorem}{Theorem}

\usepackage[width=.92\textwidth]{caption}

\usepackage{adjustbox}

\usepackage[
  backend=biber,
  style=alphabetic
]{biblatex}
\usepackage[
  colorlinks=true,
  linkcolor=darkgreen,
  citecolor=darkgreen,
  urlcolor=darkgreen
]{hyperref}

\usepackage{tikz}
\usetikzlibrary{
  cd,
  calc,
  arrows.meta,
  backgrounds,
  decorations,
  decorations.pathmorphing, 
}

\tikzcdset{
  arrow style=tikz, 
  diagrams={
    >={
      Computer Modern Rightarrow[
        length=4pt, width=4pt
      ]
    },
    shorten=0pt,
    hook/.style={
      {Hooks[right, length=2pt, width=8pt]}->
    }
  }
}

\definecolor{darkblue}{rgb}{0.05,0.25,0.65}
\definecolor{darkgreen}{RGB}{20,140,10}
\definecolor{lightgray}{rgb}{0.9,0.9,0.9}
\definecolor{darkorange}{RGB}{200,100,5}
\definecolor{darkyellow}{rgb}{.91,.91,0}
\definecolor{lightolive}{RGB}{225, 220, 185}

\newcommand{\shape}{%
  \hspace{.7pt}%
  \raisebox{0.8pt}{\rm\normalfont\textesh}%
  \hspace{1pt}%
}

\newcommand{\defneq}{\equiv}

\newcommand{\weakHomotopyEquivalence}{\sim}

\newcommand{\hotype}[1]{\mathcal{#1}}

\newcommand{\plus}{\hspace{.8pt}{\adjustbox{scale={.5}{.77}}{$\sqcup$} \{\infty\}}}
\newcommand{\cpt}{\hspace{.8pt}{\adjustbox{scale={.5}{.77}}{$\cup$} \{\infty\}}}

\newcommand{\grayunderbrace}[2]{\mathcolor{gray}{\underbrace{\mathcolor{black}{#1}}}_{\mathcolor{gray}{#2}}}

\newcommand{\grayoverbrace}[2]{\mathcolor{gray}{\overbrace{\mathcolor{black}{#1}}}^{\mathcolor{gray}{#2}}}

\newcommand{\oppositeinterval}[3]{
\draw[gray, line width=3.3]
  (#1+.2,#2) -- (#1+#3-.2,#2);
\draw[draw=gray,line width=1.3,fill=black]
  (#1,#2) circle (.2);
\draw[draw=gray,line width=1.3,fill opacity=.6,fill=white]
  (#1+#3,#2) circle (.2);
}

\newcommand{\openinterval}[3]{
\draw[gray, line width=3.3]
  (#1+.2,#2) -- (#1+#3-.2,#2);
\draw[draw=gray,line width=1.3,fill opacity=.6,fill=white]
  (#1+#3,#2) circle (.2);
\draw[draw=gray,line width=1.3,fill opacity=.6,fill=white]
  (#1,#2) circle (.2);
}

\newcommand{\closedinterval}[3]{
\draw[gray, line width=3.3]
  (#1+.2,#2) -- (#1+#3-.2,#2);
\draw[draw=gray,line width=1.3,fill=black]
  (#1+#3,#2) circle (.2);
\draw[draw=gray,line width=1.3,fill=black]
  (#1,#2) circle (.2);
}

\newcommand{\halfcircle}{
\begin{scope}
\clip
  (-4,0) rectangle (4.1,-4.1);
\draw[
  gray,
  line width=30pt,
  draw opacity=.4
]
  (0,0) circle (3);
\end{scope}

\draw[draw=gray,line width=1.3,fill=black]
  (0.1,-3.5) circle (.2);
\draw[draw=gray,line width=1.3,fill opacity=.6,fill=white]
  (-0.1,-3.5) circle (.2);

\draw[gray, line width=3.3]
  (-1.7+.2,-3.1) -- (1.7-.2,-3.1);
\draw[draw=gray,line width=1.3,fill=black]
  (1.7,-3.1) circle (.2);
\draw[draw=gray,line width=1.3,fill opacity=.6,fill=white]
  (-1.7,-3.1) circle (.2);

\closedinterval{.1}{-2.5}{2.35};
\openinterval{-2.5}{-2.5}{2.4};

\closedinterval{1.7}{-1.83}{1.34};
\openinterval{-3}{-1.83}{1.34};

\closedinterval{2.2}{-1.12}{1.12};
\openinterval{-3.31}{-1.12}{1.12};

\closedinterval{2.45}{-.38}{1.05};
\openinterval{-3.5}{-.38}{1.05};
}

\newcommand{\halfcircleAdjusted}{
\begin{scope}
\clip
  (-4,0) rectangle (4.1,-4.1);
\draw[
  gray,
  line width=30pt,
  draw opacity=.4
]
  (0,0) circle (3);
\end{scope}

\draw[draw=gray,line width=1.3,fill=black]
  (0.1,-3.5) circle (.2);
\draw[draw=gray,line width=1.3,fill opacity=.6,fill=white]
  (-0.1,-3.5) circle (.2);
  
\draw[gray, line width=3.3]
  (-1.8+.2,-3.0) -- 
  (+1.8-.2,-3.0);
\draw[draw=gray,line width=1.3,fill=black]
  (1.8,-3) circle (.2);
\draw[draw=gray,line width=1.3,fill opacity=.6,fill=white]
  (-1.8,-3) circle (.2);

\closedinterval{.1}{-2.5}{2.35};
\openinterval{-2.5}{-2.5}{2.4};

\closedinterval{1.7}{-1.83}{1.34};
\openinterval{-3}{-1.83}{1.34};

\closedinterval{2.2}{-1.12}{1.12};
\openinterval{-3.31}{-1.12}{1.12};

\closedinterval{2.45}{-.38}{1.05};
\openinterval{-3.5}{-.38}{1.05};
}

\newcommand{\halfcircleover}[2]{
\begin{scope}
\clip
  (-4,0) rectangle (4.1,-4.1);
\draw[
  white,
  line width=43pt,
]
  (0,0) circle (3);
\end{scope}
  \halfcircle{#1}{#1}
}

\newcommand{\halfcircleoverAdjusted}{
\begin{scope}
\clip
  (-4,0) rectangle (4.1,-4.1);
\draw[
  white,
  line width=43pt,
]
  (0,0) circle (3);
\end{scope}
  \halfcircleAdjusted
}

\begin{document}

\title
[Renormalization of Wilson Loops via exotic Flux Quantization]
{Renormalization of Chern-Simons Wilson Loops 
via Proper Flux Quantization in Cohomotopy}

\author{Hisham Sati}
\address{Mathematics Program and Center for Quantum and Topological Systems, New York University Abu Dhabi}
\curraddr{}
\email{hsati@nyu.edu}
\thanks{}

\author{Urs Schreiber           }
\address{Mathematics Program and Center for Quantum and Topological Systems, New York University Abu Dhabi}
\curraddr{}
\email{us13@nyu.edu}
\thanks{}



\date{\today}

\dedicatory{
  \href{https://ncatlab.org/nlab/show/Center+for+Quantum+and+Topological+Systems}{\includegraphics[width=3.1cm]{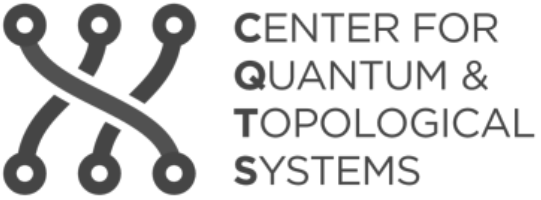}}
}

\begin{abstract}
In the practice of physics model building, the process of renormalization, resummation, and anomaly cancellation is to incrementally repair initially ill-defined Lagrangian quantum field theories by a successive choice of partial fixes. Impressive as this is, one would rather have concisely defined complete theories to begin with, and understand these choices as emergent from fundamental principles. As an instructive example, we recall renormalization choices for Wilson loop observables in abelian Chern-Simons theory. Then we show that these emerge in a novel non-Lagrangian topological completion of 5D Maxwell-Chern-Simons QFT, by means of proper flux quantization in 2-Cohomotopy. This result is a modest cousin, with applications to topologically ordered quantum materials, of the more ambitious completion of 11D supergravity by flux quantization in 4-Cohomotopy (``Hypothesis H'').
\end{abstract}

\maketitle
  
\tableofcontents


\section{Motivation and Introduction}

The mathematical status of Quantum Field Theory (QFT) remains a mixed bag \cite{Cao1999, SS11-Foundations}: While plausible axiomatics exist (in the form of \emph{Algebraic QFT}
and \emph{Functorial QFT}, cf. \cite{Schreiber2009}, although both still requiring and seeing further refinement),
the construction of non-toy examples beyond perturbation theory remains largely elusive and has been pronounced a ``Millennium Problem'' by the Clay Math Institute 
\cite{JaffeWitten2006, Chatterjee2019}. 
Even perturbatively --- where one trades tractability at the staggering cost of non-convergence of observables --- the construction of examples by perturbative \emph{quantization} of classical \emph{Lagrangian} field theories is an incremental process of \emph{anomaly cancellation} followed by \emph{renormalization} followed by \emph{resummation} (cf. \cite{Rejzner2016, Schreiber2017, Duetsch2019}), terminology reflecting the historical exasperation that Lagrangian field quantization requires adding ever more ingredients as one proceeds: the Lagrangian density that starts this process is akin to just the stone in the proverbial \emph{stone soup} \cite{Ashliman-StoneSoup} (cf. top of Fig. \ref{FromTheoryToObservables})!

\begin{figure}[htb]
\caption{
  \label{FromTheoryToObservables}
  While it is common to speak as if Lagrangian densities define quantum field theories, a whole sequence of further input data generally needs to be provided in an incremental and typically \emph{ad hoc} manner. Eventually one would hope to concisely define complete theories right away.
}

\centering

\adjustbox{
  rndfbox=5pt
}{
\begin{tikzcd}[
  column sep=35pt,
  row sep=-10pt
]
  \substack{
    \scalebox{.7}{
      \bf \color{darkblue}
      Lagrangian
    }
    \\
    \scalebox{.7}{
      \bf \color{darkblue}
      density
    }
  }
  \ar[
    rr,
    -Latex,
    dotted,
    "{
      \scalebox{.7}{
        \bf \color{darkgreen}
        anomaly
      }
    }",
    "{
      \scalebox{.7}{
        \bf \color{darkgreen}
        cancellation
      }
    }"{swap}
  ]
  &&
  {}
  \ar[
    rr,
    -Latex,
    dotted,
    "{
      \scalebox{.7}{
        \bf \color{darkgreen}
        renormalization
      }
    }"{pos=.45}
  ]
  &&
  {}
  \ar[
    rr,
    -Latex,
    dotted,
    "{
      \scalebox{.7}{
        \bf \color{darkgreen}
        resummation
      }
    }"
  ]
  &&
  \phantom{  
    \substack{
    \scalebox{.7}{
      \bf \color{darkblue}
      Quantum
    }
    \\
    \scalebox{.7}{
      \bf \color{darkblue}
      observables
    }
  }
  }
  \\
  &&&&&&
  \substack{
    \scalebox{.7}{
      \bf \color{darkblue}
      Quantum
    }
    \\
    \scalebox{.7}{
      \bf \color{darkblue}
      observables
    }
  }
  \\
  \substack{
    \scalebox{.7}{
      \bf \color{darkblue}
      Complete
    }
    \\
    \scalebox{.7}{
      \bf \color{darkblue}
      theory
    }
  }
  \ar[
    rrrrrr,
    |->,
    "{
      \scalebox{.7}{
        \bf \color{darkgreen}
        implication
      }
    }"{description}
  ]
  &&&&&&
  \phantom{
  \substack{
    \scalebox{.7}{
      \bf \color{darkblue}
      Quantum
    }
    \\
    \scalebox{.7}{
      \bf \color{darkblue}
      observables
    }
  }  
  }
\end{tikzcd}
}
  
\end{figure}
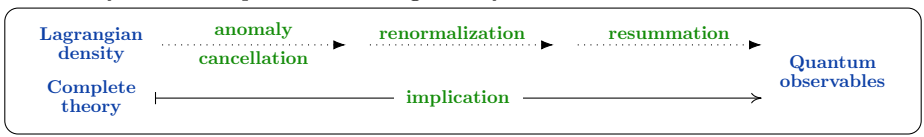

Beyond perturbation theory, even basic definitions (such as of Wilson loop observables, reviewed in \S\ref{TraditionalCSWilsonLoops}), are traditionally tied to the mere heuristics of the \emph{path integral} \cite{nLab:PathIntegral}, which is the enshrined hope that Lagrangian densities define quantum field theories by themselves. 

While it is common to envision that ultimately there may be \emph{complete} quantum theories (``UV-completion'') from which quantum observables would be systematically \emph{derivable} (bottom of Fig. \ref{FromTheoryToObservables}), little technical attention has been devoted to the question of which guise these would concretely take (resulting in discussions like \cite{Crowther2019}).

We highlight that one pronounced way in which Lagrangian densities manifestly fail to encode complete field theories, already classically, is --- in the paramount case of (higher, cf.  \cite{Alfonsi2025HigherGeometry, SatiSchreiberHigherGauge}) \emph{gauge theories} \cite{Henneaux1992} --- their sole reflection of local \emph{gauge potentials} on a single spacetime chart, insensitive to the remaining global topological content of the (higher) gauge field -- cf. Fig. \ref{CechDataForHigherGaugeField}. This global topological field content is instead embodied by the gauge field's \emph{flux quantization law} \cite{SS25-Flux}, traditionally an afterthought to the Lagrangian density imposed in the process of its ``anomaly cancellation''.

\begin{figure}[htb]
\caption{
  \label{CechDataForHigherGaugeField}
  One way in which Lagrangian densities for (higher) gauge fields \cite{Henneaux1992} are manifestly incomplete as an encoding of a field theory is in their sole dependence on local gauge potentials, insensitive to the global topological field content (the \emph{flux quantization law} \cite{SS25-Flux}): The latter is equivalently encoded in {\v C}ech cocycle data (with coefficients in some moduli stack, cf. \cite{SS24-Phase}\cite[\S 3.3]{SS25-Flux}) with respect to an open cover $Y$ of spacetime $X$ (by ``charts''). This involves field data on higher fiber products $Y^{\times_X^n}$,  while the gauge potentials are only the lowest stage data on $Y$ itself.
}
\adjustbox{
  rndfbox=5pt,
  scale=.96
}{
$
\begin{array}{l}
\\[-6pt]
\hspace{-5pt}
\begin{tikzcd}[
  column sep=9pt,
  row sep=-6pt,
  ampersand replacement=\&
]
  \mathclap{\adjustbox{
    rotate=25,
    scale=.7,
    color=olive
  }{
    \rlap{spacetime}
  }}
  \&[10pt]
  \mathclap{\adjustbox{
    rotate=25,
    scale=.7,
    color=darkblue
  }{
    \rlap{cover}
  }}
  \&\&[-5pt]
  \mathclap{\adjustbox{
    rotate=25,
    scale=.7,
    color=darkgreen
  }{
    \rlap{
      double...
    }
  }}
  \&\&[+5pt]
  \mathclap{\adjustbox{
    rotate=25,
    scale=.7,
    color=darkorange
  }{
    \rlap{
      triple...
    }
  }}
  \&\&[+20pt]
  \mathclap{\adjustbox{
    rotate=0,
    scale=.7,
    color=darkorange,
    raise=4pt
  }{
    and higher intersections
  }}
  \&\&[-10pt]
  \mathclap{\adjustbox{
    rotate=0,
    scale=.7,
    raise=4pt
  }{
    ...
  }}
  \\
  X
  \ar[
    r,
    <<-,
    "{ p }"{pos=.7}
  ]
  \&
  Y
  \ar[
    rr,
    <-,
    shift left=3pt,
    shorten=-1pt,
    "{ p_0 }"{pos=.65}
  ]
  \ar[
    rr,
    <-,
    shift right=3pt,
    shorten=-1pt,
    "{ p_1 }"{swap, pos=.65}
  ]
  \&\&
  Y \times_X Y
  \ar[
    rr,
    <-,
    shift left=6pt,
    shorten=-1pt,
    "{ p_0 }"{pos=.65}
  ]
  \ar[
    rr,
    <-,
    shorten=-1pt,
    "{ p_1 }"{description,pos=.65}
  ]
  \ar[
    rr,
    <-,
    shift right=6pt,
    shorten=-1pt,
    "{ p_2 }"{swap, pos=.65}
  ]
  \&\&
  Y 
    \times_X 
  Y
    \times_X 
  Y
  \ar[
    rr,
    <-,
    shift left=9pt,
    shorten=-1pt,
    "{ p_0 }"{pos=.65}
  ]
  \ar[
    rr,
    <-,
    shift left=3.5pt,
    shorten=-1pt,
    "{ p_1 }"{description, pos=.65}
  ]
  \ar[
    rr,
    <-,
    shift right=3.5pt,
    shorten=-1pt,
    "{ p_2 }"{description,pos=.65}
  ]
  \ar[
    rr,
    <-,
    shift right=9pt,
    shorten=-1pt,
    "{ p_3 }"{swap, pos=.65}
  ]
  \&\&
  Y 
    \times_X 
  Y
    \times_X 
  Y
    \times_X 
  Y
  \&\&
  \cdots
\end{tikzcd}
\\
\hspace{-9pt}
\begin{tikzpicture} 
\node at (0,0) {
$
  \underset{
    \mathclap{
      \;\;\;\;
      \adjustbox{
        scale=.7,
        rotate=-60
      }{
        \color{olive}
        \bf
        \rlap{
          \hspace{-10pt}
          gauge field
        }
      }
    }  
  }{
    \widehat{A}
  }
  \;=\,
  \left(
  \hspace{-3pt}
  \def\arraycolsep{0pt}
  \begin{array}{ccccc}
  \underset{
    \mathclap{
      \;\;\;\;
      \adjustbox{
        scale=.7,
        rotate=-50
      }{
        \color{darkblue}
        \bf
        \rlap{
          \hspace{-20pt}
          gauge potentials
        }
      }
    }
  }{
  \begin{tikzcd}
    A
    \,,
  \end{tikzcd}
  }
  &
  \begin{tikzcd}[
    column sep=10pt
  ]
    p_0^\ast A
    \ar[
      r,
      "{ 
  \underset{
    \mathclap{
      \adjustbox{
        scale=.7,
        rotate=-50
      }{
        \color{darkgreen}
        \bf
        \rlap{
          \hspace{-0pt}
          transition data
        }
      }
    }
  }{
        g 
  }
      }"{yshift=-10pt}
    ]
    &
    p_1^\ast A
    \,,
  \end{tikzcd}
  &
  \begin{tikzcd}[
    row sep=30pt,
    column sep=10pt
  ]
    & 
    p_1^\ast A
    \ar[
      dr,
      "{
        p_{12}^\ast g
      }"{sloped}
    ]
    \ar[
      d,
      Rightarrow,
      shorten=5pt,
      "{
        \eta
      }"
    ]
    \\
    p_0^\ast A
    \ar[
      ur,
      "{
        p_{01}^\ast g
      }"{sloped}
    ]
    \ar[
      rr,
      "{ 
        p_{02}^\ast g 
      }"{swap},
      "{
        \smash{
          \mathrlap{
           \hspace{-15pt}
           \scalebox{.7}{
             \color{darkorange}
             \bf
             higher transition data...
           }
          }
        }
      }"{swap, yshift=-18pt}
    ]
    &{}&
    p_2^\ast A
    \,,
  \end{tikzcd}  
  &
  \begin{tikzcd}
    &[-10pt]&
    &[-50pt] 
    |[alias=two]|
    p_1^\ast A
    \ar[
      ddr,
      "{
        p_{12}^\ast g
      }"{sloped},
      "{\ }"{name=twofour, swap}
    ]
    &[-5pt]
    \\
    \\
    |[alias=one]|
    p_0^\ast A
    \ar[
      rrrr,
      "{\ }"{name=onefour}
    ]
    \ar[
      drr,
      "{ p_{02}^\ast g }"{sloped, swap},
      "{\  }"{name=onethree}
    ]
    \ar[
      uurrr,
      "{
        p_{01}^\ast g
      }"{sloped}
    ]
    &&&&
    |[alias=four]|
    p_3^\ast A
    \mathrlap{\,,}
    \\[-3pt]
    &&
    |[alias=three]|
    p_2^\ast A
    \ar[
      urr,
      "{
        p_{23}^\ast g
      }"{sloped, swap}
      "{\ }"{name=one}
    ]
    \ar[
      from=two, 
      to=onefour,
      Rightarrow, 
      crossing over,
      shorten=5pt
    ]
    \ar[
      from=two, 
      to=onethree,
      Rightarrow, 
      shorten <=10pt,
      shorten >=1pt,
      crossing over,
      "{
        p_{012}^\ast \eta
      }"{sloped, yshift=2pt, pos=.55}
    ]
    \ar[
      from=three, 
      to=onefour,
      Rightarrow, 
      crossing over,
      shorten=5pt
    ]
    \ar[
      from=two,
      to=three,
      crossing over,
      "{\ }"{name=twothree}
    ]
    \ar[
      from=three, 
      to=twofour,
      shorten <= 10pt,
      shorten >= 2pt,
      Rightarrow, 
      crossing over,
    ]
  \end{tikzcd}
  &
  \cdots
  \end{array}
  \right)
  \hspace{-10pt}
$};

\begin{scope}[
  shift={(1.15,0)}
]

\node[
  rotate=-5,
  scale=.7,
  black!80
] 
  at (-6.1,1.3)
  {
    \adjustbox{
      margin=2pt,
      bgcolor=white
    }{
    \def\arraystretch{.9}%
    \def\tabcolsep{0pt}%
    \color{darkblue}
    \begin{tabular}{c}
      Lagrangians see
      \\
      {\it only} this piece
    \end{tabular}
    }
  };

 \draw[
   -Latex,
   black!40
 ]
   (-6.1,1) .. controls
   (-6.4,.4) and
   (-5.7,1.1) ..
   (-6,.3);
\end{scope}

\end{tikzpicture}
\end{array}
\hspace{-4pt}
$
}

\end{figure}
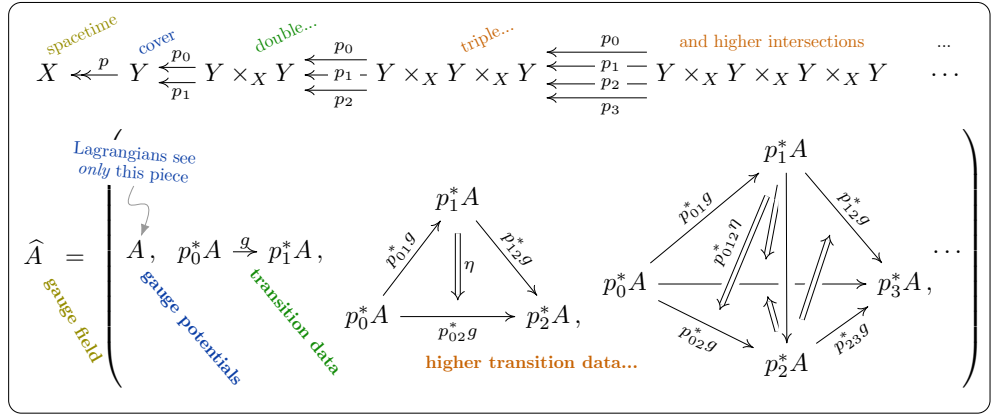

Of course, \emph{ordinary Dirac flux quantization} is a time-honored topic in itself \cite{Alvarez1985}\cite[\S 16.4e]{Frankel2011}\cite[Ex. 3.9-10]{SS25-Flux}. However, even aside from the issue of its 
\mbox{(non-)}reflection in Lagrangian densities, it never accounts for the generic non-linearity of (higher, Maxwell-type) gauge fields, whose Gauss laws are polynomials $P^i\Big(\vec F :=\big\{F^{(j)}_{\mathrm{deg}(g)} \big\}_{j \in I} \Big)$ in a set $I$ of flux species:
\begin{equation}
  \label{HigherMaxwellTypeEquationsOfMotion}
  \forall_{i \in I}
  \;\;\;\;
  \mathrm{d}
  \,
  F^{(i)}_{\mathrm{deg}_i}
  =
  P^i\big(
    \vec F
  \,\big)
  \,,
  \;\;\;
  \star
  \,
  F^{(i)}_{\mathrm{deg}(i)} 
    = 
  \textstyle{\sum_j}
  \mu^i_j
  \,
  F^{(j)}_{\mathrm{deg}_j}
  \,,
\end{equation}
such as in the example of 5D Maxwell-Chern-Simons theory (\S\ref{5DMaxwellChernSimonsTheory}). The \emph{proper} quantization of such non-linear Gauss laws must be \cite[\S 3]{SS25-Flux} in ``extraordinary'' \emph{non-abelian} cohomology \cite[\S 1]{SS25-TEC}\cite[\S 2]{FSS23-Char}, a basic fact that has received little to no traditional attention.

Turning this situation right-side-up means to promote \cite{SS24-Phase} \emph{proper flux quantization} (of non-linear Gauss-laws, in non-abelian cohomology) from an afterthought to the starting point of the construction of (higher) quantum gauge field theories, and to observe \cite{SS24-Obs} that this \emph{global completion} of the gauge field content actually determines the topological quantum observables.

Here we illustrate this novel approach --- to global topological completion of (higher) gauge quantum field theory --- in the instructive example of Chern-Simons Wilson loop observables, along the following outline (cf. Fig. \ref{OutlineDiagram}):

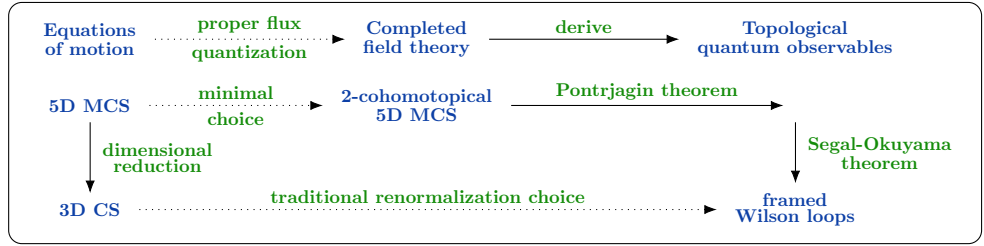
\begin{figure}[htb]
\caption{
  \label{OutlineDiagram}
  Flow chart of the logic of our discussion.
}

\centering

\adjustbox{
  rndfbox=5pt
}{
\begin{tikzcd}[
  row sep=5pt, 
  column sep=32pt
]
  \substack{  
    \scalebox{.7}{
      \bf \color{darkblue}
      Equations
    }
    \\
    \scalebox{.7}{
      \bf \color{darkblue}      
      of motion
    }
  }
  \ar[
    -Latex,
    rr,
    dotted,
    "{
      \scalebox{.7}{
        \bf \color{darkgreen}
        proper flux
      }
    }",
    "{
      \scalebox{.7}{
        \bf \color{darkgreen}
        quantization
      }
    }"{swap}
  ]
  &&
  \substack{
    \scalebox{.7}{
      \bf \color{darkblue}
      Completed
    }
    \\
    \scalebox{.7}{
      \bf \color{darkblue}
      field theory
    }
  }
  \ar[
    -Latex,
    rr,
    "{
      \scalebox{.7}{
        \bf \color{darkgreen}
        derive
      }
    }"
  ]
  &&
  \substack{
    \scalebox{.7}{
      \bf \color{darkblue}
      Topological
    }
    \\
    \scalebox{.7}{
      \bf \color{darkblue}
      quantum observables
    }
  }
  \\
  \scalebox{.7}{
    \bf \color{darkblue}
    5D MCS
  }
  \ar[
    rr,
    dotted,
    -Latex,
    "{
      \scalebox{.7}{
        \bf \color{darkgreen}
        minimal
      }
    }",
    "{
      \scalebox{.7}{
        \bf \color{darkgreen}
        choice
      }
    }"{swap}
  ]
  \ar[
    -Latex,
    d,
    "{
      \substack{
        \scalebox{.7}{
          \bf \color{darkgreen}
          dimensional
        }
        \\
        \scalebox{.7}{
          \bf \color{darkgreen}
          reduction
        }
      }
    }"
  ]
  &&
  \substack{
    \scalebox{.7}{
      \bf \color{darkblue}
      2-cohomotopical
    }
    \\
    \scalebox{.7}{
      \bf \color{darkblue}
      5D MCS
    }
  }
  \ar[
    rr,
    -Latex,
    "{
      \scalebox{.7}{
        \bf \color{darkgreen}
        Pontrjagin theorem
      }
    }"
  ]
  &&
  {}
  \ar[
    d,
    -Latex,
    "{
      \substack{
        \scalebox{.7}{
          \bf \color{darkgreen}
          Segal-Okuyama
        }
        \\
        \scalebox{.7}{
          \bf \color{darkgreen}
          theorem
        }
      }
    }"
  ]
  \\[15pt]
  \scalebox{.7}{
    \bf \color{darkblue}
    3D CS
  }
  \ar[
    rrrr,
    -Latex,
    dotted,
    "{ 
      \scalebox{.7}{
        \bf \color{darkgreen}
        traditional 
        renormalization choice
      }
    }"
  ]
  && &&
  \substack{
    \scalebox{.7}{
      \bf \color{darkblue}
      framed
    }
    \\
    \scalebox{.7}{
      \bf \color{darkblue}
      Wilson loops
    }
  }
\end{tikzcd}
}

\end{figure}

\begin{itemize}
  \item[\S\ref{TraditionalAbelianChernSimonsTheory}]
  \textbf{Review:} The traditional story of abelian 3D Chern-Simons (CS) Wilson loop quantum observables, highlighting their \emph{ad hoc} renormalization.

  \item[\S\ref{5dMaxwellChernSimonsAndItsDimReduction}] \textbf{Interlude:} Classical 3D Chern-Simons as a constrained dimensional reduction of 5D Maxwell-Chern-Simons theory (MCS).

  \item[\S\ref{CompletionViaProperFluxQuantization}]
  \textbf{Punchline:} After global completion of 5D MCS (and thus of its reduction to 3D CS) by flux quantization in \emph{Cohomotopy}, the topological quantum observables are well-defined at once, and the traditional Wilson loop observables, including their traditional renormalization, emerge as a derivable consequence.
\end{itemize}

Key mathematical results in \S\ref{CompletionViaProperFluxQuantization} are from \cite{SS25-AbelianAnyons, SS25-FQH}, here reconsidered field-theoretically in terms of topological completion of dimensionally reduced 5D Maxwell-Chern-Simons theory.

We originally discovered these results by analyzing the situation on Seifert orbifold singularities \cite{SS25-Seifert, SS25-Srni, SS26-Rickles} 
of a more ambitious program of completing the theory of \emph{M5-branes} by proper flux quantization 
\cite{GSS25-M5, FSS21-Hopf, FSS21-StrStruc}, as well as their ambient 11D supergravity theory \cite{Sati2018, GSS24-SuGra, FSS20-H}. Of course, some completion of 11D supergravity is famously conjectured to be the ultimate (``UV-'')complete fundamental quantum theory (working title ``M-theory'' \cite{Duff1999LaymansGuide, Duff1999World}, cf. \cite{nLab:MTheory, Sati2010}), and the 5D supergravity touched upon in \S\ref{5dMaxwellChernSimonsAndItsDimReduction} is famously a small sibling \cite{Mizoguchi1998}. 
In this way, one may regard our approach as following Atiyah's suggestion \cite[41:54]{Atiyah2000Millennium}
for how to attack the millennium problem.

Our discussion here highlights how the modest but instructive Chern-Simons sector of this ambitious M-theoretic framework illustrates how traditional incremental re-normalization choices may emerge intrinsically when one has a complete quantum theory to begin with.

But such a more complete theoretical understanding also entails more fine-grained predictions of potential relevance in experimental physics --- in our case concerning topological quantum materials, cf. \cite[Fig. D]{SS25-FQH}\cite{SS25-FQAH}:

Namely, abelian Chern-Simons theory serves
\cite{nLab:AbelianChernSimons} as the \emph{effective} (long-range, ``IR'') quantum field theory of magnetic flux penetrating strongly coupled 2D electron gases called \emph{fractional quantum Hall systems} \cite{nLab:FQH}, where the Wilson loops are the \emph{braided worldlines} of flux quanta (associated with vortices in the electron gas known as ``quasi-holes'') on top of a given rational fraction of the number of electrons (cf. Fig. \ref{BraidingOfFQHAnyons}). 

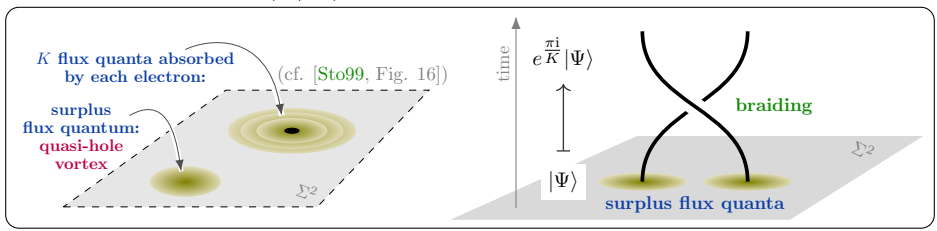
\begin{figure}[htb]
\caption{
  \label{BraidingOfFQHAnyons}
  The Wilson loop observables of abelian Chern-Simons theory describe the braiding phases of anyonic flux quanta associated with \emph{vortices} in fractional quantum Hall systems, where each crossing (of anyon worldlines) contributes a fixed phase factor $\exp(\pi \mathrm{i}/K)$, cf. \S\ref{TheKnotTheoreticExpression}.
}

\centering

\adjustbox{
  rndfbox=5pt
}{
\hspace{-.6cm}
\adjustbox{
  scale=.9
}{
\begin{tikzpicture}[
  baseline=(current bounding box.center),
  scale=.75
]

\node
  at (.3,.55+.8)
  {
    \adjustbox{
      bgcolor=white,
      scale=.7
    }{
      \color{darkblue}
      \bf
      \def\arraystretch{.9}
      \begin{tabular}{c}
        surplus
        \\
        flux quantum:
        \\
        \color{purple}
        quasi-hole
        \\
        \color{purple}
        vortex
      \end{tabular}
    }
  };

\draw[
  dashed,
  fill=lightgray
]
  (0,0)
  -- (5,0)
  -- (7+.3-.1,2+.3)
  -- (2.8+.3+.1,2+.3)
  -- cycle;

\begin{scope}[
  shift={(2.4,.5)}
]
\shadedraw[
  draw opacity=0,
  inner color=olive,
  outer color=lightolive
]
  (0,0) ellipse (.7 and .3);
\end{scope}

\begin{scope}[
  shift={(4.5,1.5)}
]

\begin{scope}[
 scale=1.8
]
\shadedraw[
  draw opacity=0,
  inner color=olive,
  outer color=lightolive
]
  (0,0) ellipse (.7 and .25);
\end{scope}

\begin{scope}[
 scale=1.45
]
\shadedraw[
  draw opacity=0,
  inner color=olive,
  outer color=lightolive
]
  (0,0) ellipse (.7 and .25);
\end{scope}

\shadedraw[
  draw opacity=0,
  inner color=olive,
  outer color=lightolive
]
  (0,0) ellipse (.7 and .25);

\begin{scope}[
  scale=.2
]
\draw[
  fill=black
]
  (0,0) ellipse (.7 and .25);
\end{scope}

\end{scope}

\draw[
  white,
  line width=2
]
  (1.3, 1.8) .. controls 
  (2,2.2) and 
  (2.2,1.5) ..
  (2.32,.7);
\draw[
  -Latex,
  black!70
]
  (1.3, 1.8) .. controls 
  (2,2.2) and 
  (2.2,1.5) ..
  (2.32,.7);

\node
  at (1.3,2.7)
  {
    \adjustbox{
      scale=.7
    }{
      \color{darkblue}
      \bf
      \def\arraystretch{.9}
      \def\tabcolsep{-5pt}
      \begin{tabular}{c}
        $K$ flux quanta
        absorbed
        \\
        by each electron:
      \end{tabular}
    }
  };

\draw[
 line width=2.5pt,
  white
]
  (2.4, 3.1) .. controls 
  (2.8,3.3) and 
  (4,3.5) ..
  (4.3,1.8);

\draw[
  -Latex,
  black!70
]
  (2.4, 3.1) .. controls 
  (2.8,3.3) and 
  (4,3.5) ..
  (4.3,1.8);

\node at 
  (5.9,2.6)
  {
   \scalebox{.8}{
     \color{gray}
     (cf. \cite[Fig. 16]{Stormer99})  
   }
  };

\node[
  gray,
  rotate=-20,
  scale=.73
] 
  at (4.8,+.3) {$\Sigma^2$};

\end{tikzpicture}
}
\hspace{-.6cm}
\adjustbox{}{
  \begin{tikzpicture}[
    baseline=(current bounding box.center),
    xscale=.7
  ]
    \draw[
      gray!30,
      fill=gray!30
    ]
      (-4.6,-1.5) --
      (+1.8,-1.5) --
      (+1.8+3-.5,-.4) --
      (-4.6+3+.5,-.4) -- cycle;

    \begin{scope}[
      shift={(-1,-1)},
      scale=1.2
    ]
    \shadedraw[
      draw opacity=0,
      inner color=olive,
      outer color=lightolive
    ]
      (0,0) ellipse (.7 and .1);
    \end{scope}

    \draw[
     line width=1.4
    ]
      (-1,-1) .. controls
      (-1,0) and
      (+1,0) ..
      (+1,+1);

  \begin{scope}
    \clip 
      (-1.5,-.2) rectangle (+1.5,1);
    \draw[
     line width=7,
     white
    ]
      (+1,-1) .. controls
      (+1,0) and
      (-1,0) ..
      (-1,+1);
  \end{scope}
  
    \begin{scope}[
      shift={(+1,-1)},
      scale=1.2
    ]
    \shadedraw[
      draw opacity=0,
      inner color=olive,
      outer color=lightolive
    ]
      (0,0) ellipse (.7 and .1);
    \end{scope}
    \draw[
     line width=1.4
    ]
      (+1,-1) .. controls
      (+1,0) and
      (-1,0) ..
      (-1,+1);

  \node[
    rotate=-25,
    scale=.7,
    gray
  ]
    at (3.2,-.58) {
      $\Sigma^2$
    };

  \draw[
    -Latex,
    gray
  ]
    (-3.4,-1.35) -- 
    node[
      near end, 
      sloped,
      scale=.7,
      yshift=7pt
      ] {time}
    (-3.4, 1.2);

  \node[
    scale=.7
  ] at 
    (0,-1.3)
   {\bf \color{darkblue} 
   surplus flux quanta};

  \node[
    scale=.7
  ] at 
    (1.5,0)
   {\bf \color{darkgreen} braiding};

  \node[
    fill=white,
    scale=.8
  ] at (-2.5,-1) {$
    \vert \Psi \rangle
  $};

  \node[
    fill=white,
    scale=.8
  ] at (-2.5,.7) {$
    e^{\tfrac{\pi \mathrm{i}}{K}}
    \vert \Psi \rangle
  $};

  \draw[
    |->,
    black!80,
    line width=.5
  ]
    (-2.5,-.6) --
    (-2.5, .3);

  \end{tikzpicture}
  }
  \hspace{-.3cm}
}
  
\end{figure}

Remarkably, these anyonic braiding phases predicted by Chern-Simons Wilson loops have consistently been observed by experimentalists in recent years (starting with \cite{Nakamura2020}, for more recent developments see \cite{Ghosh2025}). Especially in view of the expected role of anyonic quantum materials as future hardware platforms for much anticipated robust (topological) quantum computers \cite{nLab:TopologicalQuantumComputing}, this leads to many practical questions, such as whether there might also be \emph{non-abelian defect} anyons in FQH systems or how fractional quantum \emph{anomalous} Hall systems (FQAH) could exhibit more tractable anyons at room temperature and without magnetic fields. 
The completed reformulation of Chern-Simons theory discussed here makes clear predictions on both these counts, discussed in  \cite{SS25-FQH} and \cite{SS25-FQAH}, respectively.

\medskip

\paragraph{\bf Acknowledgements.}
We thank the organizers of the Oberwolfach Mini-Workshop 2539b ``Renormalization and Randomness'', where this article was first presented (in September 2025). For useful discussion we thank Grigorios Giotopoulos and Konrad Waldorf.
We acknowledge support by {\it Tamkeen UAE} under the 
{\it NYU Abu Dhabi Research Institute grant} {\tt CG008}.

\section{Traditional Abelian Chern-Simons Theory}
\label{TraditionalAbelianChernSimonsTheory}

To set the scene, we briefly review (\S\ref{TraditionalCSWilsonLoops}) the traditional story of  Wilson loop observables in abelian Chern-Simons theory, highlighting the \emph{ad hoc} renormalization choice.

\subsection{The Lagrangian density and the Level}
\label{3DCSLagrangianAndLevel}

Before we really start, for the useful understanding of the Wilson loop observables as reflecting anyon braiding phases (cf. Fig. \ref{BraidingOfFQHAnyons}) it is important to pinpoint a notoriously subtle factor of 2 in the definition of the \emph{level} of abelian Chern-Simons theory \cite[\S 5]{DijkgraafWitten1990}, and so we recall here (field theory experts may want to skip ahead) the standard definition of the abelian Chern-Simons action functional \cite{nLab:AbelianChernSimons}:

\subsubsection{Abelian Chern-Simons theory in general}
The local gauge field content of minimal abelian Chern-Simons theory is a gauge potential 1-form $A$ with field-strength/flux density 2-form  $F_2 := \mathrm{d}A$, and the (local) Lagrangian density is traditionally written as the following multiple of the \emph{Chern-Simons form} $A \wedge \mathrm{d}A$:
\begin{equation}
  \label{TheLagrangianDensity}
  L(A)
  \;\coloneqq\;
  \tfrac{K}{4\pi} 
  \,
  A \wedge F_2
  \,.
\end{equation} 
Here $K \in \mathbb{R}$ is a parameter whose admissible values are to be fixed such that
for any closed oriented manifold $\Sigma^3$ the \emph{exponentiated action functional} 
is well defined:
\begin{equation}
  \label{ExponentiatedActionFunctional}
  \exp\big(
    \mathrm{i}
    S(\Sigma^3, A)
  \big)
  \;:=\;
  \tensor[^{\mbox{``}}]{
    \exp\Big(
      \mathrm{i}
      \tfrac{K}{4\pi}
      \textstyle{\int_{\Sigma^3}} 
      \,
      A \wedge F_2
    \Big)
  }{^{\mbox{''}}}
  \;\in\;
  \mathrm{U}(1)
  \,.
\end{equation}
Of course, as stated, the right hand side of \eqref{ExponentiatedActionFunctional} is generally not defined for any value of $K$, since the gauge potential $A$ is in general not globally defined on $\Sigma^3$ (cf. Fig. \ref{CechDataForHigherGaugeField}), only its flux density $F_2$ is.

To fix this, the traditional flux quantization of the Chern-Simons gauge field is taken to be ordinary Dirac quantization \cite{Alvarez1985}\cite[\S 16.4e]{Frankel2011} in ordinary differential 2-cohomology \cite[Ex. 3.9-10]{SS25-Flux}\cite{nLab:OrdinaryDiffCohomology}.
Via the formulation of \emph{Cheeger-Simons differential characters} \cite{nLab:CheegerSimonsDiffCharacters} this means to extend $\tfrac{1}{2\pi} F_2$ to an \emph{integral form} with integral periods
\footnote{
  This condition $\tfrac{1}{2\pi} F_2 \in \Omega^2_{\mathrm{dR}}(-)_{\mathrm{int}}$ is the normalization choice common in theoretical physics. In the mathematical literature one typically nromalizes such that $F_2$ itself is integral.
}
on
closed oriented 4-manifolds $\Sigma^4$ with boundary $\partial \Sigma^4 \simeq \Sigma^3$ and declare, inspired by Stokes' theorem, that \eqref{ExponentiatedActionFunctional} really means:
\begin{equation}
  \label{ExponentiatedActionFixed}
  \exp\Big(
    \mathrm{i}
    S\big(\Sigma^3, \widehat{A}\big)
  \Big)
  :=
  \exp\Big(
    \mathrm{i}
    \tfrac{K}{4\pi}
    \textstyle{\int_{\Sigma^4}}
    \, 
    F_2 \wedge F_2
  \Big)
  \;=\;
  \exp\Big(
    \pi \mathrm{i} K
    \textstyle{\int_{\Sigma^4}}
    \frac{F_2}{2\pi} 
      \wedge 
    \frac{F_2}{2\pi}
  \Big)
  \mathrlap{\,.}
\end{equation}
Now the right hand side is well-defined for any choice of $\Sigma^4$, but to be well-defined as a functional depending only on $\Sigma^3$ one still needs to require that for \emph{closed} oriented $\widehat{\Sigma}^4$ (arising as the gluing of one choice of $\Sigma^4$ with the orientation reversal of any other, along their common boundary $\Sigma^3$) that
\begin{equation}
  \pi \mathrm{i} K
  \,
  \underset{
    c_1^2(\Sigma^4)
  }{
    \underbrace{
      \textstyle{\int_{\widehat{\Sigma}^4}}
      \frac{F_2}{2\pi} 
        \wedge 
      \frac{F_2}{2\pi}
    }
  }
  \;\in\;
  2\pi\mathrm{i}
  \mathbb{Z}
  \,.
\end{equation}
But here the \emph{Pontrjagin number} over the brace is integral (by integrality of $\tfrac{1}{2\pi}F_2$)
\begin{equation}
  c_1^2(\Sigma^4) \in \mathbb{Z}
  \,,
\end{equation} 
so that the condition for \eqref{ExponentiatedActionFixed} to be well-defined is finally that: 
\begin{equation}
  \label{TheLevel}
  k 
    \;\coloneqq\; 
  \tfrac{K}{2}
    \;\in\;
  \mathbb{Z}
\end{equation} 
must be an integer. 
This $k$ is the \emph{Chern-Simons level}. 
On the other hand, $K = 2k$ is the (even, for now) integer that gets promoted to the eponymous matrix $(K_{i j})$ in ``K-matrix Chern-Simons theory'' \cite{nLab:AbelianChernSimons}, where several gauge fields couple to each other. But beware that conventions differ and that our ``$K$'' is very often denoted ``$k$'' in the literature. 

\subsubsection{Abelian Chern-Simons theory on spin-manifolds}

It is also $K$ in \eqref{TheLevel} that is more fundamental as one restricts attention to spin manifolds $\Sigma^3$, hence to \emph{spin Chern-Simons theory} \cite[\S 5]{DijkgraafWitten1990}: From the observation that the cobordism ring for spin manifolds equipped with complex line bundles is trivial in degree 3, and that 
\begin{equation}
  \widehat{\Sigma}^4 
  \;
  \text{spin manifold}
  \;\;
  \Rightarrow
  \;\;
  \int_{\widehat{\Sigma}^4}
  \frac{F_2}{2\pi} 
    \wedge 
  \frac{F_2}{2\pi}
  \;\in\;
  2\mathbb{Z}
  \,,
\end{equation}
it follows that for spin Chern-Simons theory one may allow also half-integral level $k$ \eqref{TheLevel} and hence odd integral $K$:
\begin{equation}
  \label{HalfIntegralLevel}
  k 
    \;\in\;
  \tfrac{1}{2}\mathbb{Z} 
  \;\;\;\;\;\;
  \text{equivalently}
  \;\;\;\;\;\;
  K \coloneqq 2k 
  \;\in\;
  \mathbb{Z}
  \,.
\end{equation}

\subsubsection{Equation of motion}
In closing these basics, just to highlight that the Euler-Lagrange equation of motion associated with \eqref{TheLagrangianDensity} is simply
\begin{equation}
  \label{ChernSimonsEoM}
  F_2 = 0
  \mathrlap{\,,}
\end{equation}
and hence is, of course, not of general Maxwell-type \eqref{HigherMaxwellTypeEquationsOfMotion}. For this reason, we will realize (in \S\ref{5dMaxwellChernSimonsAndItsDimReduction} below) 3D CS as a constrained reduction of 5D Maxwell-CS theory and then (in \S\ref{CompletionViaProperFluxQuantization}) globally complete the latter by proper flux quantization.

\subsection{Traditional CS Wilson Loops in 3D}
\label{TraditionalCSWilsonLoops}

We review the traditional argument
(going back to \cite{Polyakov1988}, popularized by \cite[\S 2.1]{Witten1989}, cf. \cite[\S 1]{Giavarini1994}) that the expectation value of the Wilson loop quantum observable in $\mathrm{U}(1)$ Chern-Simons theory on $S^3$ is essentially the exponentiated \emph{writhe} of the given link (cf. Fig. \ref{FramedLinksAndTheirWrithe}), namely of the sum of its linking and framing number, where the framing arises as a re-normalization choice in the process of making sense of the path integral heuristics.

Historically important as this argument was, it should be clear that it is hardly satisfactory as a process of deriving quantum observables: Imagine, with the modest example of Chern-Simons field theory taking this much handwaving, how slim the chances are of unraveling the secrets of richer quantum field theories by similar reasoning.

\subsubsection{Wilson loops}

On $\Sigma^3 = \mathbb{R}^3_{\cpt} \simeq S^3$ the 3-sphere, the ordinary gauge field of abelian Chern-Simons theory does have a representation by a globally defined differential 1-form 
\begin{equation}
  \label{1Forms}
  A \in \Omega^1_{\mathrm{dR}}\big(S^3\big)
  \mathrlap{\,.}
\end{equation}

Then given an oriented link (cf. \cite[\S 1.1]{Ohtsuki2001})
\begin{equation}
  \label{ALink}
  \begin{tikzcd}[sep=small]
  \gamma 
    \colon 
  (S^1)^{n} 
    \ar[r, hook]
    &\mathbb{R}^3 
    \ar[r, hook]
    &
    S^3
  \end{tikzcd}
\end{equation} 
with $n$ connected components 
\begin{equation}
  \label{ConnectedComponentsOfALink}
  \big\{
    \begin{tikzcd}[sep=small]
    \gamma_i \colon S^1 
     \ar[r, hook]
     &
    \mathbb{R}^n
    \end{tikzcd}
  \big\}_{i \in [n]}
  \,,
\end{equation}
its  \emph{Wilson loop}/\emph{holonomy} with respect to such a gauge potential \eqref{1Forms} is the exponentiated integral of $A$ along the link:
\begin{equation}
  \label{Holonomy}
  \exp\bigg(
    \mathrm{i}
    \textstyle{\underset{{(S^1)^n}}{\int}} 
    \!\!
    \gamma^\ast A
  \bigg)
  \;\in\;
  \mathrm{U}(1)
  \,.
\end{equation}

The \emph{expectation value} $\langle-\rangle$ of the $\gamma$-Wilson loop \eqref{Holonomy}, regarded as a \emph{quantum observable} in abelian Chern-Simons theory, to be denoted
\begin{equation}
  \Bigg\langle\!
  \exp\bigg(
    \mathrm{i}
    \textstyle{\underset{{(S^1)^n}}{\int}} 
    \!\!
    \gamma^\ast A
  \bigg)
  \! \Bigg\rangle
  \;\in\;
  \mathbb{C}
  \mathrlap{\,,}
\end{equation}
is traditionally meant to be the normalized ``path integral $\int (\mbox{-}) D A$'' over the space of gauge orbits of $A$ \eqref{1Forms} of the expression \eqref{Holonomy} weighted by the exponentiated action functional \eqref{ExponentiatedActionFunctional} of abelian Chern-Simons theory, suggestively written as:
\begin{equation}
  \label{PathIntegralForAbelianChernSimonsWilsonLoops}
  \Bigg\langle \!
  \exp\bigg(
    \mathrm{i}
    \textstyle{\underset{{(S^1)^n}}{\int}} 
    \!\!
    \gamma^\ast A
  \bigg)
  \!\Bigg\rangle
  \;\;
  =
  \;\;
  \tensor[^{\mbox{``}}]{
  \frac{
    \displaystyle{\int D A}
    \;
    \exp\bigg(
      \mathrm{i}
      \tfrac{K}{4\pi}
      \textstyle{\underset{S^3}{\int}}
      A \wedge \mathrm{d}A
      \;+\;
      \mathrm{i}
      \textstyle{\underset{{(S^1)^n}}{\int}} 
      \!\!
      \gamma^\ast A
    \bigg)     
  }{
    \displaystyle{\int D A}
    \;
    \exp\bigg(
      \mathrm{i} 
      \tfrac{K}{4\pi}
      \textstyle{\underset{S^3}{\int}}
      A \wedge \mathrm{d}A
    \bigg)     
  }}{^{\mbox{''}}}
  \mathrlap{\,.}
\end{equation}

As usual for path integrals, the expression on the right of \eqref{PathIntegralForAbelianChernSimonsWilsonLoops} remains undefined, since a suitable measure ``$D A$'' is not known and may not exist. \footnote{
  While for the simple case of plain abelian Chern-Simons theory 
  candidate path integral measures in themselves do exist
  (since this is such a simple instance of a quantum field theory that it is ``free'' in that its action functional  is formally Gaussian), the Wilson loop functional \eqref{Holonomy} is not measurable for known choices, whence \eqref{PathIntegralForAbelianChernSimonsWilsonLoops} still remains ill-defined without further renormalization choices (cf. \cite{Leukert1996, Hahn2004a, Hahn2004b}).
}
As also usual for path integrals, one proceeds instead by postulating that the expression on the left of \eqref{PathIntegralForAbelianChernSimonsWilsonLoops} satisfies enough structural properties that are \emph{suggested} by the symbols on the right as to be uniquely fixed by these conditions --- and then take that to be the definition.

\subsubsection{The Gaussian path integral}
 \label{AbelianCSWilsonLoopsViaGaussianIntegrals}

Since the abelian Chern-Simons Lagrangian density $\tfrac{K}{4\pi} A \wedge \mathrm{d}A$ \eqref{TheLagrangianDensity} is a quadratic form in $A$, one envisions that this expression \eqref{PathIntegralForAbelianChernSimonsWilsonLoops} follows the transformation rules of Gaussian integrals over elements $\phi \in \mathbb{R}^n$ of finite-dimensional Cartesian spaces, for which 
$$
  Z(J)
  \;\coloneqq\;
  \int \mathrm{d}^n \phi\;
  \exp\Big(
    - \tfrac{1}{2}
    \phi \cdot M \cdot \phi 
    +
    J \cdot \phi
  \Big)
$$
evaluates to 
\begin{equation}
  \label{EvaluationOfOrdinaryGaussianIntegral}
  Z(J)
  \;=\;
  Z(0)
  \,
  \exp\Big(
    \tfrac{1}{2}
    \, 
    J \cdot M^{-1} \cdot J
  \Big)
  \,.
\end{equation}

Systematically following through this analogy (using Fourier transform to turn the derivative in $A \mathrm{d} A$ into a multiplication operation) with some care (such as with attention to gauge fixing) suggests the expression:
\begin{equation}
  \Bigg\langle \!
  \exp\bigg(
    \mathrm{i}
    \textstyle{\underset{{(S^1)^n}}{\int}} 
    \!\!
    \gamma^\ast A
  \bigg)
  \! \Bigg\rangle
  \;=\;
  \tensor[^{\mbox{``}}]{
  \exp\bigg(\!
    -\tfrac{1}{2}
    \textstyle{\underset{\gamma}{\iint}}
    \big\langle
      A_\mu(x)
      ,
      A_\nu(y)
    \big\rangle
    \, 
    \mathrm{d} x^\mu\, 
    \mathrm{d} y^\nu
  \bigg)
  }{^{\mbox{''}}}
  \,,
\end{equation}
whereby the \emph{Chern-Simons propagator} -- namely the analog of $M^{-1}$ in \eqref{EvaluationOfOrdinaryGaussianIntegral} -- is suggested to be
\begin{equation}
  \label{AbelianCSPropagatorForComputationOfWilsonLoops}
  \big\langle
    A_\mu(x)
    ,
    A_\nu(y)
  \big\rangle
  \;=\;
  \tensor[^{\mbox{``}}]{
  \frac{\mathrm{i}}{2K}
  \epsilon_{\mu \nu \rho}
  \frac{
    x^\rho - y^\rho
  }{
    \left\vert x - y \right\vert^3
  }
  }{^{\mbox{''}}}
  \,,
\end{equation}
which looks more concrete but is still ill-defined at $x = y$.

In summary, this suggests that
\begin{equation}
  \label{AbelianCSWilsonLoopObservableByNaivePathIntegral}
  \Bigg\langle \!
  \exp\bigg(
    \mathrm{i}
    \textstyle{\underset{{(S^1)^n}}{\int}} 
    \!\!
    \gamma^\ast A
  \bigg)
  \! \Bigg\rangle
  \;=\;
  \tensor[^{\mbox{``}}]{
  \exp\bigg(\!
    -\tfrac{\mathrm{i}}{4K}
    \textstyle{\underset{\gamma}{\iint}}
    \, 
    \mathrm{d} x^\mu\, 
    \mathrm{d} y^\nu
    \epsilon_{\mu \nu \rho}  
    \frac{
      \rule[-2.5pt]{2pt}{0pt}%
      x^\rho - y^\rho
    }{
      \rule{0pt}{7pt}%
      \left\vert x - y \right\vert^3
    } 
  \bigg)
  }{^{\mbox{''}}}
  \,.
\end{equation}

This is the proposed expression due to 
\cite[(3)]{Polyakov1988}, popularized by 
\cite[(2.29)]{Witten1989}.
It is \emph{almost} the time-honored integral expression for the Gauss linking number of the link $\gamma$ --- cf. \eqref{AbelianCSWilsonLoopsInTermsOfFramingAndLinkingNumbers} below --- except for those contributions where $x$ and $y$ run over the same connected component \eqref{ConnectedComponentsOfALink}, where it is not well-defined.

In a final step of the traditional argument, one applies a renormalization argument to ``regulate'' the "divergencies" of the exponent in \eqref{AbelianCSWilsonLoopObservableByNaivePathIntegral} when $x = y$.

\subsubsection{The renormalization choice}

Possible regularizations/renormalizations (hence: re-definitions!) of \eqref{AbelianCSPropagatorForComputationOfWilsonLoops} have been considered by \cite[(5)]{Polyakov1988} and \cite[\S 7]{Leukert1996}. 
\footnote{
  Even though \cite{Leukert1996} give a rigorous construction of the path integral measure, the Wilson loop observable still needs renormalization to make sense, cf. \S7 there.
}
But \cite[p. 363]{Witten1989} asserted that it is ``clear'' that the following \emph{point-splitting regularization} should be used instead, and this has become the commonly accepted choice since (cf. \cite[(9)]{Kaul1999}\cite[\S 7]{Hahn2004a}\cite[\S 2.5]{Hahn2004b}\cite[\S 3]{Guadagnini2008}\cite[(5.1)]{Mezei2018}):

First, re-define the oriented link $\gamma$ \eqref{ALink} to be equipped with a \emph{framing}, hence with a unit vector field \begin{equation}
  \label{FrameFieldOnLink}
  u
  \in 
  \Gamma\big(
    T_\gamma \mathbb{R}^3
  \big)
\end{equation}
on $\gamma$, which is normal to the tangent vector field $\dot \gamma$. 
With that, finally define \eqref{AbelianCSWilsonLoopObservableByNaivePathIntegral} by replacing all occurrences of $y$ with a shifted version $y + s u$, for arbitrarily small lengths $s \in \mathbb{R}_+$ of the normal/framing vector $s u$:
\begin{equation}
  \label{AbelianCSWilsonLoopObservableRenormalized}
  \Bigg\langle \!
  \exp\bigg(
    \mathrm{i}
    \textstyle{\underset{{(S^1)^n}}{\int}}
    \!\!
    \gamma^\ast A
  \bigg)
  \! \Bigg\rangle
  \coloneqq
  \underset{
    s \to 0
  }{\lim}
  \exp\bigg(\!
    -\tfrac{\mathrm{i}}{4K}
    \textstyle{\underset{\mathclap{(S^1)^n}}{\iint}}
    \, 
    \mathrm{d} x^\mu 
    \mathrm{d} y^\nu
    \,
    \epsilon_{\mu \nu \rho}  
    \frac{
      \rule[-2pt]{0pt}{0pt}%
      x^\rho 
        - 
      (y^\rho + s u^\rho)
    }{
      \rule{0pt}{7pt}%
      \left\vert 
        x
        - 
        (y + s u)
      \right\vert^3
    } 
  \bigg)
  \,,
\end{equation}
where the limit $s \to 0$ is over normal vector fields $s u$ that keep the given normal direction $u$ but shrink in length $s$ (cf. \cite[(1.5)]{Giavarini1994}).

This expression \eqref{AbelianCSWilsonLoopObservableRenormalized} is finally well-defined. 

\subsubsection{The knot-theoretic expression}
\label{TheKnotTheoreticExpression}

While the exponentiated integral form of \eqref{AbelianCSWilsonLoopObservableRenormalized} is that motivated from the path integral picture, it now turns out to have a purely combinatorial/knot-theoeretic interpretation:

The shift along the framing vector in \eqref{AbelianCSWilsonLoopObservableRenormalized}does not change that part of the integral, with respect to the naive \eqref{AbelianCSWilsonLoopObservableByNaivePathIntegral}, where $\sigma$ and $\sigma'$ run over distinct connected components \eqref{ConnectedComponentsOfALink} of the link, hence that part now does give the sum of the linking numbers among the connected components of the oriented link. At the same time, the shift makes the contributions where $\sigma$ and $\sigma'$ do run over the same connected component become the linking number of that component with its shift by the framing, which is its \emph{self-linking} or \emph{framing number}.

The end result of this argument is that the exponent becomes the sum of the framing numbers $\mathrm{frm}(\gamma_i) \in \mathbb{Z}$ of each link component $\gamma_i$ \eqref{ConnectedComponentsOfALink} with the linking numbers $\mathrm{lnk}(\gamma_i, \gamma_j) \in \mathbb{Z}$ all pairs of distinct link components, which is its \emph{writhe}, $wrth(\gamma) \in \mathbb{Z}$, of the link (\cite[(2.31) ff]{Witten1989}\cite[\S 9]{Hahn2004a}\cite[\S 5]{Hahn2004b}\cite[(5.2)]{Guadagnini2008}, cf. \cite[(9)]{Kaul1999}\cite[(5.1)]{Mezei2018}):
\begin{equation}
  \label{AbelianCSWilsonLoopsInTermsOfFramingAndLinkingNumbers}
  \Bigg\langle \!
  \exp\bigg(
    \mathrm{i}
    \textstyle{\underset{{(S^1)^n}}{\int}} 
    \!\!
    \gamma^\ast A
  \bigg)
 \! \Bigg\rangle
  =
  \exp\bigg(
    \tfrac{\pi \mathrm{i}}{K}
    \Big(
      \underset{
        \mathrm{wrth}(\gamma)
      }{
      \underbrace{
        \textstyle{\underset{i}{\sum}}
        \,
        \mathrm{frm}(\gamma_i)
        +
        \textstyle{\underset{i \neq j}{\sum}}
        \,
        \mathrm{lnk}(\gamma_i, \gamma_j)
      }
      }
    \Big)
  \!\bigg)
  \mathrlap{\,.}
\end{equation}
Note though that none of the intermediate steps towards \eqref{AbelianCSWilsonLoopObservableRenormalized} was well-defined, so that this traditional argument --- impactful as it historically was and ``correct'' as it may be in its conclusion \eqref{AbelianCSWilsonLoopsInTermsOfFramingAndLinkingNumbers} --- is somewhat unsatisfactory as a derivation of quantum observables from (Lagrangian) input data.

\begin{figure}[htb]
\caption{
  \label{FramedLinksAndTheirWrithe}
  Examples of (blackboard-)framed
  link(diagram)s and their writhe number.
}
\centering

\begin{tabular}{ccc}

\adjustbox{
  scale=.76
}{
\begin{tikzpicture}[
  baseline=(current bounding box.center)
]

\begin{scope}[
  shift={(1,0)}
]
\draw[line width=2, -Latex]
  (0:1) arc (0:180:1);
\end{scope}

\draw[line width=7,white]
  (0:1) arc (0:180:1);
\draw[line width=2, -Latex]
  (0:1) arc (0:180:1);

\draw[line width=2, -Latex]
  (180:1) arc (180:360:1);

\begin{scope}[shift={(1,0)}]
\draw[line width=7, white]
  (180:1) arc (180:360:1);
\draw[line width=2, -Latex]
  (180:1) arc (180:360:1);
\end{scope}

\node[gray]
  at (.5,.64) {\color{red} 
    \scalebox{.9}{$-$}
  };
\node[gray]
  at (.5,-.64) {\color{red}
    \scalebox{.9}{$-$}
  };  
\end{tikzpicture}}
&$\longmapsto$&
$-2$
\\[5pt]
\adjustbox{
  scale=.6
}{
\begin{tikzpicture}[
  baseline=(current bounding box.center)
]
\foreach \n in {0,1,2} {
\begin{scope}[
  rotate=\n*120-4
]
\draw[
  line width=2.3,
  -Latex
]
 (0-.1,-1+.14)
   .. controls
   (-1,.2) and (-2,2) ..
 (0,2)
   .. controls
   (1,2) and (1,1) ..
  (.9,.7);
\end{scope}

\node[darkgreen]
  at (\n*120+31:.6) {
    \scalebox{1}{$+$}
  };

};
\end{tikzpicture}
}
&$\longmapsto$&
+3
\\[-15pt]
\adjustbox{
  scale=1
}{
\begin{tikzpicture}[
  baseline=(current bounding box.center)
]

\node at (0,0) {
  \includegraphics[width=2.2cm]{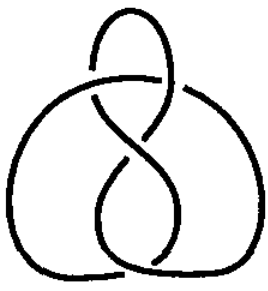}
};

\draw[line width=1.2pt,-Latex]
  (-.276,-.8) -- 
  (-.274,-.8-.02);
\draw[line width=1.2pt,-Latex]
  (.35, -.5) --
  (.35, -.47);
\draw[line width=1.2pt,-Latex]
  (-.09,1.09) --
  (-.09+.05, 1.13);
\draw[line width=1.2pt,-Latex]
  (-.18,.54) --
  (-.19-.01, .54);

\node[red]
  at (0,-1.2) {
    \scalebox{.7}{$-$}
  };
\node[red]
  at (.3,0) {
    \scalebox{.7}{$-$}
  };
\node[darkgreen]
  at (.5,.64) {
    \scalebox{.7}{$+$}
  };
\node[darkgreen]
  at (-.57,.64) {
    \scalebox{.7}{$+$}
  };

\end{tikzpicture}
}
&$\longmapsto$&
\phantom{+}0
\end{tabular}

\end{figure}
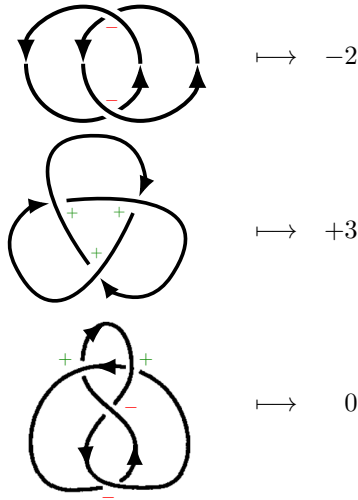

In any case, towards our alternative and systematic derivation of  \eqref{AbelianCSWilsonLoopsInTermsOfFramingAndLinkingNumbers} (culminating in \S\ref{DerivingChernSimonsQuantumObservables} below), we highlight that \eqref{AbelianCSWilsonLoopsInTermsOfFramingAndLinkingNumbers} is a function on the set of framed oriented links of the form
\begin{equation}
  \label{WilsonLoopObservableFactorization}
  \begin{tikzcd}[
    row sep=0pt
  ]
    \mathrm{FrmdOrntdLnks}
    \ar[
      rr,
      "{
        \mathrm{wrth}
      }"
    ]
    &&
    \mathbb{Z}
    \ar[
      rr,
      "{
        \exp\big(
          \tfrac{\pi \mathrm{i}}{K}
          (-)
        \big)
      }"
    ]
    &&
    \mathrm{U}(1)
    \\
    \gamma
    \ar[
      rrrr,
      |->, 
      shorten=15pt
    ]
    &&&&
   \Bigg\langle \!
      \exp\bigg(
        \mathrm{i}
        \textstyle{\underset{{(S^1)^n}}{\int}} 
        \!\!
        \gamma^\ast A
      \bigg)
   \! \Bigg\rangle
    \mathrlap{\,,}
  \end{tikzcd}
\end{equation}
depending on an integer $K \in \mathbb{Z}$,
where the \emph{writhe} of a link may be computed (cf. \cite[Def. 2.10]{SS25-AbelianAnyons}) by choosing a \emph{blackboard-framed oriented link diagram} for it (cf. Fig. \ref{FramedLinksAndTheirWrithe}) and then counting the net number $\#$ of crossings, where each crossing contributes according to:
\begin{equation}
\label{TotalLinkingNumber}
\#\left(\!
\adjustbox{raise=-.63cm, scale=.5}{
\begin{tikzpicture}

\draw[
  line width=1.2,
  -Latex
]
  (-.7,-.7) -- (.7,.7);
\draw[
  line width=7,
  white
]
  (+.7,-.7) -- (-.7,.7);
\draw[
  line width=1.2,
  -Latex
]
  (+.7,-.7) -- (-.7,.7);
 
\end{tikzpicture}
}
\!\right)
\;:=\;
+1
\,,
\hspace{1cm}
\#\left(\!
\adjustbox{raise=-.63cm, scale=0.5}{
\begin{tikzpicture}[xscale=-1]

\draw[
  line width=1.2,
  -Latex
]
  (-.7,-.7) -- (.7,.7);
\draw[
  line width=7,
  white
]
  (+.7,-.7) -- (-.7,.7);
\draw[
  line width=1.2,
  -Latex
]
  (+.7,-.7) -- (-.7,.7);
 
\end{tikzpicture}
}
\!\right)
\;\;
:=
\;\;
-1
\,.
\end{equation}

In the next sections we turn to an actual derivation of this situation, from first principles.

\section{Via Dimensional Reduction from 5D}
\label{5dMaxwellChernSimonsAndItsDimReduction}

For its completion in \S\ref{CompletionViaProperFluxQuantization}, we show here how classical 3D abelian Chern-Simons theory (\S\ref{TraditionalAbelianChernSimonsTheory}) is a constrained \emph{dimensional reduction} of a higher dimensional gauge theory of Maxwell-type (\cite[p. 4]{SS24-Phase}\cite[Def. 2.6]{SS25-Flux}), namely of 5D \emph{Maxwell-Chern-Simons theory} (\S\ref{5DMaxwellChernSimonsTheory}). This follows standard arguments but seems not to be citable from the literature, therefore we make it explicit below (\S\ref{DimensonalReductionTo3D}).

This section is purely about classical field theory. Its purpose is to give an idea as to \emph{why} it can be, as concerns the physics picture, that the completed 5D TQFT constructed in \S\ref{CompletionViaProperFluxQuantization} produces --- as we prove there that it does --- the renormalized Wilson loop quantum observables of \S\ref{TraditionalCSWilsonLoops}.

\subsection{5D Maxwell-Chern-Simons theory}
\label{5DMaxwellChernSimonsTheory}

The field content of 5D \emph{Maxwell-Chern-Simons theory} \cite{nLab:MaxwellChernSimons} is still a single abelian gauge field, hence locally a differential 1-form $A$ with flux density $F_2 := \mathrm{d}A$, with Lagrangian density
\begin{equation}
  L_{\mathrm{MCS}}^{5D}(A)
  \;\propto\;
  \grayunderbrace{
  \tfrac{1}{2}
  F_2 \wedge \star F_2
  }{
    \mathclap{
    \scalebox{.7}{Maxwell term}
    }
  }
  \;-\;
  \grayunderbrace{
  \tfrac{1}{3} A \wedge F_2 \wedge F_2
  \mathrlap{\,,}
  }{
    \mathclap{
    \scalebox{.7}{Chern-Simons term}
    }
  }
\end{equation}
hence with Maxwell-like equations of motion:
\begin{equation}
  \label{EoMof5dMCS}
  \begin{aligned}
    \mathrm{d} F_2 & = 0
    \\
    \mathrm{d} \star F_2 & = F_2 \wedge F_2
    \,.
  \end{aligned}
\end{equation}
This notably describes the gauge sector of minimal $D=5$ supergravity \cite{nLab:5DSuGra}. It is structurally identical, just up to dimension and degrees, to the gauge sector of $D=11$ supergravity (cf. \cite{GSS24-SuGra}).

\subsection{Dimensional reduction to 3D}
\label{DimensonalReductionTo3D}

We consider the dimensional reduction of the 5D Maxwell-Chern-Simons equations of motion \eqref{EoMof5dMCS} to 3D, specifically (cf. Fig. \ref{SpacetimeDomain}): 

Consider the case 
that 5D spacetime is product manifold (with product metric)
\begin{equation}
  \label{DimRedSpacetime}
  X^{1,4}
  =
  \mathbb{R}^{1,0} 
    \times 
  X^3
    \times 
  V^1
  \,,
\end{equation}
of:
\begin{itemize}
\item 
  a globally hyperbolic Lorentzian manifold $X^{1,3} = \mathbb{R}^{1,0} \times X^3$,

\item 
  a 1-dimensional Riemannian manifold $V^1$ (not necessarily connected or closed), which in suitable coordinates $v$ has metric
  $
      g_V 
        = 
      \ell_v^2 \, \mathrm{d}v \otimes \mathrm{d}v\,,
  $
  for $\ell_v \in \mathbb{R}_+$.

\item 
such that the flux density is time-independent and constant along the fiber, in that
\begin{equation}
  \label{DimRedDecompositionOfFluxDensity}
  \begin{aligned}
    F_2 
    = \,
    & 
    B_2 
    \\
    & 
    +
    E_0
    \wedge
    \mathrm{d}t \wedge \mathrm{d}v  
  \end{aligned}
  \;\;
  \text{for}
  \;\;
  \left\{
  \begin{aligned}
    B_2
    & 
    \in
    \Omega^2_{\mathrm{dR}}\big(X^{3}\big)
    \xhookrightarrow{ p_X^\ast }
    \Omega^2_{\mathrm{dR}}\big(X^{1,4}\big)
    \\
    E_0 
    & 
    \in
    \Omega^0_{\mathrm{dR}}\big(X^{3}\big)
    \xhookrightarrow{ p_X^\ast }
    \Omega^0_{\mathrm{dR}}\big(X^{1,4}\big)
    \mathrlap{\,,}
  \end{aligned}
  \right.
\end{equation}
\item and constituting an electric \emph{flux compactification}:
\begin{equation}
  \label{DimRedFluxCompactification}
  \underset{x \in X}{\forall} 
  \;
  E_0(x) \neq 0
  \mathrlap{\,.}
\end{equation}

\end{itemize}

We make the following simple but noteworthy observation:

\begin{proposition}
  \label{How5DMCSReducesTo3DCS}
  In this situation and in the limit $\ell_v \to 0$, the equations of motion \eqref{EoMof5dMCS} of 5D MCS are equivalent to those of 3D CS
  \eqref{ChernSimonsEoM}
  on the magnetic flux density over $X^3$:
\begin{equation}
  \label{5DMCSEoMdimReducedTo3D}
    B_2  = 0
    \mathrlap{\,.}
\end{equation}
\end{proposition}
\begin{proof}
First to note that
$$
  \begin{aligned}
  \star F_2
  =
  &
  \;
  \ell_v
  \,
  (\star_3 B_2) \wedge 
  \mathrm{d}t \wedge \mathrm{d}v
  \\
  &
  +
  \tfrac{1}{\ell_v}
  \star_3 E_0
  \,,
  \end{aligned}
$$
whence
$$
  \begin{aligned}
    \underset{\underset{\ell_v \to 0}{\longrightarrow}}{\lim}
    \;
    \mathrm{d}
    \star F_2
    & =
    \underset{\underset{\ell_v \to 0}{\longrightarrow}}{\lim}
    \Big(
    \ell_v \big(
      \mathrm{d}_3 \star_3 B_2
    \big)
    \wedge \mathrm{d}t \wedge \mathrm{d}v
    \;+\;
    \tfrac{1}{\ell_v}\
    \underbrace{
      \mathrm{d}_3 \star_3 E_0
    }_0
    \Big)
    \\
    & = 0
    \,.
  \end{aligned}
$$
Therefore the equation of motion 
in the limit is equivalent to 
$$
  \begin{aligned}
    \underset{\underset{\ell_v \to 0}{\longrightarrow}}{\lim}
    \big(
      F_2 \wedge F_2
      -
      \mathrm{d}\star F_2 
      =
      0
    \big)
    & \Leftrightarrow\;
    F_2 \wedge F_2  = 0
    \\
    &
    \Leftrightarrow\;
    \underbrace{B_2 \wedge B_2}_{= 0}
    \;+\;
    2 E_0 \, B_2 
      \wedge 
    \mathrm{d}t
      \wedge 
    \mathrm{d}v
    \;=\; 
    0
    \\
    &
    \underset{
      \scalebox{.7}{\eqref{DimRedFluxCompactification}}
    }{\Leftrightarrow}\;
    B_2 = 0
    \mathrlap{\,.}
  \end{aligned}
$$

\vspace{-6mm}
\end{proof}

\section{Completion by Proper Flux Quantization}
\label{CompletionViaProperFluxQuantization}

The previous \S\ref{5dMaxwellChernSimonsAndItsDimReduction} showed how 3D Chern-Simons theory (\S\ref{3DCSLagrangianAndLevel}) is, classically and locally, a dimensional reduction of 5D Maxwell-Chern-Simons theory. But such Maxwell-type gauge theories admit global completion of their field content by \emph{flux quantization laws}  \cite[p. 4]{SS24-Phase}\cite[Def. 2.6]{SS25-Flux}, which thereby should also provide a complete quantization of 3D Chern-Simons theory.

\subsection{Topological Quantum Observables on Quantized Flux}

We review the recent approach \cite{SS25-Flux, SS24-Phase, SS24-Obs} to produce complete (systematically well-defined) topological quantum observables from the input datum of a proper flux quantization law, specified to the example of $D=5$ Maxwell-Chern-Simons theory (\S\ref{5dMaxwellChernSimonsAndItsDimReduction}) --- which turns out to admit flux quantization in 2-Cohomotopy.

\subsubsection{Flux Quantization of 5D MCS in 2-Cohomotopy}

Consider on a globally hyperbolic $D=1+4$ spacetime
\begin{equation}
  X^{1,4} 
  \simeq
  \mathbb{R}^{1,0}
  \times
  X^4
\end{equation}
a gauge field (whose full nature is to be specified in a moment, but) whose flux density is a single differential 2-form 
\begin{equation}
  F_2 
  \in
  \Omega^2_{\mathrm{dR}}(X^{1,4})
  \mathrlap{\,,}
\end{equation}
and whose (higher Maxwell-type \cite[p. 4]{SS24-Phase}\cite[Def. 2.6]{SS25-Flux}) equations of motion are those of 5D Maxwell-Chern-Simons theory \eqref{EoMof5dMCS}:
\begin{equation}
  \label{The5DEquationsOfMotion}
  \begin{aligned}
    \mathrm{d}\, F_2
    & = 0
    \\
    \mathrm{d}\, \star_{{}_{1,4}}
    F_2 & =
    F_2 \wedge F_2
    \mathrlap{\,,}
  \end{aligned}
\end{equation}
where $\begin{tikzcd}\Omega^\bullet_{\mathrm{dR}}\big(X^{1,4}\big) \ar[r, "{ \star_{{}_{1,4}} }"] & \Omega^{5-\bullet}_{\mathrm{dR}}\big(X^{1,4}\big)\end{tikzcd}$ is the Hodge star operator on our spacetime.

Now fixing any Cauchy surface
\begin{equation}
  \begin{tikzcd}
    X^4 
    \ar[
      rr,
      hook,
      "{ 
        \iota
        :=
        (0, \mathrm{id}) 
      }"
    ]
    &&
    \mathbb{R}^{1,0} \times X^4
    \mathrlap{\,,}
  \end{tikzcd}
\end{equation}
the solutions to the equations of motion \eqref{The5DEquationsOfMotion} are 
\cite[Thm. 2.2, Rem. 2.4]{SS24-Phase} in natural bijection to closed $\mathfrak{l}S^2$-valued differential forms \eqref{SheafOfFlatLInfinityValuedForms}
on $X^4$:
\begin{equation}
  \label{SolutionSetOf5DMaxwellChernSimons}
  \mathrm{SolSet}_{X^4}
  \simeq
  \left\{
  \begin{aligned}
    B_2 & 
    \in \Omega^2_{\mathrm{dR}}\big(
      X^4
    \big)
    \\
    E_3 & \in \Omega^3_{\mathrm{dR}}\big(
      X^4
    \big)
  \end{aligned}
  \;\middle\vert\;
  \begin{aligned}
    \mathrm{d}\, B_2 & = 0
    \\
    \mathrm{d}\, E_3 & = B_2 \wedge B_2
  \end{aligned}
  \right\}
  \defneq
  \Omega^1_{\mathrm{dR}}\big(
    X^4
    ;
    \mathfrak{l}S^2
  \big)_{\mathrm{cl}}
  \mathrlap{\,,}
\end{equation}
where $\mathfrak{l}S^2$ denotes the real Whitehead $L_\infty$-algebra \cite[Prop. 5.11]{FSS23-Char}
of the 2-sphere, whose Chevalley-Eilenberg algebra is \cite[Ex. 5.3]{FSS23-Char}\cite[p. 21]{SS25-Flux}:
\begin{equation}
  \label{CEAlgebraOf2Sphere}
  \mathrm{CE}\big(
    \mathfrak{l}S^2
  \big)
  \simeq
  \mathbb{R}_{{}_{\mathrm{d}}}
  \left[
  \begin{aligned}
    b_2
    \\
    e_3
  \end{aligned}
  \right]
  \Big/
  \left(
  \begin{aligned}
    \mathrm{d}\, b_2 & = 0
    \\
    \mathrm{d}\, e_3 & = b_2 b_2
  \end{aligned}
  \right)
\end{equation}
the differential graded-commutative algebra over $\mathbb{R}$ which generated by $b_2$ in degree 2 and $e_3$ in degree 3, subject to the differential relations $\mathrm{d} b_2 = 0$ and $\mathrm{d} e_3 = b_2^2$.

To highlight here that the magnetic flux density $B_2$ is accompanied by an electric flux density $E_3$ whose Gauss law couples it back to the magnetic flux --- signalling that the tradition to apply ordinary Dirac flux quantization only to $B_2$ does not do justice to the full field content of the theory. Instead:

Instead, an admissible \emph{proper flux quantization law} for this situation is classified \cite[\S 3.2]{SS25-Flux} by a space $\hotype{A}$ (nilpotent and of rational finite type) whose real Whitehead $L_\infty$-algebra $\mathfrak{l}\hotype{A}$ \cite[Prop. 5.11]{FSS23-Char} is isomorphic to $\mathfrak{l}S^2$ \eqref{CEAlgebraOf2Sphere}, whence the flux quantization condition is that the class $[B_2, E_3]$ of the flux densities in non-abelian de Rham cohomology \cite[\S 6]{FSS23-Char} is in the image under the \emph{character map} $\mathrm{ch}_{\hotype{A}}$ \cite[\S IV]{FSS23-Char}
of a \emph{charge} $\chi$ in non-abelian $\Omega \hotype{A}$-cohomology \cite[\S 2]{FSS23-Char}, cf \eqref{TheNonabelianCharacterMap}:
\begin{equation}
  \begin{tikzcd}
    &[+7pt]&[-5pt]
    H^1\big(
      X^4;
      \Omega \hotype{A}
    \big)
    \ar[
      d,
      "{ \mathrm{ch}_{\hotype{A}} }"
    ]
    &[-30pt]
    \defneq
    &[-30pt]
      \pi_0
      ,
      \mathrm{Map}\big(
        X^4, \hotype{A}
      \big)
    \\
    \ast
    \ar[
      r,
      "{
        (B_2, E_3)
      }"{swap}
    ]
    \ar[
      urr,
      bend left=10,
      dashed,
      "{ \chi }"
    ]
    &
    \Omega^1_{\mathrm{dR}}\big(
      X^4;
      \mathfrak{l}\hotype{A}
    \big)_{\mathrm{cl}}
    \ar[
      r,
      "{ [-] }"{swap}
    ]
    &
    H^1_{\mathrm{dR}}\big(
      X^4;
      \mathfrak{l}\hotype{A}
    \big)
    &\defneq&
    \Omega^1_{\mathrm{dR}}\big(
      X^4;
      \mathfrak{l}\hotype{A}
    \big)_{\mathrm{cl}}\big/
      _{\mathrlap{\!\mathrm{cncrdnc}\,.}}
  \end{tikzcd}
\end{equation}

Here, we consider now the evident choice of the 2-sphere itself,
\begin{equation}
  \label{ChoiceOfFluxQuantization}
  \hotype{A} := S^2
  ,
\end{equation}
which is the classifying space of the extraordinary nonabelian cohomology theory called \emph{2-Cohomotopy} $\pi^2$ \cite[\S VII]{STHu59}\cite[Ex. 2.7]{FSS23-Char}\cite[Def. 3.1]{SS25-FQH}:
\begin{equation}
  \label{2Cohomotopy}
  \pi^2(X)
  :=
  \pi_0
  \,
  \mathrm{Map}\big(X, S^2\big)
  \,.
\end{equation}
Compare this to the traditional quantization of magnetic flux (only) in ordinary integral 2-cohomology
\begin{equation}
  \label{Ordinary2Cohomology}
  H^2(X;\mathbb{Z})
  \simeq
  \pi_0
  \,
  \mathrm{Map}\big(X, B^2 \mathbb{Z}\big)
\end{equation}
of which it is an unstable stage, 
\begin{equation}
  \begin{tikzcd}
    S^2 
      \simeq 
    \mathbb{C}P^1
    \ar[r, hook]
    &
    \mathbb{C}P^\infty
    \weakHomotopyEquivalence
    B^2 \mathbb{Z}
    \mathrlap{\,.}
  \end{tikzcd}
\end{equation}
It is this ``de-stabilization'' which makes the Hopf generator $e_3$ in \eqref{CEAlgebraOf2Sphere} appear, whose effect is to extend the flux quantization from the magnetic field $B_2$ to the electric field $E_3$ with its non-linear Bianchi identity \eqref{SolutionSetOf5DMaxwellChernSimons}.

The choice of flux quantization \eqref{ChoiceOfFluxQuantization} is the analog of the choice of a Lagrangian density, in that it is a choice of physical model. But in contrast to the Lagrangian density, the choice of flux quantization gives a completed \emph{global} definition the on-shell field content (cf. Fig. \ref{FromTheoryToObservables}) and hence of the phase space  (\S\ref{TheFluxQuantizedPhaseSpace} below). Moreover and of prime importance for the present purpose, it completely determines the topological quantum observables of the theory (\S\ref{TheTopologicalQuantumObservables} below).

\subsubsection{Flux Quantization of the monopole string in 5D MCS}

For completeness, we make the following side remark on boundary effects (for more on which see \cite{SS25-OrbiK}):

5D Supergravity accomodates a 1-brane species known as the ``monopole string'' or ``solitonic string'' \cite{Mizoguchi1998, Fujii2000} or ``L1-brane'' \cite[p. 3]{Howe1998}, whose worldvolume $\phi : \begin{tikzcd}[sep=small] \Sigma^{1,1} \ar[r] & X^{1,4} \end{tikzcd}$ carries a ``chiral'' 1-form flux density, $H_1$, satisfying the Biachi identity
\begin{equation}
  \label{BianchiIdentityOnBrane}
  \mathrm{d} H_1
  =
  \phi^\ast F_2
  \,.
\end{equation}
This may be thought of as the dimensional reduction of (the Bianchi identity for) the chiral 3-form on M5-branes in 11D SuGra (cf. \cite[(1)]{GSS25-M5}). In its further reduction to 3D Chern-Simons theory (\S\ref{5dMaxwellChernSimonsAndItsDimReduction}) it corresponds to the field strength of the abelian WZW model (the ``chiral boson'') constituting a boundary field theory for 3D Chern-Simons. 

The flux-quantiztion of such such brane-supported gauge fields inside a given background spacetime happens in \emph{twisted} nonabelian cohomology \cite[\S3]{FSS23-Char}, classified now by a \emph{fibration} $\begin{tikzcd}\hotype{B} \ar[r, ->>] & \hotype{A} \end{tikzcd}$ such that charges on (a compatible Cauchy surface $\Sigma^1$ of) the brane are homotopy classes of dashed maps making the following diagram commute:
\begin{equation}
  \left[
  \begin{tikzcd}[
    baseline={([yshift=-3pt]current bounding box.center)}
  ]
    \Sigma^1
    \ar[
      rr,
      dashed,
    ]
    \ar[
      d,
      "{ \phi }"
    ]
    &&
    \hotype{B}
    \ar[
      d,
      ->>
    ]
    \\
    X^{4}
    \ar[
      rr,
      "{ \chi }"
    ]
    &&
    \hotype{A}
  \end{tikzcd}
  \right]
  \;\in\;
  H^{\phi^\ast\chi}\big(
    \Sigma^{1}
    ;
    \hotype{B}
  \big)
  \mathrlap{\,,}
\end{equation}
and the compatibility condition now is that 
the brane Bianchi identity \eqref{BianchiIdentityOnBrane}
is the closure condition for differential forms with coefficients in the \emph{relative} Whitehead $L_\infty$-algebra $\mathfrak{l}_{\hotype{A}} \hotype{B}$ of the fibration. The minimal choice of classifying fibration now turns out to be the \emph{complex Hopf fibration} $h : \begin{tikzcd}[sep=small]S^3 \ar[r, ->>] & S^2\end{tikzcd}$ (this follows as in \cite[\S 3.7]{FSS20-H}\cite[\S 3]{FSS21-Hopf}, there shown for the M5-brane and the quaternionic Hopf fibration):
\begin{equation}
  \left\{
  \begin{aligned}
    \mathcolor{purple}{H_1} 
    & 
    \in \Omega^2_{\mathrm{dR}}\big(
      \Sigma^1
    \big)
    \\
    \mathcolor{darkblue}{B_2} 
    & 
    \in \Omega^2_{\mathrm{dR}}\big(
      X^4
    \big)
    \\
    \mathcolor{darkblue}{E_3} 
    & 
    \in \Omega^3_{\mathrm{dR}}\big(
      X^4
    \big)
  \end{aligned}
  \;\middle\vert\;
  \begin{aligned}
    \mathrm{d}\, \mathcolor{purple}{H_1} 
    & = \phi^\ast 
    \mathcolor{darkblue}{B_2}
    \\
    \mathrm{d}\, \mathcolor{darkblue}{B_2} & = 0
    \\
    \mathrm{d}\, \mathcolor{darkblue}{E_3} & = B_2 \wedge B_2
  \end{aligned}
  \right\}
  \defneq
  \left\{
  \begin{tikzcd}[
    baseline=(current bounding box.center),
    row sep=10pt
  ]
  \Sigma^1
  \ar[d, "{\phi}"]
  \ar[
    r, 
    dashed,
    "{ H_1 }"
  ]
  &
  \Omega^1_{\mathrm{dR}}\big(
    -
    ;
    \mathfrak{l}_{\mathcolor{darkblue}{S^2}}
    \mathcolor{purple}{S^3}
  \big)_{\mathrm{cl}}  
  \ar[
    d,
    "{ h_\ast }"
  ]
  \\
  X^4
  \ar[
    r, 
    dashed,
    "{
      (B_2, E_3)
    }"
  ]
  &
  \Omega^1_{\mathrm{dR}}\big(
    -
    ;
    \mathfrak{l} \mathcolor{darkblue}{S^2}
  \big)_{\mathrm{cl}}
  \end{tikzcd}
  \right\}
  \mathrlap{\,.}
\end{equation}

\subsubsection{The flux-quantized phase space}
\label{TheFluxQuantizedPhaseSpace}

Given the choice of flux-quantization \eqref{ChoiceOfFluxQuantization}, the completed \emph{flux quantized phase space}
\eqref{PhaseSpaceInAppendix}
of the gauge theory is \cite[\S 2.2]{SS24-Phase}\cite[\S 3.3]{SS25-Flux}
the smooth mapping $\infty$-stack (cf. \cite[p. 41]{FSS23-Char}\cite[\S 4]{SS25-Bun}\cite{FSS15-Stacky}) 
\begin{equation}
  \label{The2CohomotopicalPhaseSpace}
  \mathrm{PhsSp}
  :=
  \mathrm{Map}\big(
    X^4
    ,
    S^2_{\mathrm{diff}}
  \big)
\end{equation}
from the Cauchy surface
into the homotopy fiber product $S^2_{\mathrm{diff}}$
of the sheaf of solution sets \eqref{SolutionSetOf5DMaxwellChernSimons}with the \emph{shape} \eqref{ShapeModality}
of the classifying space \eqref{ChoiceOfFluxQuantization} over the stack of flux deformations \cite[Def. 9.3]{FSS23-Char}, cf. \eqref{DiffCohomologyPullback}:
\begin{equation}
  \begin{tikzcd}
    S^2_{\mathrm{diff}}
    \ar[
      rr,
      "{ \chi }"
    ]
    \ar[
      d,
      shorten <=-3pt,
      "{ 
        (B_2,E_3) 
      }"{swap, pos=.4}
    ]
    \ar[
      dr,
      phantom,
      "{ \lrcorner }"{pos=.1}
    ]
    &&
    \shape S^2
    \ar[
      d,
      "{
        \mathbf{ch}_{S^2}
      }"
    ]
    \\
    \Omega^1_{\mathrm{dR}}(
      -;\mathfrak{l}S^2
    )_{\mathrm{cl}}
    \ar[
      rr,
      "{
        \eta^{\scalebox{.7}{$\shape$}}
      }"{description}
    ]
    &{}&
    \shape
    \Omega^1_{\mathrm{dR}}(
      -;\mathfrak{l}S^2
    )_{\mathrm{cl}}
    \mathrlap{\,,}
  \end{tikzcd}
\end{equation}
this being the smooth moduli $\infty$-stack of \emph{differential 2-Cohomotopy} \cite[Ex. 9.3]{FSS23-Char}\cite{FSS15-M5WZW}\cite[\S 3.1]{Grady2021}. It is the homotopy filling this square which encodes the gauge potentials (cf. \cite[\S 2.1.4]{GSS24-SuGra}\cite[\S B]{SS25-Srni}) and thereby completes the higher gauge field content.

But here we are interested only in the topological sector of the theory, hence only in the shape of the phase space, for which the gauge potentials and flux densities drop out and only the charge sector remains:
\footnote{
  In \eqref{ShapeOfPhaseSpace}, the first step is by the \emph{smooth Oka principle} \eqref{SmoothOkaPrinciple}  and the second step uses that the shape modality $\shape$ preserves homotopy fiber products over shapes \eqref{ShapeOfFiberProductsOverShapes}  and is idempotent \eqref{IdempotencyOfShape}.
}
\begin{equation}
  \label{ShapeOfPhaseSpace}
  \begin{aligned}
    \shape \mathrm{PhsSp}
    & \defneq
    \shape 
    \mathrm{Map}\big(
      X^4
      ,
      S^2_{\mathrm{diff}}
    \big)
    \\
    & \simeq
    \mathrm{Map}\big(
      \shape X^4
      ,
      \shape S^2_{\mathrm{diff}}
    \big)    
    \\
    & 
    \simeq
    \mathrm{Map}\big(
      \shape X^4
      ,
      \shape S^2
    \big)        
    \\
    & \simeq
    \shape 
    \mathrm{Map}\big(
      X^4
      ,
      S^2
    \big).
  \end{aligned}
\end{equation}

In fact, since $X^4$ may be non-compact while we are interested in the \emph{solitonic}
phase space $\mathrm{PhsSp}^{\mathrm{solit}}$ \eqref{TheSolitonicPhaseSpace}
involving 
solitonic fluxes \eqref{FluxesVanishingAtInfinity}
which \emph{vanish at infinity} \cite[\S 2.2]{SS25-Flux}, the end result of the correspondingly adjusted version of this derivation yields the pointed mapping space out of the one-point compactification $(-)_{\cpt}$ of the space domain (by adjoining the \emph{point at infinity}):
\begin{equation}
  \label{ShapeOfSolitonicPhaseSpace}
  \shape \mathrm{PhsSp}
    ^{\mathrm{solit}}
  \weakHomotopyEquivalence
  \shape 
  \,
  \mathrm{Map}^\ast\big(
    X^4_{\cpt}
    ,
    S^2
  \big)
  \,.
\end{equation}

This is the first remarkable consequence of proper flux quantization of gauge fields: It completely determines the shape (homotopy type) of the global phase space of the field theory and thereby completely determines the ``topology'' of the solitonic field configurations. It's exactly this that the topological quantum observables reflect the quantum dynamics of (\S\ref{TheTopologicalQuantumObservables} below).

\subsubsection{The Topological Quantum Observables}
\label{TheTopologicalQuantumObservables}

The \emph{observables} of a physical system are (compactly supported, in experimental practice) smooth functions on its phase space (cf. \cite[\S 4]{Simms1976})
\begin{equation}
  \label{ObservableAsAMap}
  O 
   : 
  \begin{tikzcd}
    \mathrm{PhsSpc}
    \ar[r]
    &
    \mathbb{C}
    \,.
  \end{tikzcd}
\end{equation}
Since $\mathbb{C}$ is an ordinary manifold and hence \emph{concrete} as an $\infty$-stack (\cite[Def. 4.1.8, Ex. 4.1.19]{SS26-Orb}) these observables \eqref{ObservableAsAMap} are automatically gauge invariant (in that these maps factor \cite[Prop. 4.1.9]{SS26-Orb} through the concretification of the 0-truncation of phase space, hence in particular through the gauge equivalence classes of fields).

The \emph{topological observables} on the system, hence those insensitive to local variation of fields, are thus \emph{locally constant} such functions, in that they factor furthermore through the shape \eqref{ShapeOfPhaseSpace} of the phase space 
\begin{equation}
  O_{\mathrm{top}}
  :
  \begin{tikzcd}
    \mathrm{PhsSpc}
    \ar[
      r,
      "{ \eta^{\scalebox{.7}{$\shape$}} }"
    ]
    &
    \shape \, \mathrm{PhsSpc}
    \ar[r]
    &
    \mathbb{C}
    \mathrlap{\,,}
  \end{tikzcd}
\end{equation}
and hence (if taken to be compactly supported) form its 0th homology:
\begin{equation}
  \mathrm{Obs}_{\mathrm{top}}
  \:=
  C^\infty\big(
    \pi_0 \shape \mathrm{PhsSp}
    ,
    \mathbb{C}
  \big)_{\mathrm{cpt}}
  \;\simeq\;
  \mathbb{C}\big[
    \pi_0 \shape \mathrm{PhsSp}
  \big]
  \;\simeq\;
  H_0\big(
    \shape \mathrm{PhsSp}
    ;
    \mathbb{C}
  \big)
  \,.
\end{equation}
Now, if spacetime $X^{1,4}$ has an ``M-fiber'', in that it factors as (cf. Fig. \ref{SpacetimeDomain})
\begin{equation}
  \label{TheMTheoryDirection}
  X^4 
  \simeq
  X^3 \times \mathbb{R}^1
  \mathrlap{\,,}
\end{equation}
so that 
\begin{equation}
  \begin{aligned}
  \label{TheCompactifiedMTheoryDirection}
  X^4_{\cpt}
  & \simeq
  X^3_{\cpt} \wedge \mathbb{R}^1_{\cpt}
  \\
  & \simeq
  X^3_{\cpt} \wedge S^1 
  \,,
  \end{aligned}
\end{equation}
then these observables inherit a non-abelian \emph{Pontrjagin algebra} structure 
\cite{nLab:PontrjaginRing}
from convolution along the M-theory circle, being the group $C^\ast$-algebra algebra 
\begin{equation}
  \label{AlgebraOfQuantumObservables}
  \begin{aligned}
  \mathrm{Obs}[X^3]_{\mathrm{top}}
  & \simeq
  \mathbb{C}\Big[
    \pi_0  
    \mathrm{Map}^\ast\big(
      \mathbb{R}^1_{\cpt} \wedge  X^3_{\cpt} 
      ,
      \hotype{A}
    \big)
  \Big]
  \\
  & 
  \simeq
  \mathbb{C}\Big[
    \pi_1  
    \mathrm{Map}^\ast\big(
      X^3_{\cpt} 
      ,
      \hotype{A}
    \big)
  \Big]  
  \,,
  \end{aligned}
\end{equation}
 of the fundamental group of the space of pointed maps $X^3_{\cpt} \xrightarrow{} S^2$. As such, $\mathrm{Obs}[X^3]_{\mathrm{top}}$ is the \emph{algebra of topological quantum observables} on the system \cite{SS24-Obs}.

\subsubsection{The 2D TQFT on the M-sheet}
\label{The2DTQFTOnTheMFiber}

In terms of topological quantum field theory (TQFT, cf. \cite{Cohen2008, nLab:TQFT}) we may naturally understand the algebra \eqref{AlgebraOfQuantumObservables} of topological quantum observables as follows (cf. Figs. \ref{HCFT} and \ref{SpacetimeDomain}): 

\begin{figure}[htb]
\caption{
  \label{HCFT}
  The topological dynamics (of gauge fields flux-quantized in $\hotype{A}$-cohomology, \S\ref{The2DTQFTOnTheMFiber}) on a 5D spacetime with an ``M-theoretic'' 5th dimension, $X^{1,4} \simeq \mathbb{R}^{1,0} \times \mathbb{R}^1_M  \times X^3$, yields a $D$=2 genus=0 open TQFT on the ``M-sheet'' $\mathbb{R}^{1,0} \times \mathbb{R}^1_M$, which is the ``HCFT'' induced by \emph{open string topology operations} over the moduli space of topological charges on $X^3$. 
  Here the solitonic nature of the topological charges, namely their \emph{vanishing at infinity}, entails that the spatial ends $\infty$ of the ``M-sheet'' $\mathbb{R}^{1,0} \times \mathbb{R}^1_M$ are stuck on the ``0-brane'' locus $\infty \mapsto \{\ast\} \hookrightarrow \mathrm{Map}^\ast\big(X^3_{\cup\{\infty\}},\hotype{A}\big)$ in the moduli space, which reduces the string product to the Pontrjagin product (cf. \cite[Fig. 11]{SS25-Srni}).
}

\centering

\adjustbox{
  rndfbox=5pt
}{
$
  \begin{aligned}
  \grayunderbrace{
    \mathrm{TQFT}_{1+1}[X^3]
  }{
    \mathclap{
    \substack{
      \scalebox{.7}{5D topological dynamics}
      \\
      \scalebox{.7}{seen on the M-sheet}
    }
    }
  }
  :
  \scalebox{1}{M-fiber}
  &
  \;\;\;\longmapsto\;\;\;
  \grayunderbrace{
  H_0\Big(
  \grayoverbrace{
    \Omega_\ast
    \mathclap{\phantom{X^3_{\cpt}}}
  }{
    \mathclap{
    \substack{
      \scalebox{.7}{Open}
      \\
      \scalebox{.7}{string}
      \\
      \scalebox{.7}{topology}
    }
    }
  }
  \grayoverbrace{
  \mathrm{Map}^\ast\big(
    X^3_{\cpt}
    ,
    \hotype{A}
  \big)
  }{
    \substack{
      \scalebox{.7}{Moduli space of charges}
      \\
      \scalebox{.7}{on $X^3$ quantized}
      \\
      \scalebox{.7}{in $\hotype{A}$-cohomology}
    }
  }
  ;
  \mathbb{C}
  \Big)
  }{
    \mathrm{Obs}[X^3]_{\mathrm{top}}
  }
  \\
\begin{tikzpicture}[
  baseline=(current bounding box.center)
]

\def\width{2}

\draw[
  line width=.5,
  draw=black!70,
  fill=gray!30
] 
  (-\width, 0) 
    to[bend left=20]
  (+\width,0) 
    to[bend left=18]
  (0, -1)
    to[bend left=18]
  (-\width,0); 

\fill[white] (-\width,0) circle (.15);
\fill[white] (+\width,0) circle (.15);
\fill[white] (0,-1) circle (.15);
\node[scale=1] at (-\width,0) {$\infty$};
\node[scale=1] at (+\width,0) {$\infty$};
\node[scale=1] at (0,-1) {$\infty$};

\begin{scope}[
  shift={(-.3-.05,-.48-.05)}
]
\draw[
  -Stealth,
  gray
] 
  (-.2,0) 
    to[
      "$\mathbb{R}^1_M$"{
        near end,
        scale=.7
      },
      "{ M-fiber }"{
        swap,
        scale=.7
      }
    ]
  (1.3,0);
\draw[
  -Stealth,
  gray
] 
  (0,-.2) 
    to[
      "{ time }"{
        scale=.7
      },
      "{
        $\mathbb{R}^{1,0}$
      }"{
        swap,
        near end,
        scale=.7
      }
    ]
  (0,.8);
\end{scope}
\end{tikzpicture}
&
\;\;\;\longmapsto\;\;\;
\begin{tikzcd}
  \mathrm{Obs}[X^3]_{\mathrm{top}}
  \\
  \mathrm{Obs}[X^3]_{\mathrm{top}}
  \otimes
  \mathrm{Obs}[X^3]_{\mathrm{top}}
  \ar[
    u,
    "{\color{darkgreen} 
      \scalebox{.7}{
        Pontrjagin
      }
    }",
    "{\color{darkgreen} 
      \scalebox{.7}{
        product
      }
    }"{swap},
  ]
\end{tikzcd}
  \end{aligned}
$
}

\end{figure}
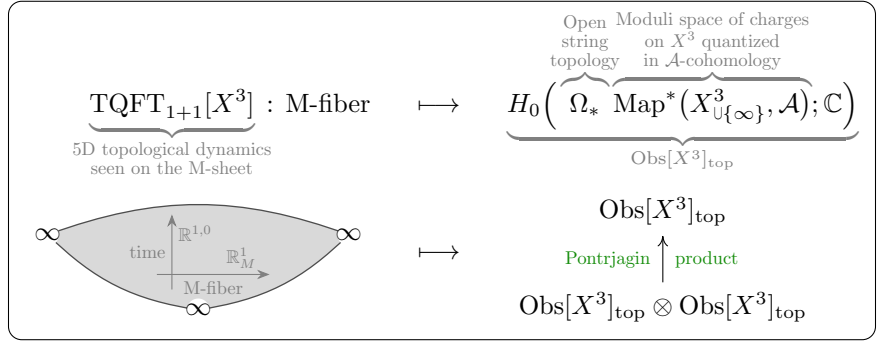

By regarding the above spacetime factorization \eqref{TheCompactifiedMTheoryDirection} as KK-compactification \emph{on} $X^3$, we are to be left with a $D=1+1$ TQFT of an ``open string''-like M-fiber propagating in the moduli space of topological fluxes on $X^3$, with endpoints attached there to the ``0-brane'' corresponding to the vanishing flux configuration.

Such topological open string propagating in topological target spaces are naturally quantized by \emph{open string topology HCFT}  \cite{Sullivan2005, Godin2007, Kupers2011},  and for the case of such open strings attached to a 0-brane, the string topology operation coincides \cite[Ex. 3]{Sullivan2005} 
\footnote{
  To see this in detail, use \cite[p 136-7]{Kupers2011}: 
  In the notation there, for $A = B = C := \{\ast\}$ a zero brane, the map ``$M^i$'' on p. 136 there is an isomorphism, and the map ``$M^j$'' (beware the typo in the arrow orientation there) is loop concatenation, whence the claim that the string product coincides with the Pontrjagin product follows from the formula for $\mu_{a b c}$ on p. 137 there. 
  
  To note here that the string product in general exists only for strings propagating in closed manifolds (or at least in Poincar{\'e} duality spaces, such as classifying spaces of compact Lie groups), but that the Pontrjagin product exists for string propagation in any topological space, such as in our case the moduli space $\mathrm{Map}^\ast\big(X^3_{\cpt}, \hotype{A}\big)$ of charges on $X^3$. This is why we manifestly only get an open genus=0 HCFT here, where the worldsheets are all glued from those shown in Fig. \ref{HCFT}.
}
with the Pontrjagin algebra structure \eqref{AlgebraOfQuantumObservables}, 
cf. again Fig. \ref{HCFT}.

\subsubsection{The homogeneous and microscopic observables}

With the group algebra of $\pi_1$ of  the moduli space being the topological observables $\mathrm{Obs}[X^3]_{\mathrm{top}}$ \eqref{SurfaceQuantumObservablesRecalled}, the fundamental group itself is the set of \emph{homogeneous} topological observables, $\mathrm{Obs}[X^3]_{\mathrm{top}}^{\mathrm{homog}}$ (those not arising as a nontrivial sum of basis elements), and the actual loops are the \emph{microscopic} (as opposed to topological) homogeneous observables:
\begin{equation}
  \label{HomogeneousAndMicroscopicObservables}
  \begin{tikzcd}[
    row sep=-2pt, 
    column sep=15pt
  ]
    \mathrm{Obs}[X^3]
      _{\mathrm{micro}}
      ^{\mathrm{homog}}
    \ar[
      r,
      ->>
    ]
    &
    \mathrm{Obs}[X^3]
      _{\mathrm{top}}
      ^{\mathrm{homog}}
    \ar[
      r,
      hook
    ]
    &
    \mathrm{Obs}[X^3]
      _{\mathrm{top}}
    \\
    \rotatebox[origin=c]{-90}{$:=$}
    &
    \rotatebox[origin=c]{-90}{$:=$}
    &
    \rotatebox[origin=c]{-90}{$:=$}
    \\
    \Omega
    \,
    \mathrm{Map}^\ast\big(
      X^3_{\cpt}
      ,
      \hotype{A}
    \big)  
    \ar[
      r,
      "{ [-] }"
    ]
    &
    \pi_1
    \,
    \mathrm{Map}^\ast\big(
      X^3_{\cpt}
      ,
      \hotype{A}
    \big)  
    \ar[
      r,
      hook
    ]
    &
    \mathbb{C}\Big[
    \pi_1
    \,
    \mathrm{Map}^\ast\big(
      X^3_{\cpt}
      ,
      \hotype{A}
    \big)  
    \Big]
    \mathrlap{\,.}
  \end{tikzcd}
\end{equation}

\subsubsection{The Topological Dimensional Reduction}
\label{TheTopologicalDimensionalReduction}

We are finally interested in effective reduction to 1+2 dimensions, whence we furthermore consider the situation (cf. Fig. \ref{SpacetimeDomain}) that 3D space is a cylinder over surface $\Sigma^2$,
\begin{equation}
  X^3 
  \simeq
  \Sigma^2
  \times
  [0,1]
  \,.
\end{equation}

\begin{figure}[htb]
\caption{
  \label{SpacetimeDomain}
  The topology of the 5D spacetime considered in \S\ref{TheTopologicalDimensionalReduction}. Reduction on the M-fiber brings it down to 1+3 dimensions and induces the Pontrjagin convolution algebra structure on the topological quantum observables \eqref{AlgebraOfQuantumObservables}, cf. Fig. \ref{HCFT}. Further reduction on the interval models the effective 2 dimensional spatial situation (e.g. of a 2D electron gas in a fractional quantum Hall system) in which one may expect Wilson line observables (of anyon braiding phases), cf. \S\ref{DimensonalReductionTo3D}.
}

\centering

\adjustbox{
  rndframe=4pt
}{
  \begin{tikzcd}[sep=0pt]
    &&&
    \mathclap{
      \grayoverbrace{\phantom{--}}{
        \adjustbox{scale=.8}{
          time
       }
      }
    }
    &&
    &&&
    \mathclap{
      \;\;
      \grayoverbrace{\phantom{---------}}{
        \adjustbox{scale=.8}{
          3d space
       }
      }
    }
    \\[-10pt]
    \scalebox{.7}{
      \def\arraystretch{.9}
      \begin{tabular}{c}
        spacetime
        \\
        domain
      \end{tabular}
    }
    &
    X^{1,4}
    & 
    \simeq
    &
    \mathbb{R}^{1,0}
    &\times&
    \mathbb{R}^1
    &\times&
    \Sigma^2
    &\times&
    {[0,1]}
    \\
    \scalebox{.7}{
      \def\arraystretch{.9}
      \begin{tabular}{c}
        space including 
        \\
        point-at-infinity
      \end{tabular}
    }
    &
    X^{4}_{\cpt}
    & 
    \simeq
    &
    &&
    S^1
    &\wedge&
    \Sigma^2_{\cpt}
    &\wedge&
    {[0,1]_{\plus}}
    \\[-15pt]
    &&&&&
    \mathclap{
      \grayunderbrace{\phantom{--}}{
        \adjustbox{scale=.8}{
          M-fiber
       }
      }
    }
  \end{tikzcd}
  \;\;
}
\end{figure}

Since the interval $[0,1]$ is contractible, it does not contribute in \eqref{AlgebraOfQuantumObservables}
and the final answer for the topological quantum observables in this situation is the $C^\ast$-algebra
\begin{equation}
  \label{QuantumObservablesOnSurface}
  \mathrm{Obs}[X^3]_{\mathrm{top}}
  \simeq
  \mathbb{C}\Big[
    \pi_1
    \mathrm{Map}^\ast\big(
      \Sigma^2_{\cpt}
      ,
      \hotype{A}
    \big)
  \Big]
  \,.
\end{equation}

This is the second remarkable consequence of proper flux quantization: That it completely determines the algebra of topological quantum observables.

Remarkably, we will see (in \S\ref{DerivingChernSimonsQuantumObservables} below) that for 2-cohomotopical flux quantization, $\hotype{A} := S^2$,
equation \eqref{QuantumObservablesOnSurface} knows everything about the traditionally expected 3D Chern-Simons quantum observables (\S\ref{TheKnotTheoreticExpression}). 
While this is just a fact, it may seem surprising given that we have found \eqref{QuantumObservablesOnSurface}, in this case, as the algebra of quantum observables of a completion of 5D Maxwell-Chern-Simons theory.

On the other hand, we have already seen in \S\ref{DimensonalReductionTo3D} that, classically, the latter does reduce to the former if the fiber components of the flux density is suitably constrained. Here we show that this constraint does not change the homotopy type (in physics jargon: ``the topology'') of the phase space \eqref{ShapeOfSolitonicPhaseSpace} and therefore does not change the topological quantum observables. 

This serves to explain ``why'' the quantum observables of an unconstrained completion of 5D Maxwell-Chern-Simons theory may subsume those of 3D Chern-Simons theory.

Concretely, with the \emph{smooth set} (0-truncated smooth stack, cf. \cite{Schreiber2025, GS25-FieldsI}) of solutions \eqref{SolutionSetOf5DMaxwellChernSimons},
denoted in boldface
\begin{equation}
  \mathbf{\Omega}^1_{\mathrm{dR}}\big(
    X^4;
    \mathfrak{l}S^2
  \big)_{\mathrm{cl}}
  :=
  \mathrm{Map}\Big(
    X^4
    ,\,
    \Omega^1_{\mathrm{dR}}\big(
      -;
      \mathfrak{l}S^2
    \big)_{\mathrm{cl}}
  \Big)
  \,,
\end{equation}
a constraint on the flux densities corresponds to a subobject 
\begin{equation}
  \label{ConstraintSubobject}
  \begin{tikzcd}
  \mathbf{\Omega}^1_{\mathrm{dR}}\big(
    X^4
    ;
    \mathfrak{l}S^2
  \big)_{\mathrm{cl}}^{\mathrm{constr}}
  \ar[r, hook]
  &
  \mathbf{\Omega}^1_{\mathrm{dR}}\big(
    X^4
    ;
    \mathfrak{l}S^2
  \big)_{\mathrm{cl}}
  \end{tikzcd}
\end{equation}
and the correspondingly constrained phase space $\mathrm{PhsSp}_{\mathrm{sol}}^{\mathrm{constr}}$ is the homotopy pullback of the plain phase space \eqref{ShapeOfSolitonicPhaseSpace}
along this inclusion:
\begin{equation}
  \label{ConstrainedPhaseSpace}
  \begin{tikzcd}
    \mathrm{PhsSp}
      ^{\mathrm{constr}}
    \ar[
      rr
    ]
    \ar[
      d
    ]
    \ar[
      dr,
      phantom,
      "{ \lrcorner }"{pos=.1}
    ]
    &&
    \mathrm{PhsSp}
    \ar[
      d,
      "{
        (B_2, E_3)_\ast
      }"
    ]
    \\
    \mathbf{\Omega}^1_{\mathrm{dR}}\big(
      X^4
      ;
      \mathfrak{l}S^2
    \big)
      _{\mathrm{cl}}
      ^{\mathrm{constr}}
    \ar[
      rr,
      hook
    ]
    &{}&
    \mathbf{\Omega}^1_{\mathrm{dR}}\big(
      X^4
      ;
      \mathfrak{l}S^2
    \big)
      _{\mathrm{cl}}
    \,.
  \end{tikzcd}
\end{equation}

In our case, 
\begin{equation}
  X^4 = 
  \Sigma^2 
    \times 
  \grayunderbrace{[0,1]}{V^1} 
    \times 
  \grayunderbrace{\mathclap{\phantom{[}}\mathbb{R}}{L^1}
\end{equation}
(Fig. \ref{SpacetimeDomain}) and the constraint is that enforcing equations
\eqref{DimRedDecompositionOfFluxDensity} and \eqref{DimRedFluxCompactification}.
While this constraint clearly changes the phase space \eqref{ConstrainedPhaseSpace}, we claim that it does not affect its shape, and hence not the topological quantum observables:
\begin{proposition}
  \label{ShapeOfConstrainedPhaseSpace}
  In the above situation, for $\Sigma^2 = \mathbb{R}^2$, the comparison map in \eqref{ConstrainedPhaseSpace} is a bijection on connected components:
  \begin{equation}
    \begin{tikzcd}
      \pi_0
      \shape
      \mathrm{PhsSp}^{\mathrm{constr}}
      \ar[
        r,
        "{ \weakHomotopyEquivalence }"
      ]
      &
      \pi_0
      \shape
      \mathrm{PhsSp}     \,. 
    \end{tikzcd}
  \end{equation}
\end{proposition}
\begin{proof}
  It is sufficient 
  (cf. Lem. \ref{IfConstrainedFluxIsEquivalenceThenSoIsConstrainedPhaseSpace})
  to see that this is already the case for the constraint subobject \eqref{ConstraintSubobject} of flux densities
  \begin{equation}
    \label{ShapeComparisonOfConstrained5DMCSFluxes}
    \begin{tikzcd}[
      sep=25pt
    ]
      \pi_0
      \shape
      \mathbf{\Omega}^1_{\mathrm{dR}}\big(
        \mathbb{R}^3_{\cpt} 
          \wedge 
        [0,1]_{\plus};
        \mathfrak{l}S^2
      \big)
        _{\mathrm{cl}}
        ^{\mathrm{constr}}
      \ar[
        r,
        "{ \weakHomotopyEquivalence }"
      ]
      &
      \pi_0
      \shape
      \mathbf{\Omega}^1_{\mathrm{dR}}\big(
        \mathbb{R}^3_{\cpt} 
        ;
        \mathfrak{l}S^2
      \big)
        _{\mathrm{cl}}
      \mathrlap{\,,}
    \end{tikzcd}
  \end{equation}
  if there is no higher homotopy on the right.
  This turns out to be the case, by Lem. \ref{ProvingShapeOfConstrainedFluxes}.
  We have compiled the details of the argument in \S\ref{ProofOfShapeOfConstrainedPhaseSpace}
\end{proof}

In consequence, Prop. \ref{ShapeOfConstrainedPhaseSpace} entails that the topological quantum observables 
\eqref{AlgebraOfQuantumObservables}
of 5D Maxwell-Chern-Simons theory, completed by proper flux quantization in 2-Cohomotopy, do not distinguish between solitonic field configurations that do or do not satisfy the flux-compactification constraints 
\eqref{DimRedDecompositionOfFluxDensity} and \eqref{DimRedFluxCompactification}.
But since we have seen, in Prop. \ref{How5DMCSReducesTo3DCS}, that classically the flux configurations that do satisfy these constraints dimensionally reduce to 3D Chern-Simons theory, this suggests that the topological quantum observables of the completed 5D theory also reduce to those of 3D Chern-Simons theory. 

That this is indeed the case is the topic of the next \S\ref{DerivingChernSimonsQuantumObservables}.

\subsection{Deriving Chern-Simons Quantum Observables}
\label{DerivingChernSimonsQuantumObservables}

We have thus constructed a completion of the topological sector of 5D Maxwell-Chern-Simons theory, whose topological quantum observables on a globally hyperbolic spacetime with an ``M''-direction over 3D space $X^3$ form the group algebra
\begin{equation}
  \label{QuantumObservablesRecalled}
  \mathrm{Obs}[X^3]_{\mathrm{top}}
  \simeq
  \mathbb{C}\Big[
    \pi_1
    \mathrm{Map}^\ast\big(
      X^3_{\cpt}
      ,
      S^2
    \big)
  \Big]
  \,,
\end{equation}
and we have shown that for 3D space being the cylinder over a surface $\Sigma^2$, in which case the formula reduces to 
\begin{equation}
  \label{SurfaceQuantumObservablesRecalled}
  X^3 = \Sigma^2 \times [0,1]
  \;\;\;\;\;
  \Rightarrow
  \;\;\;\;\;
  \mathrm{Obs}[X^3]_{\mathrm{top}}
  \simeq
  \mathbb{C}\Big[
    \pi_1
    \mathrm{Map}^\ast\big(
      \Sigma^2_{\cpt}
      ,
      S^2
    \big)
  \Big]
  \,,
\end{equation}
this is also a completed quantization of 3D abelian Chern-Simons theory on $(\mathbb{R}^{1} \times \Sigma_2)_{\cpt}$.

In any case, it is now a matter of direct algebro-topological analysis (no further choices involved) to compute \eqref{SurfaceQuantumObservablesRecalled}. This has been done in \cite{SS25-AbelianAnyons, SS25-FQH} and we review these results now. They happen to reproduce expected Chern-Simons observables in fine detail.

\subsubsection{Flux Solitons from the Pontrjagin Theorem}
\label{FluxSolitonsFromThePontrjaginTheorem}

In order to understand the evaluation of \eqref{SurfaceQuantumObservablesRecalled}, hence the fundamental group $\pi_!$ of the moduli space, it is helpful to first consider the situation on $\pi_0$. Here
\begin{equation}
  \label{Reduced2Cohomotopy}
  \widetilde \pi^2\big(\Sigma^2_{\cpt}\big)
  :=
  \pi_0\,
  \mathrm{Map}^\ast\big(
    \Sigma^2_{\cpt},
    S^2
  \big)
\end{equation}
is the \emph{2-Cohomotopy} (reduced, i.e., vanishing-at-infinity) of the surface $\Sigma^2$, in our context now understood as the set of possible \emph{charge sectors} of the flux quantized gauge theory. The (unstable) \emph{Pontrjagin theorem} (cf. \cite[\S II.16]{Bredon1993}\cite[p. 13 \& \S 3.2]{SS23-MF}\cite{nLab:PontrjaginTheorem}) identifies such charge sectors with (normally framed) cobordism classes of (codimension=2) submanifolds  of configurations of cores of field solitons that source such charge (cf. Fig. \ref{SolitonsViaPontrjaginTheorem}):
\begin{equation}
  \label{PontrjaginTheorem}
  \begin{tikzcd}
  \pi_0\,
  \mathrm{Map}^\ast\big(
    \Sigma^2_{\cpt},
    S^2
  \big) 
  \ar[
    rr,
    "{ \sim }",
    "{
      \scalebox{.7}{Pontrjagin}
    }"{swap}
  ]
  &&
  \left\{
  \substack{
    \scalebox{.7}{Cobordism classes of}
    \\
    \scalebox{.7}{signed configurations of}
    \\
    \scalebox{.7}{cores of field solitons}
  }
  \right\}
  \end{tikzcd}
\end{equation}

\begin{figure}[htb]
\caption{
  \label{SolitonsViaPontrjaginTheorem}
  The Pontrjagin theorem \eqref{PontrjaginTheorem} identifies charge  measured in 2-Cohomotopy with framed tubuluar neighbourhoods of codimension=2 submanifolds. Under flux quantization \eqref{The2CohomotopicalPhaseSpace}, these tubular neighbourhoods are the support of solitonic flux quanta.
}
\centering

 \adjustbox{
   rndfbox=5pt,
    scale=1
  }{
\begin{tikzpicture}
\begin{scope}[
  scale=.8,
  shift={(.7,-4.9)}
]

\draw[
  line width=.8,
  ->,
  darkgreen
]
  (1.5,1) 
  .. controls (2,1.6) and (3,2.6) .. 
  (6,1);

\node[
  scale=.7,
  rotate=-18
] at (4.7,1.8) {
  \color{darkblue}
  \bf
  \def\arraystretch{.85}
  \begin{tabular}{c}
    charge
    \\
    classifying map
    \\
    {}
  \end{tabular}
};

\node[
  scale=.7,
  rotate=-18
] at (4.5,1.4) {
  \color{darkblue}
  \bf
  $n$
};

  \shade[
    right color=gray, left color=lightgray,
    fill opacity=.9
  ]
    (3,-3)
      --
    (-1,-1)
      --
        (-1.21,1)
      --
    (2.3,3);

  \draw[dashed]
    (3,-3)
      --
    (-1,-1)
      --
    (-1.21,1)
      --
    (2.3,3)
      --
    (3,-3);

  \node[
    scale=1
  ] at (3.2,-2.1)
  {$\infty$};

  \begin{scope}[rotate=(+8)]
  \shadedraw[
    dashed,
    inner color=olive,
    outer color=lightolive,
  ]
    (1.5,-1)
    ellipse
    (.2 and .37);
  \draw
   (1.5,-1)
   to 
    node[above, yshift=-1pt]{
     \;\;\;\;\;\;\;\;\;\;\;
     \rotatebox[origin=c]{7}{
     \scalebox{.7}{
     \color{darkorange}
     \bf
       anyon
     }
     }
   }
    node[below, yshift=+6.3pt]{
     \;\;\;\;\;\;\;\;\;\;\;\;
     \rotatebox[origin=c]{7}{
     \scalebox{.7}{
     \color{darkorange}
     \bf
       worldline
     }
     }
   }
   (-2.2,-1);
  \draw
   (1.5+1.2,-1)
   to
   (4,-1);
  \end{scope}

  \begin{scope}[shift={(-.2,1.4)}, scale=(.96)]
  \begin{scope}[rotate=(+8)]
  \shadedraw[
    dashed,
    inner color=olive,
    outer color=lightolive,
  ]
    (1.5,-1)
    ellipse
    (.2 and .37);
  \draw
   (1.5,-1)
   to
   (-2.3,-1);
  \draw
   (1.5+1.35,-1)
   to
   (4.1,-1);
  \end{scope}
  \end{scope}
  \begin{scope}[shift={(-1,.5)}, scale=(.7)]
  \begin{scope}[rotate=(+8)]
  \shadedraw[
    dashed,
    inner color=olive,
    outer color=lightolive,
  ]
    (1.5,-1)
    ellipse
    (.2 and .32);
  \draw
   (1.5,-1)
   to
   (-1.8,-1);
  \end{scope}
  \end{scope}
  
\end{scope}

\node[
  scale=.73,
  rotate=-27
] at (2.21,-5.21) {
  \color{darkblue}
  \bf
  \def\arraystretch{.9}
  \begin{tabular}{l}
    flux
    \\
    $F_2 =$ 
    \\
    $\;\;n^\ast(\mathrm{dvol}_{S^2})$
  \end{tabular}
};

\node[
  scale=.73,
] at (2.05,-5) {
  \color{darkblue}
  \bf
};

\node[
  rotate=-140
] at (6,-4) {
  \includegraphics[width=2cm]{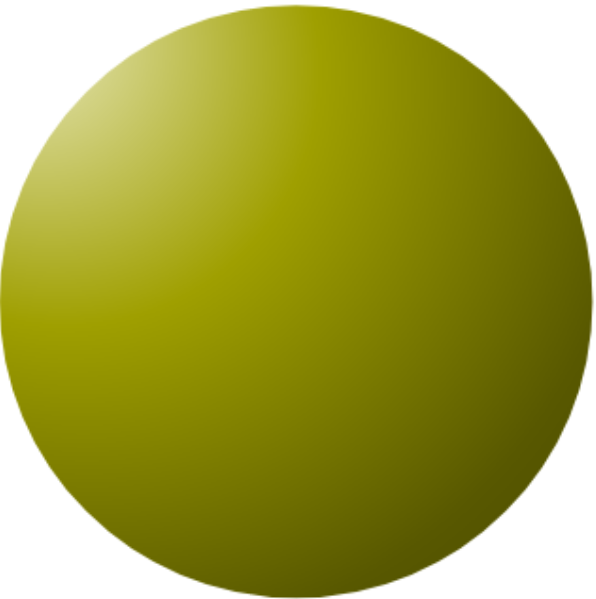}
};

\node[
  scale=.7
] at (6,-5.2) {
  \color{darkblue}
  \bf
  2-sphere $S^2$
};

\node[
  scale=.8,
  rotate=-22
] at (.05,-4.55)
{$\Sigma^2$};
 
\end{tikzpicture}
}

\end{figure}
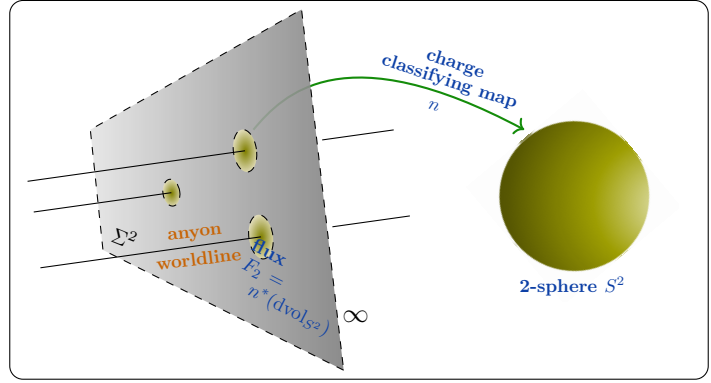

\subsubsection{Soliton Processes from the Segal-Okuyama Theorem}

From the Pontrjagin theorem (\S\ref{FluxSolitonsFromThePontrjaginTheorem}) it is already plausible that the elements of the fundamental group in \eqref{SurfaceQuantumObservablesRecalled}, are like loops (vacuum-to-vacuum processes) in the configuration space of flux solitons --- and indeed they will be recognized as the Wilson loops!
Remarkably, a precise analysis reveals that these Wilson loops emerge complete with the \emph{framing} data that traditionally is introduced in an \emph{ad hoc} fashion \eqref{FrameFieldOnLink}:

\begin{theorem}[{\cite{SS25-AbelianAnyons}}]
For solitonic flux on $\mathbb{R}^2$ quantized in 2-Cohomotopy:
\begin{enumerate} 
\item[\bf (i)]
the microscopic homogeneous quantum observables \eqref{HomogeneousAndMicroscopicObservables} are framed oriented links,
\item[\bf (ii)] 
the  topological pure quantum states $\vert K \rangle$ are labeled by $K \in \mathbb{R} \setminus 0$,
\item[\bf (iii)]
the expectation value of the former with respect to the latter are exactly those of renormalized Wilson loop in level $K/2$ Chern-Simons theory \eqref{WilsonLoopObservableFactorization},
\end{enumerate}
in that the following diagram commutes:
\begin{equation}
  \label{MainTheoremDiagram}
  \begin{tikzcd}[
    column sep=15pt,
    row sep=0pt
  ]
  \grayoverbrace{
  \Omega
  \,
  \mathrm{Map}^\ast\big(
    \mathbb{R}^2_{\cpt},
    S^2
  \big)
  }{
    \mathrm{Obs}[\mathbb{R}^2]
    ^{\mathrm{homog}}
    _{\mathrm{micro}}
  }
  \ar[
    r,
    ->>,
    "{
      [-]
    }"
  ]
  \ar[
    d,
    "{ \sim }"{sloped}
  ]
  &[5pt]
  \grayoverbrace{
  \pi_1\,
  \mathrm{Map}^\ast\big(
    \mathbb{R}^2_{\cpt},
    S^2
  \big)
  }{
    \mathrm{Obs}[\mathbb{R}^2]
      ^{\mathrm{homog}}
      _{\mathrm{top}}
  }
  \ar[
    r,
    "{
      \langle K \vert 
        -
      \vert K \rangle
    }"
  ]
  \ar[
    d,
    "{ \sim }"{sloped}
  ]
  &
  \mathbb{C}
  \ar[
    d, 
    equals
  ]
  \\[15pt]
  \mathrm{FrmdOrntdLnks}
  \ar[
    r,
    "{
      \mathrm{wrth}
    }"
  ]
  &
  \mathbb{Z}
  \ar[
    r,
    "{
      \exp\big(
        \tfrac
          {\pi \mathrm{i}}
          { K }
        (-)
      \big)
    }"
  ]
  &
  \mathbb{C}
  \\
    \gamma
    \ar[
      rr,
      |->, 
      shorten=15pt
    ]
    &&
   \Bigg\langle
      \!\!
      \exp\bigg(
        \mathrm{i}
        \textstyle{\underset{{(S^1)^n}}{\int}} 
        \!\!
        \gamma^\ast A
      \bigg)
    \! \Bigg\rangle .
  \end{tikzcd}
\end{equation}
\end{theorem}
\begin{proof}[Proof outline.]
  The first step is the homotopical (or ``dynamical'') refinement of Pontrjagin's theorem \eqref{PontrjaginTheorem} due to Segal \cite[Thm. 1]{Segal1973}, identifying the homotopy type of our moduli space of cohomotopical charges with the \emph{group completion}, $\mathbb{G}(-) := \Omega B_{\sqcup}(-)$, of the \emph{configuration space of points} $\mathrm{Conf}(-)$ (cf. \cite{williams2020, Kallel2025, SS22-Conf}), in the plane:
  \begin{equation}
    \mathrm{Map}^\ast\big(
      \mathbb{R}^2_{\cpt}
      ,
      S^2
    \big)
    \weakHomotopyEquivalence
    \mathbb{G}\mathrm{Conf}\big(\mathbb{R}^2\big)
    \,.
  \end{equation}
  The effect of this group completion on $\pi_0$ is just to equip the points in the configurations with signs, signifying the difference between solitons and anti-solitons as already captured by the framing in Pontrjagin's theorem \eqref{PontrjaginTheorem}. On the other hand, the higher homotopy effect on configurations of the group completion was understood only more recently by Okuyama \cite{Okuyama2005}, and \emph{here} is where the framing of links emerges: The group completed configuration space is equivalently a configuration space of  \emph{intervals}, $\mathrm{Conf}^I(-)$, all parallel to a fixed coordinate axis and equipped with signed endpoints:
  \begin{equation}
    \label{GroupCompletionViaIntervalConfiguration}
    \mathbb{G}\mathrm{Conf}\big(\mathbb{R}^2\big)
    \weakHomotopyEquivalence
    \mathrm{Conf}^I\big(\mathbb{R}^2\big)
  \end{equation}
  Here the topology of $\mathrm{Conf}^I(-)$ is such that the signed endpoints of the intervals may undergo pair creation/annihilation processes, as shown in Fig. \ref{PathsInIntervalConfigurationSpace}.

\begin{figure}[htb]
\caption{
  \label{PathsInIntervalConfigurationSpace}
  Continuous paths in the space $\mathrm{Conf}^I(-)$ of interval configurations \eqref{GroupCompletionViaIntervalConfiguration} include, besides the evident movement of endpoints of intervals, processes where oppositely charged endpoints coincide and annihilate (or conversely, running the path backwards: emerge and separate). 
}

\centering

\adjustbox{scale=.8}{
\adjustbox{
  scale=.8
}{

\begin{tikzpicture}[
  decoration=snake
]

\begin{scope}[shift={(2,-2)}]

\begin{scope}[shift={(-.2,0)}]
\draw[line width=2pt, gray]
  (0,0) -- (1,0);
\draw[line width=1pt,draw=gray,fill=white]
  (0,0) circle (.2);
\draw[line width=1pt,draw=gray,fill=white]
  (1,0) circle (.2);
\end{scope}

\begin{scope}[shift={(1.8,0)}]
\draw[line width=2pt, gray]
  (0,0) -- (1,0);
\draw[line width=1pt,draw=gray,fill=black]
  (0,0) circle (.2);
\draw[line width=1pt,draw=gray,fill=black]
  (1,0) circle (.2);
\end{scope}

\begin{scope}[shift={(1.6,-1.3)}]
\draw[line width=2pt, gray]
  (-.05,0) -- (1,0);
\draw[line width=1pt,draw=gray,fill=black]
  (-.2,0) circle (.2);
\draw[line width=1pt,draw=gray,fill=black]
  (1,0) circle (.2);
\end{scope}

\begin{scope}[shift={(0,-1.3)}]
\draw[line width=2pt, gray]
  (0,0) -- (1,0);
\draw[line width=1pt,draw=gray,fill=white]
  (0,0) circle (.2);
\draw[line width=1pt,draw=gray,fill=white, fill opacity=.5]
  (1.2,0) circle (.2);
\end{scope}

\draw[
  decorate,
  ->
] (1.3,-.3) -- (1.3,-1);

\draw[
  decorate,
  ->
] (1.3,-1.7) -- (1.3,-2.4);

\begin{scope}[shift={(0,-1.2)}]
\draw[
  decorate,
  ->
] (1.3,-1.7) -- (1.3,-2.4);
\end{scope}

\begin{scope}[shift={(0,-2.4)}]
\draw[
  decorate,
  ->
] (1.3,-1.7) -- (1.3,-2.4);
\end{scope}

\begin{scope}[shift={(0,-3.9)}]
\draw[
  decorate,
  ->
] (1.3,-1.7) -- (1.3,-2.4);
\end{scope}

\begin{scope}[shift={(0,-2.6)}]
\draw[line width=2pt, gray]
  (0,0) -- (2.6,0);
\draw[line width=1pt,draw=gray,fill=white]
  (0,0) circle (.2);
\draw[line width=1pt,draw=gray,fill=black]
  (2.6,0) circle (.2);
\end{scope}

\begin{scope}[shift={(0,-3.9)}]
\draw[line width=2pt, gray]
  (.6,0) -- (2,0);
\draw[line width=1pt,draw=gray,fill=white]
  (.6,0) circle (.2);
\draw[line width=1pt,draw=gray,fill=black]
  (2,0) circle (.2);
\end{scope}

\begin{scope}[shift={(.3,-5.2)}]
\draw[line width=1pt,draw=gray,fill=black]
  (1.07,0) circle (.2);
\draw[line width=1pt,draw=gray,fill=white, fill opacity=.5]
  (.93,0) circle (.2);
\end{scope}

\draw (1.3, -6.6) node {$\varnothing$};

\end{scope}
  
\end{tikzpicture}

}
\hspace{20pt}
\def\yrescale{1.55}
\adjustbox{
  scale=.85
}{
\begin{tikzpicture}[
  yscale=\yrescale
]

\clip 
  (1.7,-1.2) rectangle
  (5.25,-6.1);

\begin{scope}[
  shift={(2,-2)}
]

\begin{scope}[shift={(0,0)}]
\draw[line width=2pt, gray]
  (0,0) -- (1,0);
\draw[
  line width=1pt,
  draw=gray,
  fill=white,
  yscale=1/\yrescale
]
  (0,0) circle (.2);
\draw[
  line width=1pt,
  draw=gray,
  fill=white,
  yscale=1/\yrescale
]
  (1,0) circle (.2);
\end{scope}

\begin{scope}[shift={(2,0)}]
\draw[line width=2pt, gray]
  (0,0) -- (1,0);
\draw[
  line width=1pt,
  draw=gray,
  fill=black,
  yscale=1/\yrescale
]
  (0,0) circle (.2);
\draw[
  line width=1pt,
  draw=gray,
  fill=black,
  yscale=1/\yrescale
]
  (1,0) circle (.2);
\end{scope}

\begin{scope}[shift={(2,-.8)}]
\draw[line width=2pt, gray]
  (-.2,0) -- (1,0);
\draw[
  line width=1pt,
  draw=gray,
  fill=black,
  yscale=1/\yrescale
]
  (-.4,0) circle (.2);
\draw[
  line width=1pt,
  draw=gray,
  fill=black,
  yscale=1/\yrescale
]
  (1,0) circle (.2);
\end{scope}

\begin{scope}[shift={(0,-.8)}]
\draw[line width=2pt, gray]
  (0,0) -- (1.2,0);
\draw[
  line width=1pt,
  draw=gray,
  fill=white,
  yscale=1/\yrescale
]
  (0,0) circle (.2);
\draw[
  line width=1pt,
  draw=gray,fill=white, 
  fill opacity=.5,
  yscale=1/\yrescale
]
  (1.4,0) circle (.2);
\end{scope}

\begin{scope}[shift={(0,-1.6)}]
\draw[line width=2pt, gray]
  (0,0) -- (2.8,0);
\draw[
  line width=1pt,
  draw=gray,
  fill=white,
  yscale=1/\yrescale
]
  (0.2,0) circle (.2);
\draw[
  line width=1pt,
  draw=gray,
  fill=black,
  yscale=1/\yrescale
]
  (2.8,0) circle (.2);
\end{scope}

\begin{scope}[shift={(.15,-2.4)}]
\draw[line width=2pt, gray]
  (.5,0) -- (2,0);
\draw[
  line width=1pt,
  draw=gray,
  fill=white,
  yscale=1/\yrescale
]
  (.5,0) circle (.2);
\draw[
  line width=1pt,
  draw=gray,
  fill=black,
  yscale=1/\yrescale
]
  (2.2,0) circle (.2);
\end{scope}

\begin{scope}[shift={(.5,-3.2)}]
\draw[
  line width=1pt,
  draw=gray,
  fill=black,
  yscale=1/\yrescale
]
  (1.07,0) circle (.2);
\draw[
  line width=1pt,
  draw=gray,
  fill=white, 
  fill opacity=.5,
  yscale=1/\yrescale
]
  (.93,0) circle (.2);
\end{scope}

\draw (1.5, -4) node {$\varnothing$};

\end{scope}

\draw[
  gray,
  smooth,
  fill=gray,
  fill opacity=.3,
  draw opacity=.3
]
  plot 
  coordinates{
    (2,-1.06) 
    (2,-2.8)
    (2.2,-3.6)
    (2.65,-4.4)
    (3.5,-5.3)
    (4.35,-4.4)
    (4.8,-3.6)
    (5,-2.8)
    (5,-1.06)
  }
  -- (4.05,-1.06)
  plot 
  coordinates {
    (4.05,-1.06)
    (4,-2)
    (3.6, -2.8)
    (3.4, -2.8)
    (3,-2.1)
    (2.95,-1.06)
  }
  --(2,-1.06);

\begin{scope}[
  shift={(0,-5pt)}
]

\draw[
  line width=2pt,
  white
]
  (1.9,-1.5) -- (5.1,-1.5);
\draw[
  line width=2pt,
  white
]
  (1.9,-1.34) -- (5.1,-1.34);
\draw[
  line width=2pt,
  white
]
  (1.9,-1.18) -- (5.1,-1.18);

\end{scope}
 
\end{tikzpicture}
}
}

\end{figure}
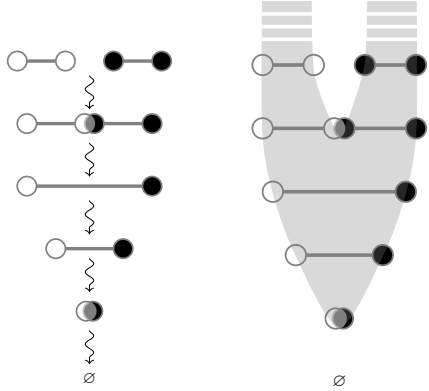
  
Thereby the loops in $\mathrm{Conf}^I\big(\mathbb{R}^2\big) \weakHomotopyEquivalence \mathrm{Map}^\ast\big(\mathbb{R}^2_{\cpt}, S^2 \big)$ are identified with framed links \cite[Prop. 2.15]{SS25-AbelianAnyons} (as illustrated in Fig. \ref{FramedLinksAsLoopsOfIntervals})! 

\begin{figure}[htb]
\caption{
  \label{FramedLinksAsLoopsOfIntervals}
  Loops 
  $\gamma \in \Omega \, \mathrm{Conf}^I\big(\mathbb{R}^2\big)$
  in the space \eqref{GroupCompletionViaIntervalConfiguration} of interval configurations (cf. Fig. \ref{PathsInIntervalConfigurationSpace}) are naturally identified with framed oriented link(diagram)s (cf. Fig. \ref{FramedLinksAndTheirWrithe})
}

\centering

\adjustbox{
  raise=-.5cm,
  margin=-5pt,
  scale=.6,
  rndfbox=5pt
}{
\def\tabcolsep{2pt}
\begin{tabular}{ccc}
\adjustbox{
  scale=.5
}{
\begin{tikzpicture}[
  baseline=(current bounding box.center)
]
 \begin{scope}
   \halfcircle{0,0};
 \end{scope}
 \begin{scope}[yscale=-1]
   \halfcircle{0,0};
 \end{scope} 
\end{tikzpicture}
}
&
\scalebox{1.6}{
$\leftrightarrow$
}
&
\adjustbox{}{
\begin{tikzpicture}[
  baseline=(current bounding box.center)
]
  \draw[
    line width=1.4,
    -Latex
  ]
    (0:1.6) arc (0:180:1.6);
  \draw[
    line width=1.4,
    -Latex
  ]
    (180:1.6) arc (180:360:1.6);
\end{tikzpicture}
}
\\
\adjustbox{
  scale=.5,
}{
\begin{tikzpicture}[
  baseline=(current bounding box.center)
]
\begin{scope}[yscale=-1]
  \halfcircle{0,0};
\end{scope}
\begin{scope}[
  shift={(3.4,.2)},
  xscale=-1
]
\halfcircle{0}{0};
\end{scope}
\begin{scope}[shift={(0,0)}]
\halfcircleover{0}{0};
\end{scope}
\begin{scope}[
  yscale=-1,
  shift={(3.4,-.2)},
  xscale=-1
]
\halfcircleover{0}{0};
\end{scope}
  
\end{tikzpicture}
}
&
\scalebox{1.6}{
$\leftrightarrow$
}
&
\adjustbox{}{
\begin{tikzpicture}[
  baseline=(current bounding box.center)
]
  \draw[
    line width=1.4,
    -Latex
  ]
    (0:1.6) arc (0:180:1.6);

\begin{scope}[shift={(2,0)}]
  \draw[
    line width=10,
    white
  ]
    (180:1.6) arc (180:0:1.6);
  \draw[
    line width=1.4,
    -Latex,
  ]
    (180:1.6) arc (180:0:1.6);
  \draw[
    line width=1.4,
    -Latex,
  ]
    (0:1.6) arc (0:-180:1.6);
\end{scope}
  
  \draw[
    line width=10,
    white
  ]
    (180:1.6) arc (180:360:1.6);
  \draw[
    line width=1.4,
    -Latex
  ]
    (180:1.6) arc (180:360:1.6);

\end{tikzpicture}
}
\\
\adjustbox{
  scale=.5
}{
\begin{tikzpicture}[
  baseline=(current bounding box.center)
]

\begin{scope}[yscale=-1]
  \halfcircle{0}{0}
\end{scope}
\begin{scope}[shift={(6,0)},scale=-1]
  \halfcircleover{0}{0}
\end{scope}
\begin{scope}[shift={(6,0)}, xscale=-1]
\halfcircle{0,0}
\end{scope}
\begin{scope}
\halfcircleover{0}{0}
\end{scope}

\end{tikzpicture}
}
&
\scalebox{1.6}{
$\leftrightarrow$
}
&
\adjustbox{}{
\begin{tikzpicture}[
  baseline=(current bounding box.center)
]
  \draw[
    line width=1.4,
    -Latex
  ]
    (0:1.6) arc (0:180:1.6);

\begin{scope}[shift={(3.2,0)}]
  \draw[
    line width=12,
    white
  ]
    (180:1.6) arc (180:0:1.6);
  \draw[
    line width=1.4,
    -Latex,
  ]
    (180:1.6) arc (180:0:1.6);
  \draw[
    line width=1.4,
    -Latex,
  ]
    (0:1.6) arc (0:-180:1.6);
\end{scope}
  
  \draw[
    line width=12,
    white
  ]
    (180:1.6) arc (180:360:1.6);
  \draw[
    line width=1.4,
    -Latex
  ]
    (180:1.6) arc (180:360:1.6);

\end{tikzpicture}
}
\\
\adjustbox{
  scale=.5
}{
\begin{tikzpicture}[
  baseline=(current bounding box.center)
]

\begin{scope}[
  shift={(.1,6)},
  yscale=-1
]
\halfcircle
\end{scope}

\begin{scope}[
  shift={(0,6)},
]
\clip (+4,-4) rectangle 
      (+0,+0);
\halfcircleAdjusted
\end{scope}

\begin{scope}[
  shift={(0,0)},
  yscale=-1,
  xscale=-1
]
\clip (+4,-4) rectangle 
      (+0,+0);
\halfcircleAdjusted
\end{scope}

\begin{scope}[xscale=-1]
\halfcircle
\end{scope}
\begin{scope}[
  xscale=-1,
  yscale=-1
]
\clip (-4,-4) rectangle 
      (+0,+0);
\halfcircleoverAdjusted
\end{scope}

\begin{scope}[
  shift={(.2,6)},
]
\clip (-4,-4) rectangle 
      (-.1,+0);
\halfcircleoverAdjusted
\end{scope}

\end{tikzpicture}
}
&
$\longmapsto$
&
\adjustbox{}{
\begin{tikzpicture}[
  rotate=-90,
  baseline=(current bounding box.center)
]
  \draw[
    line width=1.4,
    -Latex
  ]
    (0:1.6) arc (0:180:1.6);

\begin{scope}[shift={(3.2,0)}]
  \draw[
    line width=12,
    white
  ]
    (180:1.6) arc (180:0:1.6);
  \draw[
    line width=1.4,
    -Latex,
  ]
    (180:1.6) arc (180:0:1.6);
  \draw[
    line width=1.4,
    -Latex,
  ]
    (0:1.6) arc (0:-180:1.6);
\end{scope}
  
  \draw[
    line width=12,
    white
  ]
    (180:1.6) arc (180:360:1.6);
  \draw[
    line width=1.4,
    -Latex
  ]
    (180:1.6) arc (180:360:1.6);

\end{tikzpicture}
}
\\[-12pt]
&&
\end{tabular}
}

\end{figure}
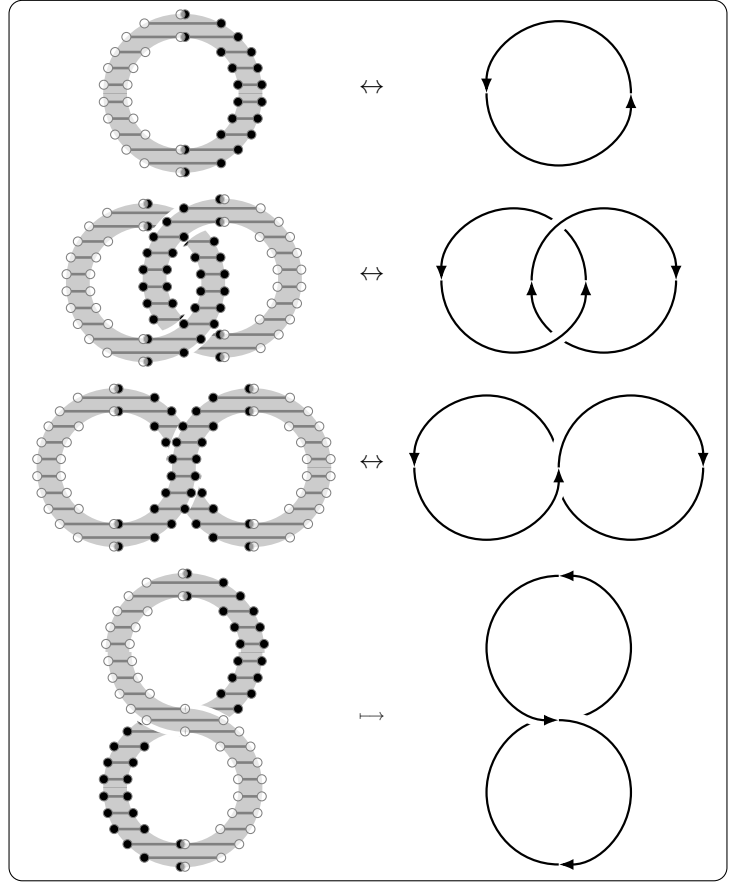

This shows the left part of the claim \eqref{MainTheoremDiagram}. The next step is to see that the connected components of the space of loops are labeled, under this identification, by the writhe of the corresponding framed links.

To see this, first observe that, under the identification just established, continuous deformations of loops of interval configurations (hence second order paths) induce \emph{cobordism} relations on framed links, as indicated in Fig. \ref{LinkCobordism}. 

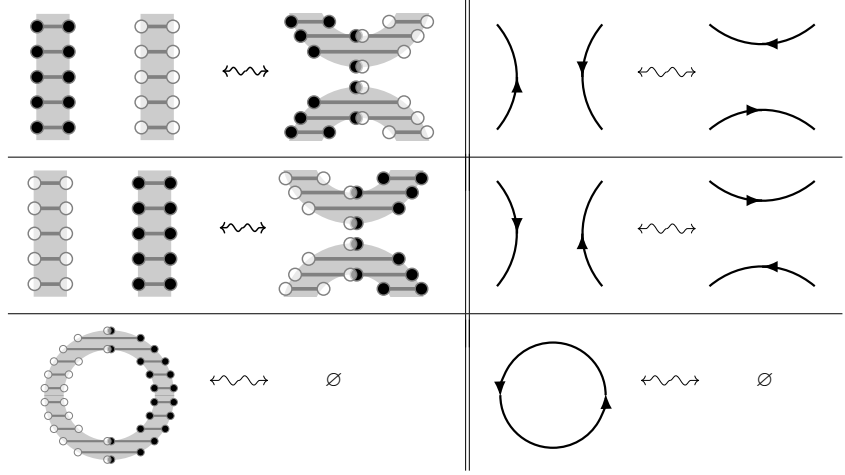
\begin{figure}[htb]
  \caption{
    \label{LinkCobordism}
    The continuous deformations (according to Fig. \ref{PathsInIntervalConfigurationSpace}) of framed links, regarded as loops of interval configurations (cf. Fig. \ref{FramedLinksAsLoopsOfIntervals}), include the \emph{saddle moves} and \emph{birth/death moves} of \emph{link cobordism}.
  }

\adjustbox{scale=.7}{
\def\tabcolsep{10pt}
\def\arraystretch{1.5}
\begin{tabular}{c||c}
\adjustbox{
  scale=.6,
}{
\begin{tikzpicture}[
  baseline=(current bounding box.center)
]
\draw[
  gray,
  line width=30pt,
  draw opacity=.4,
]
 (0,2) 
 --
 (0,-2);

\closedinterval{-.5}{1.6}{1}
\closedinterval{-.5}{.8}{1}
\closedinterval{-.5}{0}{1}
\closedinterval{-.5}{-.8}{1}
\closedinterval{-.5}{-1.6}{1}

\begin{scope}[
  xshift=3.3cm,
  xscale=-1
]
\draw[
  gray,
  line width=30pt,
  draw opacity=.4
]
 (0,2) 
 --
 (0,-2);

\openinterval{-.5}{1.6}{1}
\openinterval{-.5}{.8}{1}
\openinterval{-.5}{0}{1}
\openinterval{-.5}{-.8}{1}
\openinterval{-.5}{-1.6}{1}

\end{scope}
\end{tikzpicture}
}
\;\;\;\;
\adjustbox{scale=.8}{
\begin{tikzpicture}[decoration=snake]
\draw[decorate,->, line width=1]
  (0,0) -- (+0.55,0);
\draw[decorate,->, line width=1]
  (0,0) -- (-0.55,0);
\end{tikzpicture}
}
\adjustbox{
  scale=.6
}{
\begin{tikzpicture}[
  baseline=(current bounding box.center)
]

\begin{scope}
\clip
  (-.6,0) rectangle (5,2);
\draw[
  gray,
  line width=28,
  draw opacity=.4
]
  (1.75,2.6) circle (1.8);
\end{scope}
\closedinterval{-.4}{1.75}{1.2};
\closedinterval{-.1}{1.3}{1.75};
\begin{scope}[
  xshift=3.5cm,
  xscale=-1
]
\openinterval{-.4}{1.75}{1.2};
\openinterval{-.1}{1.3}{1.75};
\end{scope}
\oppositeinterval{.32}{.8}{2.85};
\oppositeinterval{1.65}{.33}{.2};

\begin{scope}[yscale=-1]
\begin{scope}
\clip
  (-.6,0) rectangle (5,2);
\draw[
  gray,
  line width=28,
  draw opacity=.4
]
  (1.75,2.6) circle (1.8);
\end{scope}
\closedinterval{-.4}{1.75}{1.2};
\closedinterval{-.1}{1.3}{1.75};
\begin{scope}[
  xshift=3.5cm,
  xscale=-1
]
\openinterval{-.4}{1.75}{1.2};
\openinterval{-.1}{1.3}{1.75};
\end{scope}
\oppositeinterval{.32}{.8}{2.85};
\oppositeinterval{1.65}{.33}{.2};
\end{scope}
  
\end{tikzpicture}
\hspace{-.7cm}
}
&
\adjustbox{}{
\begin{tikzpicture}[
  yscale=-1,
  baseline=(current bounding box.center)
]
\draw[
  line width=1.2,
  ]
  (-1,1)
  edge[bend left=40]
  (-1,-1);
\draw[-Latex]
  [line width=1.2]
  (-1+.38,+0)
  --
  (-1+.38,-.0001);

\begin{scope}[
  xscale=-1
]
\draw[
  line width=1.2,
  ]
  (-1,1)
  edge[bend left=40]
  (-1,-1);
\draw[-Latex]
  [line width=1.2]
  (-1+.38,0)
  --
  (-1+.38,+.0001);
\end{scope}
\end{tikzpicture}
}
\;\;\;
\begin{tikzpicture}[decoration=snake]
\draw[decorate,->]
  (0,0) -- (+.55,0);
\draw[decorate,->]
  (0,0) -- (-.55,0);
\end{tikzpicture}
\adjustbox{}{
\begin{tikzpicture}[
  xscale=-1,
  baseline=(current bounding box.center)
]
\begin{scope}[
  shift={(4.5,0)}
]
\draw[
  line width=1.2,
  ]
  (-1,1)
  edge[bend right=40]
  (+1,+1);
\draw[
  line width=1.2,
  -Latex
]
  (0,+.63)
  --
  (0.0001,+.63);

\begin{scope}[
  yscale=-1
]
\draw[
  line width=1.2,
  ]
  (-1,1)
  edge[bend right=40]
  (+1,+1);
\draw[
  line width=1.2,
  -Latex
]
  (+.0001,+.63)
  --
  (0,+.63);
\end{scope}

\end{scope}

\end{tikzpicture}
}
\\[30pt]
\hline
\\[-20pt]
\hspace{-.8cm}
\adjustbox{
  scale=.6
}{
\begin{tikzpicture}[
  baseline=(current bounding box.center)
]
\draw[
  gray,
  line width=30pt,
  draw opacity=.4,
]
 (0,2) 
 --
 (0,-2);

\openinterval{-.5}{1.6}{1}
\openinterval{-.5}{.8}{1}
\openinterval{-.5}{0}{1}
\openinterval{-.5}{-.8}{1}
\openinterval{-.5}{-1.6}{1}

\begin{scope}[
  xshift=3.3cm,
  xscale=-1
]
\draw[
  gray,
  line width=30pt,
  draw opacity=.4
]
 (0,2) 
 --
 (0,-2);

\closedinterval{-.5}{1.6}{1}
\closedinterval{-.5}{.8}{1}
\closedinterval{-.5}{0}{1}
\closedinterval{-.5}{-.8}{1}
\closedinterval{-.5}{-1.6}{1}

\end{scope}
\end{tikzpicture}
}
\;\;\;\;
\adjustbox{scale=.8}{
\begin{tikzpicture}[decoration=snake]
\draw[decorate,->, line width=1]
  (0,0) -- (+0.55,0);
\draw[decorate,->, line width=1]
  (0,0) -- (-0.55,0);
\end{tikzpicture}
}
\hspace{-.7cm}
\adjustbox{
  scale=.6,
}{
\begin{tikzpicture}[
  xscale=-1,
  baseline=(current bounding box.center)
]

\begin{scope}
\clip
  (-.6,0) rectangle (5,2);
\draw[
  gray,
  line width=28,
  draw opacity=.4
]
  (1.75,2.6) circle (1.8);
\end{scope}
\closedinterval{-.4}{1.75}{1.2};
\closedinterval{-.1}{1.3}{1.75};
\begin{scope}[
  xshift=3.5cm,
  xscale=-1
]
\openinterval{-.4}{1.75}{1.2};
\openinterval{-.1}{1.3}{1.75};
\end{scope}
\oppositeinterval{.32}{.8}{2.85};
\oppositeinterval{1.65}{.33}{.2};

\begin{scope}[yscale=-1]
\begin{scope}
\clip
  (-.6,0) rectangle (5,2);
\draw[
  gray,
  line width=28,
  draw opacity=.4
]
  (1.75,2.6) circle (1.8);
\end{scope}
\closedinterval{-.4}{1.75}{1.2};
\closedinterval{-.1}{1.3}{1.75};
\begin{scope}[
  xshift=3.5cm,
  xscale=-1
]
\openinterval{-.4}{1.75}{1.2};
\openinterval{-.1}{1.3}{1.75};
\end{scope}
\oppositeinterval{.32}{.8}{2.85};
\oppositeinterval{1.65}{.33}{.2};
\end{scope}
  
\end{tikzpicture}
\hspace{-.7cm}
}
&
\adjustbox{}{
\begin{tikzpicture}[
  baseline=(current bounding box.center)
]
\draw[
  line width=1.2,
  ]
  (-1,1)
  edge[bend left=40]
  (-1,-1);
\draw[-Latex]
  [line width=1.2]
  (-1+.38,+0)
  --
  (-1+.38,-.0001);

\begin{scope}[
  xscale=-1
]
\draw[
  line width=1.2,
  ]
  (-1,1)
  edge[bend left=40]
  (-1,-1);
\draw[-Latex]
  [line width=1.2]
  (-1+.38,0)
  --
  (-1+.38,+.0001);
\end{scope}
\end{tikzpicture}
}
\;\;\;
\begin{tikzpicture}[decoration=snake]
\draw[decorate,->]
  (0,0) -- (+.55,0);
\draw[decorate,->]
  (0,0) -- (-.55,0);
\end{tikzpicture}
\adjustbox{}{
\begin{tikzpicture}[
  baseline=(current bounding box.center)
]
\begin{scope}[
  shift={(4.5,0)}
]
\draw[
  line width=1.2,
  ]
  (-1,1)
  edge[bend right=40]
  (+1,+1);
\draw[
  line width=1.2,
  -Latex
]
  (0,+.63)
  --
  (0.0001,+.63);

\begin{scope}[
  yscale=-1
]
\draw[
  line width=1.2,
  ]
  (-1,1)
  edge[bend right=40]
  (+1,+1);
\draw[
  line width=1.2,
  -Latex
]
  (+.0001,+.63)
  --
  (0,+.63);
\end{scope}

\end{scope}

\end{tikzpicture}
}
\\
\hline
&
\\[-15pt]
\adjustbox{
  scale=.35
}{
\begin{tikzpicture}[
  baseline=(current bounding box.center)
]
  \halfcircle{0}{0}
  \begin{scope}[yscale=-1]
  \halfcircle{0}{0}
  \end{scope}
\end{tikzpicture}
}
\;\;
\begin{tikzpicture}[decoration=snake]
\draw[decorate,->]
  (0,0) -- (+.55,0);
\draw[decorate,->]
  (0,0) -- (-.55,0);
\node[scale=1.2] at (1.8,0) {$\varnothing$};
\end{tikzpicture}
\hspace{1.6cm}
&
\adjustbox{}{
\begin{tikzpicture}[
  baseline=(current bounding box.center)
]
  \draw[
    line width=1.2,
    -Latex
  ]
    (0:1) arc (0:182:1);
  \draw[
    line width=1.2,
    -Latex
  ]
    (180:1) arc (180:362:1);
\end{tikzpicture}
}
\;\;
\begin{tikzpicture}[decoration=snake]
\draw[decorate,->]
  (0,0) -- (+.55,0);
\draw[decorate,->]
  (0,0) -- (-.55,0);
\node[scale=1.2] at (1.8,0) {$\varnothing$};
\end{tikzpicture}
\hspace{.7cm}
\end{tabular}
}

\end{figure}

But under such \emph{saddle moves} each crossing of disconnected strands of a  framed link diagram may be transformed into an avoided crossing with one of the strands picking up a twist:
\begin{equation}
  \label{AvoidCrossing}
  \scalebox{0.65}{$
\adjustbox
{
  scale=.8,
}{
\begin{tikzpicture}[
  baseline=(current bounding box.center)
]

\draw[
  line width=1.2,
  -Latex
]
  (-1.6,-1.6) -- (+1.6,+1.6);

\draw[
  line width=15,
  white
]
  (+1.6,-1.6) -- (-1.6,+1.6);
\draw[
  line width=1.2,
  -Latex
]
  (+1.6,-1.6) -- (-1.6,+1.6);
  
\end{tikzpicture}}
\;\;\;\;\;\;
\begin{tikzpicture}[decoration=snake]
  \draw[decorate,->]
    (0,0) -- (+.55,0);
  \draw[decorate,->]
    (0,0) -- (-.55,0);
\end{tikzpicture}
\;\;\;\;\;\;
\adjustbox{
  scale=.8,
}
{
\begin{tikzpicture}[
  baseline=(current bounding box.center)
]

\draw[
  line width=1.2,
  -Latex
]
  (-1.7,-1.7) -- (0,0);

\draw[
  -Latex,
  line width=1.2,
]
  (+1.7,-1.7) -- (-1.7,1.7);

\draw[
  white,fill=white
]
  (-1.5,-1.5) rectangle (1.5,1.5);

\draw[line width=1.2]
  (-1.5,-1.5)
   .. controls
     (-.8,-.8) and (-.8,+.8) ..
  (-1.5,+1.5);

\begin{scope}[shift={(-.42,0)}]
\draw[
  line width=1.25
]
  (45:.3) arc (45:360-45:.3);
\end{scope}

\draw[line width=1.2]
  (+1.6,-1.6) -- (-.2,.2);
\draw[line width=1.2, -Latex]
  (+.2,+.2) -- (1.6,1.6);

\draw[white,fill=white]
  (-1.1,-.2) rectangle (-.5,+.2);
\clip
  (-1.1,-.2) rectangle (-.5,+.2);

\draw[line width=1.2]
  (-.82,.24) circle (.17);

\draw[line width=1.2]
  (-.81,-.24) circle (.18);
\end{tikzpicture}
}
\;\;\;\;\;\;
\begin{tikzpicture}[decoration=snake]
  \draw[decorate,->]
    (0,0) -- (+.55,0);
  \draw[decorate,->]
    (0,0) -- (-.55,0);
\end{tikzpicture}
\;\;\;\;\;\;
\adjustbox{
  scale=.8,
}
{
\begin{tikzpicture}[
  baseline=(current bounding box.center)
]

\draw[
  line width=1.2,
  -Latex
]
  (-1.7,-1.7) -- (0,0);

\draw[
  line width=1.2,
  -Latex, 
]
  (+1.7,-1.7) -- (-1.7,+1.7);

\draw[
  white,fill=white
]
  (-1.5,-1.5) rectangle (1.5,1.5);

\draw[line width=1.2]
  (-1.5,-1.5)
   .. controls
     (-.8,-.8) and (-.8,+.8) ..
  (-1.5,+1.5);

\begin{scope}[shift={(-.42,0)}]
\draw[
  line width=1.25
]
  (45:.3) arc (45:360-45:.3);
\end{scope}

\draw[line width=1.2]
  (+1.7,-1.7) -- (-.2,.2);
\draw[line width=1.2, -Latex]
  (+.2,+.2) -- (1.7,1.7);
 
\end{tikzpicture}
}
$}
\end{equation}

This way, each loop of interval configurations is continuously connected to a loop corresponding to a disjoint union of $\pm 1$-framed unknots, which are thereby seen to represent the generators of $\pi_1 \, \mathrm{Conf}^I\big(\mathbb{R}^2\big)$
\begin{equation}
\label{UnitFramedUnknotIsTheGenerator}
  \scalebox{0.7}{$
\left[
\adjustbox{
  scale=.3,
}
{
\begin{tikzpicture}[
  baseline=(current bounding box.center),
]

\begin{scope}[scale=-1]
\halfcircleover{0}{0}
\end{scope}

\draw[
  line width=43,
  white
]
  (-3,-.01) -- (-3,2.8);
\draw[
  gray,
  draw opacity=.4,
  line width=30
]
  (-3,-2.8) -- (-3,2.8);

\draw[
  gray,
  draw opacity=.4,
  line width=30
]
  (-9,-2.8) -- 
  (-9,2.8);

\begin{scope}[xscale=-1]
\halfcircleover{0}{0}
\end{scope}

\begin{scope}[
  shift={(-6,2.8)},
  yscale=-1
]
\halfcircleover{0}{0}
\end{scope}

\begin{scope}[
  shift={(-6,-2.8)}
]
\halfcircleover{0}{0}
\end{scope}

\foreach \n in {0,...,6} {
  \openinterval{-9.5}{-2.4+\n*.8}{1}
}

\foreach \n in {0,...,2} {
  \closedinterval{-3.5}{.5+\n*.9}{1}
}
  
\end{tikzpicture}
}
\right]
$}
\;\;
=
\;\;
1
\quad \in\;
\mathbb{Z}
\;\simeq\;
\pi_1\big(
  \mathrm{Conf}^I(\mathbb{R}^2)
\big)
\,.
\end{equation}

But since the saddle mode \eqref{AvoidCrossing} manifestly preserves the total crossing number \eqref{TotalLinkingNumber}, the net number of these framed unknots is equal to the crossing number of the initial framed link diagram and hence to the writhe of the framed link.

This proves the middle part of the claim \eqref{MainTheoremDiagram}.

To conclude, it remains to understand the pure quantum states on these observables. But since our algebra of topological quantum observables is commutative, $\mathrm{Obs}[X^3]_{\mathrm{top}} \simeq  \mathbb{C}[\mathbb{Z}]$, the expectation values of pure quantum states $\vert \psi \rangle$ are the star-algebra homomorphisms from the observables to the complex numbers
\cite[Lem. 3.1]{SS24-Obs}:
\begin{equation}
  \label{PureStatesAsStarHomomorphisms}
  \begin{tikzcd}[
    sep=0pt
  ]
    \mathbb{Z}
    \ar[
      rr, 
      hook
    ]
    &&
    \mathbb{C}\big[
      \mathbb{Z}
    \big]
    \ar[
      rr,
      "{
        \langle \psi \vert 
          - 
        \vert \psi \rangle      
      }"
    ]
    &&
    \mathbb{C}
    \\
    1 
    \ar[
      rrrr,
      |->,
      shorten=10pt
    ]
    &\phantom{-}&
    &\phantom{---}&
    e^{ \pi \mathrm{i}/K }
    \mathrlap{\,.}
  \end{tikzcd}
\end{equation}
As shown, as algebra homomorphisms there are clearly fixed by their value on $1 \in \mathbb{Z} \hookrightarrow \mathbb{C}[\mathbb{Z}]$, and star-homomorphy forces this complex number to be unimodular, hence of the form $\exp(\pi \mathrm{i}/K)$ for some $K \in \mathbb{R} \setminus \{0\}$.

This establishes the right part of \eqref{MainTheoremDiagram} and thereby completes the proof.
\end{proof}

We have thus established our claim that the traditionally renormalized Chern-Simons Wilson loop observables arise directly as the topological quantum observables of the completion of 3D CS theory by dimensional reduction of cohomotopically flux quantized 5d Maxwell-Chern-Simons theory. 

The only property missing at this point is the condition that the parameter $K/2$ in \eqref{PureStatesAsStarHomomorphisms} be (half) integral, as expected for (spin) Chern-Simons levels \eqref{HalfIntegralLevel}. This does follow, too, when extending the discussion here from $X^2 = \mathbb{R}^2$ to the torus $X^2 = T^2$, discussed in \cite[\S 3.4]{SS25-FQH}.

\section{Conclusion and Outlook}

While the idea of replacing Lagrangian quantum field theories (and their iterative piece-meal completion via processes of renormalization as well as of anomaly cancellation and of resummation) by (``UV-'')\emph{complete} quantum theories is widely advertised in discussion of fundamental physics, little technical attention has been devoted (as one can see from discussions like \cite{Crowther2019}) to the question of which guise such completed non-perturbative quantum theories might actually take.

In previous work 
\cite{SS24-Phase, SS25-Flux}
we have highlighted that at least one neglected form of completion is \emph{proper flux quantization} of non-linear higher gauge field sectors --- crucially taking into account not just the magnetic but also the electric fluxes and their generally non-linear Gauss laws -- and that from such indeed emerges \cite{SS24-Obs} the complete flux quantum observables at least in the topological sector.

Here we have laid out the example application of this program to the modest but instructive example of Wilson loop observables in abelian Chern-Simons (CS) theory. Beginning with a review of the fact that even for such an extremely well-studied QFT the traditional definition of its quantum observables is an \emph{ad hoc} choice, we highlighted that classical 3D CS is a constrained dimensional reduction of 5D Maxwell-Chern-Simons theory (MCS), and then considered the global completion of the latter by proper flux quantization in 2-Cohomotopy. This turns out to fully determine the topological quantum observables at once, without further choices involved, which turn out to reproduce the traditional Chern-Simons Wilson loop observables, including their otherwise \emph{ad hoc} choice of point-splitting renormalization by link framing data.  

In fact, beyond deriving the traditional renormalization choice for Wilson loops, our approach derives the Wilson loops themselves, in that it shows that these and nothing else are the topological quantum observables on $\mathbb{R}^3$. This may not sound surprising on the backdrop of tradition which has always assumed this to be the case, but it shows that completion by proper flux quantization gets to the heart of the nature of the gauge theory. Concretely, for our flux quantization in Cohomotopy it is the differential topology expressed by the seminal theorems of Pontrjagin and Segal which identifies flux quanta with framed points (abelian anyons, Fig. \ref{SolitonsViaPontrjaginTheorem}) and their vacuum-scattering processes with the framed links that ``are'' the Wilson loops (Fig. \ref{FramedLinksAsLoopsOfIntervals}).

This has non-trivial implications, such as in application to topological quantum materials known as \emph{fractional quantum Hall systems} (FQH \cite{Stormer99, nLab:FQH}): Elsewhere we have shown \cite{SS25-FQH} that over more general surfaces $\Sigma^2$ (hence with $X^3 = \Sigma^2 \times \mathbb{R}$) the cohomotopically completed topological quantum observables first of all recover fine detail of spin Chern-Simons invariants --- such as the modular data over spin tori characterizing the \emph{topological order} of the FQH system --- and moreover predicts that much anticipated \emph{non-abelian defect} anyons are going to be localized where flux is expelled (on superconducting islands modeled by punctured surfaces).

\appendix
\section{\texorpdfstring{Some Cohesive Homotopy Theory}{Shape of Smooth Infinity-Groupoids}}
\label{Proofs}

For reference in the main text, here we briefly recall relevant background on the cohesive $\infty$-topos of smooth $\infty$-groupoids (\cite[\S 4.3]{SS25-Bun}\cite[\S 4.1]{SS26-Orb}, going back to \cite[\S 3.1]{SSS12}\cite{Sc13-dcct})
and of nonabelian differential cohomology (\cite[\S 9]{FSS23-Char}, exposition in \cite[\S 3]{SS25-Flux})
and spell out details of the proof of Prop. \ref{ShapeOfConstrainedPhaseSpace}.

This appendix assumes general familiarity with geometric homotopy theory ($\infty$-topos theory \cite{ToenVezzosi2005, Lurie2009}\cite[\S 1]{FSS23-Char}, see exposition for physicists in \cite{Schreiber2025}).

\subsection{Smooth $\infty$-Groupoids}

In discussing the phase space of flux quantized (higher) gauge fields \eqref{TheFluxQuantizedPhaseSpace},
we work in the $\infty$-topos of \emph{smooth $\infty$-groupoids}, 
$$
  \mathrm{SmthGrpd}_\infty
  :=
  \mathrm{Sh}\big(
    \mathrm{CartSp}
    ;
    \mathrm{Grpd}_\infty
  \big)
  \,,
$$
being that of $\infty$-stacks over the site $\mathrm{CartSp}$ of Cartesian spaces $\mathbb{R}^d$, $d \in \mathbb{N}$, with smooth maps between them and equipped with the coverage (Grothendieck pre-topology) of differentially good open covers \cite[\S A]{FSSt12-DiffClasses}.

The terminal geometric morphism of this $\infty$-topos admits a further left adjoint (and a further right adjoint, it is \emph{cohesive}), called the \emph{shape} operation \cite[\S 4.3]{SS25-Bun}\cite[\S 4.1]{SS26-Orb}:
\begin{equation}
  \begin{tikzcd}
    \mathrm{SmthGrpd}_\infty
    \ar[
      rr,
      shift left=14pt,
      "{ \mathrm{Shp} }"{description}
    ]
    \ar[
      rr,
      phantom,
      shift left=7pt,
      "{ \bot }"{scale=.6}
    ]
    \ar[
      rr,
      <-,
      "{ \mathrm{Dscr} }"{description}
    ]
    \ar[
      rr,
      phantom,
      shift right=7pt,
      "{ \bot }"{scale=.6}
    ]
    \ar[
      rr,
      shift right=14pt,
      "{ \mathrm{Pnts} }"{description}
    ]
    &&
    \mathrm{Grpd}_\infty
    \mathrlap{\,,}
  \end{tikzcd}
\end{equation}
with induced \emph{shape modality}
\begin{equation}
  \label{ShapeModality}
  \shape 
    := 
  \mathrm{Dscr} \circ \mathrm{Shp}
    :
  \begin{tikzcd}[sep=small]
    \mathrm{SmthGrpd}_\infty
    \ar[rr]
    &&
    \mathrm{SmthGrpd}_\infty
    \,,
  \end{tikzcd}
\end{equation}
and unit natural transformation
\begin{equation}
  \eta
    ^{\scalebox{.7}{$\shape$}}
    _{\mathbf{X}}
  :
  \begin{tikzcd}
    \mathbf{X}
    \ar[r]
    &
    \shape \mathbf{X}
  \end{tikzcd}
\end{equation}
satisfying the following properties, for $\mathbf{X}, \mathbf{Y}, \mathbf{B} \in \mathrm{SmthGrpd}_\infty$:
\begin{itemize}

\item idempotency:
\begin{equation}
  \label{IdempotencyOfShape}
  \shape\Big( \!\!
    \begin{tikzcd}
    \mathbf{X} 
    \ar[
      r,
      "{
        \eta
          ^{\scalebox{.7}{$\shape$}}
          _{\mathbf{X}}
      }"
    ]
    &
    \shape \mathbf{X}
    \end{tikzcd}
  \!\! \Big)
  \;\weakHomotopyEquivalence\;
  \Big(\!\!
  \begin{tikzcd}
    (\shape \mathbf{X})
    \ar[
      r,
      "{ \mathrm{id} }"
    ]
    &
    (\shape \mathbf{X})
  \end{tikzcd}
  \!\!\! \Big)
  \mathrlap{\,,}
\end{equation}

\item localization:
\begin{equation}
  \label{LocalizationPropertyOfShape}
  \begin{tikzcd}
    \mathrm{Map}\big(
      \mathbf{X}
      ,
      \shape 
      \mathbf{Y}
    \big)
    \simeq
    \mathrm{Map}\big(
      \shape \mathbf{X}
      ,
      \shape 
      \mathbf{Y}
    \big)
    \mathrlap{\,,}
  \end{tikzcd}
\end{equation}

\item preservation of homotopy fiber products over shapes \cite[Prop. 4.3.8]{SS25-Bun}:
\begin{equation}
  \label{ShapeOfFiberProductsOverShapes}
  \shape\Big(
    \mathbf{X} 
    \underset{
      \scalebox{.7}{$\shape \mathbf{B}$}
    }{\times}
    \mathbf{Y}
  \Big)
  \weakHomotopyEquivalence
   \big( \shape \mathbf{X} \big)
    \underset{
      \scalebox{.7}{$\shape \mathbf{B}$}
    }{\times}
    \big( \shape \mathbf{Y} \big)
    \mathrlap{\,,}
\end{equation}
\item smooth Oka principle \cite{BerwickEvans2024}\cite[\S 1.1.2]{SS25-Bun}:
\begin{equation}
  \label{SmoothOkaPrinciple}
  X \in \begin{tikzcd}[sep=small]\mathrm{SmthMfd} \ar[r, hook] & \mathrm{SmthGrpd}_\infty
  \;\;
  \Rightarrow
  \;\;
  \shape
  \mathrm{Map}\big(
    X,
    \mathbf{Y}
  \big)
  \simeq
  \mathrm{Map}\big(
    \shape X
    ,
    \shape \mathbf{Y}
  \big).
  \end{tikzcd}
\end{equation}

\end{itemize}

For example, the category of (``convenient'' Delta-generated \cite[\S 4.3.4]{SS25-Bun}) topological spaces embeds fully faithfully into the 0-truncated smooth $\infty$-groupoids (the \emph{smooth sets} \cite{GS25-FieldsI}\cite{nLab:SmoothSet})
\begin{equation}
  \label{SmoothSet}
  \begin{tikzcd}
    \mathrm{DTopSp}
    \ar[r, hook]
    &
    \mathrm{SmthGrpd}_0
    \ar[r, hook]
    &
    \mathrm{SmthGrpd}_\infty
  \end{tikzcd}
\end{equation}
and for $X \in \mathrm{DTopSp}$ its shape $\shape X$ is its underlying homotopy type, equivalent to the traditional singular simplicial complex (hence is its \emph{shape} in the sense of \emph{shape theory}, whence the name).
Generally, the shape operation may be understood as forming underlying homotopy types of \emph{topological realizations} of smooth $\infty$-groupoids.

\subsection{Nonabelian Differential Cohomology}

Another important example for the shape \eqref{ShapeModality} of smooth $\infty$-groupoids arises for a real $L_\infty$-algebra $\mathfrak{a}$ with Chevalley-Eilenberg algebra $\mathrm{CE}(\mathfrak{a})$ \cite[\S 4]{FSS23-Char}, where the system of closed 
\footnote{
  The $L_\infty$-algebra valued differential forms \eqref{SheafOfFlatLInfinityValuedForms} being \emph{closed} means that they are \emph{flat} in that they are \emph{Maurer-Cartan elements} in $\Omega^\bullet_{\mathrm{dR}}\big(\mathbb{R}^d\big) \otimes \mathfrak{a}$. However, the latter terminology tends to suggest a context where these differential forms play the role of local higher \emph{gauge potentials}/\emph{connections} --- whereas in our context these forms are instead \emph{field strengths}/\emph{flux densities} and their flatness/MC-property is really the latter's \emph{Bianchi identity} (cf. \cite[(261-2)]{SatiSchreiberStasheff2009}), which generalizes the plain closure of abelian flux densities. This shift in perspective on the \emph{role} of $L_\infty$-algebras in higher gauge theory (from coefficients of gauge potentials to coefficients of their field strengths/fluxes) is crucial but may still appear unusual. In order for the terminology to be suggestive of this situation, we speak in \eqref{SheafOfFlatLInfinityValuedForms} of \emph{closed} instead of \emph{flat} $L_\infty$-valued forms, even though mathematically it means the same.
}
$\mathfrak{a}$-valued differential forms
\cite[\S 6]{FSS23-Char}
over Cartesian spaces is a sheaf and hence a smooth set \eqref{SmoothSet}:
\begin{equation}
  \label{SheafOfFlatLInfinityValuedForms}
  \begin{aligned}
  &
  \Omega^1_{\mathrm{dR}}(
    -;
    \mathfrak{a}
  )_{\mathrm{cl}}
  \in
  \mathrm{SmthGrpd}_0
  \hookrightarrow
  \mathrm{SmthGrpd}_\infty
  \,,
  \\
  &
  \mathbb{R}^d
  \longmapsto
  \Omega^1_{\mathrm{dR}}(
    \mathbb{R}^d
    ;
    \mathfrak{a}
  )_{\mathrm{cl}}
  :=
  \mathrm{Hom}_{\mathrm{dgcAlg}_{\mathbb{R}}}\big(
    \mathrm{CE}(\mathfrak{a})
    ,
    \Omega^\bullet_{\mathrm{dR}}(
      \mathbb{R}^d
    )
  \big)
  \mathrlap{\,.}
  \end{aligned}
\end{equation}
When $\mathfrak{a} \simeq \mathfrak{l}\hotype{A}$ is the real \emph{Whitehead $L_\infty$-algebra} of the shape $\hotype{A}$ of a topological space (nilpotent and of rational finite type), hence when $\mathrm{CE}(\mathfrak{a})$ is the minimal \emph{Sullivan model} for $\hotype{A}$ \cite[\S 5]{FSS23-Char}, then the shape of \eqref{SheafOfFlatLInfinityValuedForms} is the homotopy type of the $\mathbb{R}$-\emph{rationalization} $L^{\mathbb{R}}\hotype{A}$ of that space \cite[Lem 9.1]{FSS23-Char}:
\begin{equation}
  \shape
  \,
  \Omega^1_{\mathrm{dR}}\big(
    -;
    \mathfrak{l}\hotype{A}
  \big)_{\mathrm{cl}}
  \weakHomotopyEquivalence
  L^{\mathbb{R}}\hotype{A}
  \,,
\end{equation}
and the \emph{rationalization unit} in this guise is the avatar $\mathbf{ch}_{\hotype{A}}$ of the \emph{character map} $\mathrm{ch}_{\hotype{A}}$ \cite[\S 9]{FSS23-Char} from generalized non-abelian cohomology $H^1\big(-;\Omega \hotype{A}\big)$ \cite[\S 2]{FSS23-Char}\cite[\S 1]{SS25-TEC}
to $\mathfrak{l}\hotype{A}$-valued nonabelian de Rham cohomology \cite[\S 6]{FSS23-Char}:  
\begin{equation}
  \label{TheNonabelianCharacterMap}
  \begin{tikzcd}[
    row sep=7pt
  ]
    H^1\big(
      -;
      \Omega\hotype{A}
    \big)
    \ar[
      d,
      phantom,
      "{ := }"{rotate=-90}
    ]
    \ar[
      r,
      "{ \mathrm{ch}_{\hotype{A}} }"
    ]
    &
    H^1_{\mathrm{dR}}\big(
      -;
      \mathfrak{l}\hotype{A}
    \big)
    \ar[
      d,
      phantom,
      "{ := }"{rotate=-90}
    ]
    \\
    \pi_0
    \,
    \mathrm{Map}\big(
      -,
      \hotype{A}
    \big)
    \ar[
      r,
      "{
        (\mathbf{ch}_{\hotype{A}})_\ast
      }"
    ]
    &
    \pi_0
    \,
    \mathrm{Map}\Big(
      -,
      \shape
      \, 
      \Omega^1_{\mathrm{dR}}\big(
        -;
        \mathfrak{l}\hotype{A}
      \big)_{\mathrm{cl}}
    \Big)
    \mathrlap{\,.}
  \end{tikzcd}
\end{equation}
This allows to form the smooth $\infty$-groupoid $\hotype{A}_{\mathrm{diff}}$ which classifies closed $\mathfrak{l}\hotype{A}$-valued  differential form data, whose cocycle image in de Rham/rational $\Omega\hotype{A}$-cohomology is equipped with a coherent lift to a cocycle in actual $\Omega \hotype{A}$-cohomology, hence the following homotopy pullback \cite[Def. 9.3]{FSS23-Char}:
\begin{equation}
  \label{DiffCohomologyPullback}
  \begin{tikzcd}
    \hotype{A}_{\mathrm{diff}}
    \ar[
      rr,
      "{ \vec \chi }"
    ]
    \ar[
      d,
      "{ \vec F }"{swap}
    ]
    \ar[
      dr,
      phantom,
      "{ \lrcorner }"{pos=.1}
    ]
    &&
    \hotype{A}
    \ar[
      d,
      "{
        \mathbf{ch}_{\hotype{A}}
      }"
    ]
    \\
    \Omega^1_{\mathrm{dR}}\big(
      -;
      \mathfrak{l}\hotype{A}
    \big)_{\mathrm{cl}}
    \ar[
      rr,
      "{
        \eta^{\scalebox{.6}{$\shape$}}
      }"
    ]
    &{}&
    \shape
    \Omega^1_{\mathrm{dR}}\big(
      -;
      \mathfrak{l}\hotype{A}
    \big)_{\mathrm{cl}}
    \mathrlap{\,.}
  \end{tikzcd}
\end{equation}

If here we think of $\mathfrak{l}\hotype{A}$ as encoding the Gauss laws/Bianchi identities of flux densities on phase spaces of some higher gauge field theory, then $\hotype{A}$ encodes a compatible flux quantization law and 
\begin{equation}
  \label{PhaseSpaceInAppendix}
  \mathrm{PhsSp}
  :=
  \mathrm{Map}\big(X^d,\hotype{A}_{\mathrm{diff}}\big)
  \in
  \mathrm{SmthGrpd}_\infty
\end{equation}
is the \emph{flux quantized phase space} of the higher gauge theory over any given Cauchy surface $X^d$ \cite{SS24-Phase}\cite[\S 3.3]{SS25-Flux}.

Generally, the geometrically connected components of this smooth $\infty$-groupoid is the \emph{differential $\Omega\hotype{A}$-cohomology} of $X^d$ \cite[\S 9]{FSS23-Char}:
\begin{equation}
  H^1_{\mathrm{diff}}\big(
    X^d;
    \Omega\hotype{A}
  \big)
  :=
  \pi_0
  \shape
  \mathrm{Map}\big(X^d,\hotype{A}_{\mathrm{diff}}\big)
  \,.
\end{equation}

More generally, when the domain manifold arises as a one-point compactification $X^d_{\cpt}$
equipped with a contractible open neighbourhood of the point at infinity
\begin{equation}
  \label{NeighbourhoodOfInfinity}
  O_\infty
  :
  \begin{tikzcd}
    D^d
    \ar[
      r,
      hook,
    ]
    &
    X^d_{\cpt}
    \mathrlap{\,,}
  \end{tikzcd}
\end{equation}
then passage to homotopy fibers of restriction to the neighbourhood of infinity $O_\infty$ gives the spaces of fluxes/charges which are \emph{solitonic}, denoted $(-)^{\mathrm{solit}}$, in that they \emph{vanish at infinity}:
\begin{equation}
  \label{FluxesVanishingAtInfinity}
  \begin{tikzcd}
    \mathbf{\Omega}^1_{\mathrm{dR}}\big(
      X^d_{\cpt}
      ;
      \mathfrak{l}\hotype{A}
    \big)
      _{\mathrm{cl}}
      ^{ \mathrm{solit} }
    \ar[
      d, 
      "{
        \mathrm{fib}
      }"
    ]
    \ar[
      r,
      "{
        (
        \eta
          ^{\scalebox{.6}{$\shape$}}
        )_\ast
      }"
    ]
    &
    \shape
    \,
    \mathbf{\Omega}^1_{\mathrm{dR}}\big(
      X^d_{\cpt}
      ;
      \mathfrak{l}\hotype{A}
    \big)
      _{\mathrm{cl}}
      ^{ \mathrm{solit} }
    \ar[
      d, 
      "{
        \shape \mathrm{fib}
      }"
    ]
    \ar[
      r,
      <-,
      "{
        (\mathbf{ch}_{\hotype{A}})_\ast
      }"
    ]
    &
    \mathrm{Map}^\ast\big(
      X^d_{\cpt}
      ,
      \hotype{A}
    \big)
    \ar[
      d, 
      "{
        \mathrm{hofib}
      }"
    ]
    \\
    \mathbf{\Omega}^1_{\mathrm{dR}}\big(
      X^d_{\cpt}
      ;
      \mathfrak{l}\hotype{A}
    \big)_{\mathrm{cl}}
    \ar[
      d, 
      "{
        (O_\infty)^\ast
      }"
    ]
    \ar[
      r,
      "{
        (
        \eta
          ^{\scalebox{.6}{$\shape$}}
        )_\ast
      }"
    ]
    &
    \shape
    \,
    \mathbf{\Omega}^1_{\mathrm{dR}}\big(
      X^d_{\cpt}
      ;
      \mathfrak{l}\hotype{A}
    \big)_{\mathrm{cl}}
    \ar[
      d, 
      "{
        \shape (O_\infty)^\ast
      }"
    ]
    \ar[
      r,
      <-,
      "{
        (\mathbf{ch}_{\hotype{A}})_\ast
      }"
    ]
    &
    \mathrm{Map}\big(
      X^d_{\cpt}
      ,
      \hotype{A}
    \big)
    \ar[
      d, 
      "{
        \mathrm{ev}
      }"
    ]
    \\
    \mathbf{\Omega}^1_{\mathrm{dR}}\big(
      D^d
      ;
      \mathfrak{l}\hotype{A}
    \big)_{\mathrm{cl}}
    \ar[
      r,
      "{
        (
        \eta
          ^{\scalebox{.6}{$\shape$}}
        )_\ast
      }"
    ]
    &
    \shape
    \,
    \mathbf{\Omega}^1_{\mathrm{dR}}\big(
      D^d
      ;
      \mathfrak{l}\hotype{A}
    \big)_{\mathrm{cl}}
    \ar[
      r,
      <-,
      "{
        (\mathbf{ch}_{\hotype{A}})_\ast
      }"
    ]
    &
    \mathrm{Map}\big(
      \ast
      ,
      \hotype{A}
    \big)
    \mathrlap{\,,}
  \end{tikzcd}
\end{equation}
and the \emph{flux quantized solitonic phase space} is the corresponding homotopy fiber product
\begin{equation}
  \label{TheSolitonicPhaseSpace}
  \begin{tikzcd}
    \mathrm{PhsSp}^{\mathrm{solit}}
    \ar[
      rr,
      "{
        \chi
      }"
    ]
    \ar[
      d,
      "{
        \vec F
      }"{swap}
    ]
    \ar[
      dr,
      phantom,
      "{ \lrcorner }"{pos=.1,}
    ]
    &&
    \mathrm{Map}^\ast\big(
      X^d_{\cpt}
      ,
      \hotype{A}
    \big)
    \ar[
      d,
      "{
        (\mathbf{ch}_{\hotype{A}})_\ast
      }"
    ]
    \\
    \mathbf{\Omega}^1_{\mathrm{dR}}\big(
      X^d_{\cpt}
      ;
      \mathfrak{l}\hotype{A}
    \big)
      _{\mathrm{cl}}
      ^{\mathrm{solit}}
    \ar[
      rr,
      "{
        (\eta^{\scalebox{.6}{$\shape$}})_\ast
      }"
    ]
    &{}&
    \shape
    \,
    \mathbf{\Omega}^1_{\mathrm{dR}}\big(
      X^d_{\cpt}
      ;
      \mathfrak{l}\hotype{A}
    \big)
      _{\mathrm{cl}}
      ^{\mathrm{solit}}
    \mathrlap{\,.}
  \end{tikzcd}
\end{equation}

\subsection{Proof of Prop. \ref{ShapeOfConstrainedPhaseSpace}}
\label{ProofOfShapeOfConstrainedPhaseSpace}

We decompose the proof of Prop. \ref{ShapeOfConstrainedPhaseSpace} into two Lemmas, \ref{IfConstrainedFluxIsEquivalenceThenSoIsConstrainedPhaseSpace} and \ref{ProvingShapeOfConstrainedFluxes} below, the first being a general statement about constrained flux quantized phase spaces, and the second being specific to the situation of 2-cohomotopical flux quantization on the plane.

Consider therefore first the homotopy pullback defining a general constrained phase space \eqref{ConstrainedPhaseSpace} according to the left part of the following diagram,
where we show it ``pasted'' on its right to the image under $\mathrm{Map}(X,-)$ of the homotopy pullback \eqref{DiffCohomologyPullback} defining $\hotype{A}_{\mathrm{diff}}$:
\vspace{-3mm} 
\begin{equation}
  \begin{tikzcd}[column sep=16pt]
    \mathrm{PhsSp}
      ^{\mathrm{constr}}
    \ar[d]
    \ar[r]
    \ar[
      dr,
      phantom,
      "{ \lrcorner }"{pos=.1}
    ]
    &
    \grayoverbrace{
    \mathrm{Map}\big(
      X,
      \hotype{A}_{\mathrm{diff}}
    \big)
    }{
      \mathrm{PhsSp}
    }
    \ar[
      d,
      "{
        \vec F_\ast
      }"{swap}
    ]
    \ar[
      r,
      "{
        \chi_\ast
      }"
    ]
    \ar[
      dr,
      phantom,
      "{ \lrcorner }"{pos=.1}
    ]
    &[+5pt]
    \mathrm{Map}\big(
      X,
      \hotype{A}
    \big)
    \ar[
      d,
      "{
        (\mathbf{ch}_{\hotype{A}})_\ast
      }"
    ]
    \\
    \mathbf{\Omega}^1_{\mathrm{dR}}\big(
      X;
      \mathfrak{a}
    \big)
      _{\mathrm{cl}}
      ^{\mathrm{constr}}
    \ar[
      r,
      hook,
      "{ \iota }"
    ]
    &
    \grayunderbrace{
    \mathrm{Map}\Big(
      X,
      \Omega^1_{\mathrm{dR}}\big(
        -;
        \mathfrak{a}
      \big)
        _{\mathrm{cl}}
    \Big)
    }{
      \mathbf{\Omega}^1_{\mathrm{dR}}
      (
        X,
        \mathfrak{a}
      )_{\mathrm{cl}}
    }
    \ar[
      r,
      "{
        (\eta^{\scalebox{.7}{$\shape$}})_\ast
      }"
    ]
    &
    \mathrm{Map}\Big(
      X,
      \shape\,
      \Omega^1_{\mathrm{dR}}\big(
        -;
        \mathfrak{a}
      \big)
        _{\mathrm{cl}}
    \Big)
    \mathrlap{\,.}
  \end{tikzcd}
\end{equation}
Now, since $\mathrm{Map}(X,-)$ preserves homotopy limits, the right square is also a homotopy pullback, and hence so is the total rectangle, by the \emph{pasting law} \cite[(1.14)]{SS25-Bun}. Observe then that the bottom right object is pure shape, by the localization property \eqref{LocalizationPropertyOfShape} and the smooth Oka principle \eqref{SmoothOkaPrinciple}. Therefore passage to shapes \eqref{ShapeModality} preserves the pullback property both of the right square as well as of the total rectangle, by \eqref{ShapeOfFiberProductsOverShapes}, whence the pasting law applies again to show that also the following left square is a homotopy pullback:
\begin{equation}
  \label{ShapeOfConstrainedPhaseSpaceAsPullbackOfConstrainedFluxes}
  \begin{tikzcd}
    \shape 
    \,
    \mathrm{PhsSp}^{\mathrm{constr}}
    \ar[r]
    \ar[d]
    \ar[
      dr,
      phantom,
      "{ \lrcorner }"{pos=.1}
    ]
    &
    \shape\, \mathrm{PhsSp}
    \ar[d]
    \\
    \shape\,
    \mathbf{\Omega}^1_{\mathrm{dR}}\big(
      X;
      \mathfrak{a}
    \big)
      _{\mathrm{cl}}
      ^{\mathrm{constr}}
    \ar[
      r,
      "{ \shape \iota }"
    ]
    &
    \shape\,
    \mathbf{\Omega}^1_{\mathrm{dR}}\big(
      X;
      \mathfrak{a}
    \big)
      _{\mathrm{cl}}
  \end{tikzcd}
\end{equation}
From this, we obtain:
\begin{lemma}
  \label{IfConstrainedFluxIsEquivalenceThenSoIsConstrainedPhaseSpace}
  A sufficient condition
  for the constrained phase space to have the same connected components as the unconstrained phase space is:
  \begin{equation}
    \left.
    \begin{aligned}
      & 
      \underset{n \geq 1}
       {\forall}
      \;\;
      \pi_n
      \shape
      \mathbf{\Omega}^1_{\mathrm{dR}}\big(
        X ;
        \mathfrak{a}
      \big)_{\mathrm{cl}}
      = 1
      \\
      & 
      \mbox{\rm and}
      \;\;
      \pi_0 \shape \iota
      =
      \mathrm{id}
    \end{aligned}
    \right\}
    \Rightarrow
    \Big(
        \pi_0
        \shape 
        \mathrm{PhsSp}^{\mathrm{constr}}
      \simeq
        \pi_0
        \shape \mathrm{PhsSp}
    \Big).
  \end{equation}
\end{lemma}

Thereby, the proof of Prop. \ref{ShapeOfConstrainedPhaseSpace} is reduced to:

\begin{lemma}
  \label{ProvingShapeOfConstrainedFluxes}
  The comparison map 
  \eqref{ShapeComparisonOfConstrained5DMCSFluxes}  
  \begin{equation}
    \label{ComparisonMapInAppendix}
    \begin{tikzcd}
      \shape
      \mathbf{\Omega}^1_{\mathrm{dR}}\big(
        \mathbb{R}^3_{\cpt} 
          \wedge 
        [0,1]_{\plus}
        ;
        \mathfrak{l}S^2
      \big)
        _{\mathrm{cl}}
        ^{\mathrm{constr}}
      \ar[
        r,
        "{ \shape \iota }"
      ]
      &
      \shape
      \mathbf{\Omega}^1_{\mathrm{dR}}\big(
        \mathbb{R}^3_{\cpt} 
        ;
        \mathfrak{l}S^2
      \big)
        _{\mathrm{cl}}
        ^{\mathrm{solit}}
    \end{tikzcd}
  \end{equation}  
  satisfies the condition of Lem. \ref{IfConstrainedFluxIsEquivalenceThenSoIsConstrainedPhaseSpace}.
\end{lemma}
\begin{proof}
  First, it is immediate that the right hand side is equivalent to the homotopy type of the pointed mapping space of $S^3$ into the $\mathbb{R}$-rationalized 2-sphere,
  \begin{equation}
    \label{IncarnationsOfRightHandSide}
      \shape
      \mathbf{\Omega}^1_{\mathrm{dR}}\big(
        \mathbb{R}^3_{\cpt}
        ;
        \mathfrak{l}S^2
      \big)
        _{\mathrm{cl}}
        ^{\mathrm{solit}}
    \weakHomotopyEquivalence
    \shape \, \mathrm{Map}^\ast\big(S^3, L^{\mathbb{R}} S^2\big)
    \,,
  \end{equation} 
  whence its homotopy is indeed concentrated in degree 0, where it is the set of real numbers, representing real multiples of the Hopf generator: 
  \begin{equation}
    \begin{aligned}
    \pi_n \, \shape \, 
    \mathrm{Map}^\ast\big(
      \mathbb{R}^3_{\cpt}
      ,
      L^{\mathbb{R}} S^2
    \big)
    &
    \simeq
    \pi_{3+n}\big(
      L^{\mathbb{R}}S^2
    \big)
    \\
    &
    \simeq
      \pi_{3+n}(S^2) 
        \otimes_{\mathbb{Z}}
      \mathbb{R}
    \simeq
    \left\{
    \begin{aligned}
      \mathbb{R}
      &
      \mbox{ if }
      n = 0
      \\
      0
      &
      \mbox{ if }
      n \geq 1 
      \mathrlap{\,,}
    \end{aligned}
    \right.
    \end{aligned}
  \end{equation}
  hence
  \begin{equation}
      \shape
      \mathbf{\Omega}^1_{\mathrm{dR}}\big(
        \mathbb{R}^3_{\cpt}
        ;
        \mathfrak{l}S^2
      \big)
        _{\mathrm{cl}}
        ^{\mathrm{solit}}
        \;
     \weakHomotopyEquivalence
     \;
     \mathrm{Dsc}
     (\mathbb{R})
     \mathrlap{\,.}
  \end{equation}
  Here $r \in \mathbb{R} \simeq \pi_3(S^2) \otimes_{\mathbb{Z}} \mathbb{R}$ is represented by any 
  $$
    (F_2, H_3)
    \in
    \Omega^1_{\mathrm{dR}}\big(
      \mathbb{R}^3_{\cpt}
      ;
      \mathfrak{l}S^2
    \big)
     _{\mathrm{cl}}
    \;\;
    \mbox{with}
    \;\;
    \int_{\mathbb{R}^3_{\cpt}}
  \!\!  H_3 
    \;=\;
    0
    \mathrlap{\,.}
  $$

  Next, to recall the constraint object on the left of \eqref{ComparisonMapInAppendix}, in the notation of \eqref{DimRedSpacetime}, we are looking at the decomposition
  \begin{equation}
    X^4_{\cpt}
    \simeq
    \grayunderbrace{
    \mathbb{R}^3_{\cpt}
    }{ X^3 }
    \wedge
    \grayunderbrace{
    {[0,1]}_{\plus}
    }{ V }
    \mathrlap{\,.}
  \end{equation}
  The comparison map \eqref{ComparisonMapInAppendix} is thereby equivalently given by pullback of differential forms along the inclusion
  \begin{equation}
    \begin{tikzcd}
      \mathbb{R}^3_{\cpt}
      \ar[
        rrr,
        hook,
        "{
          \big(
            \mathrm{id}_{\mathbb{R}^3_{\cpt}}
            ,\,
            0
          \big)
        }"
      ]
      &&&
      \mathbb{R}^3_{\cpt} 
      \times [0,1]\,,
    \end{tikzcd}
  \end{equation} 
  and the constrained flux data are
  $$
    (F_2, H_3)
    =
    \big(
      B_2 
      ,
      \star_3 E_0
    \big)
  $$  
  This makes it immediate that  $\pi_0 \shape \iota$ is bijective and hence completes the proof.
\end{proof}


\printbibliography

\end{document}